\RequirePackage{fix-cm}
\documentclass[numbook, envcountsect, envcountsame, envcountreset, runningheads, smallextended]{svjour3}
\smartqed  
\usepackage{graphicx}

\usepackage[misc]{ifsym}

\usepackage{amsfonts}
\usepackage{moreverb}
\usepackage{bm}
\usepackage{mathtools}
\usepackage{enumitem}
\usepackage{diagbox} 

\usepackage{algorithm}  
\usepackage{algcompatible}
\usepackage{subfigure}

\usepackage{xcolor}
\usepackage{soul}

\newcommand\BibTeX{{\rmfamily B\kern-.05em \textsc{i\kern-.025em b}\kern-.08em
T\kern-.1667em\lower.7ex\hbox{E}\kern-.125emX}}

\journalname{Finance Stoch}
\begin{document}

\title{Dynamic Mean-Variance Asset Allocation in General Incomplete Markets}
\subtitle{A Nonlocal BSDE-based Feedback Control Approach}

\titlerunning{Dynamic MV Asset Allocation in General Incomplete Markets}

\author{Qian Lei         \and
        Chi Seng Pun    \and
        Jingxiang Tang
}


\institute{Qian Lei \at
              School of Physical and Mathematical Sciences, Nanyang Technological University, Singapore \\
              \email{qian.lei@ntu.edu.sg}
           \and
           Chi Seng Pun (\textrm{\Letter}) \at
              School of Physical and Mathematical Sciences, Nanyang Technological University, Singapore \\
              \email{cspun@ntu.edu.sg}
           \and
           Jingxiang Tang \at
              School of Physical and Mathematical Sciences, Nanyang Technological University, Singapore \\
              \email{tang0506@e.ntu.edu.sg}
}

\date{ }

\maketitle

\begin{abstract}
This paper studies dynamic mean-variance (MV) asset allocation problems in general incomplete markets. Besides of the conventional MV objective on portfolio's terminal wealth, our framework can accommodate running MV objectives with general (non-exponential) discounting factors while in general, any time-dependent preferences. We attempt the problem with a game-theoretic framework while decompose the equilibrium control policies into two parts: the first part is a myopic strategy characterized by a linear Volterra integral equation of the second kind and the second part reveals the hedging demand governed by a system of nonlocal backward stochastic differential equations. We manage to establish the well-posedness of the solutions to the two aforementioned equations in tailored Bananch spaces by the fixed-point theorem. It allows us to devise a numerical scheme for solving for the equilibrium control policy with guarantee and to conclude that the dynamic (equilibrium) mean-variance policy in general settings is well-defined. Our probabilistic approach allows us to consider a board range of stochastic factor models, such as the Chan--Karolyi--Longstaff--Sanders (CKLS) model. For which, we verify all technical assumptions and provide a sound numerical scheme. Numerical examples are provided to illustrate our framework.


\end{abstract}

\section{Introduction} \label{sec:intro}


Markowitz's mean-variance (MV) analysis, introduced by \cite{Markowitz1952}, pioneered the field of modern portfolio theory by providing a quantitative framework for portfolio optimization. The central idea is to construct an efficient portfolio by striking the balance between its reward and risk that are typically measured by expected return and variance (or standard deviation), respectively.
In its simplest form, MV analysis was first studied within a single-period framework. As financial markets and investment opportunities grew in complexity, the need for a multi-period even continuous-time approach to MV analysis became apparent. However, the MV criteria induce time inconsistency (TIC) in searching for optimal dynamic controls, which violates Bellman's principle of optimality (BPO) and is thus not in favor of dynamic programming; see \cite{Basak2010} for the enlightenment.
Despite the inherent challenges of dynamic MV analysis, \cite{Li2000OptDynMV} and \cite{leippold2004geometric} addressed the issues within a discrete-time framework while \cite{zhou2000continuous} and \cite{lim2002} explored the continuous-time counterparts with embedding and geometric approaches.
However, their solutions are pre-commitment investment policies fixed at an initial time, where the investor may find it optimal to consider an alternative policy as time evolves.
A critical ingredient of rational decision-making is time consistency, a principle articulated by Strotz in his pioneering work \cite{Strotz1955} on time-consistent planning. He asserted that an investor should choose ``\textit{the best plan among those that he will actually follow}." Note that the time-consistent planning was aligned with BPO when TIC incentives are non-existing. Beyond the classical MV framework, researchers have examined the TIC stochastic control problem from a game-theoretic perspective, which introduced subgame perfect equilibrium (SPE) policies in line with \cite{Strotz1955}'s time-consistent planning. 


The TIC preferences are common in behavioral economics, while the desire of the SPE policies align with the approaches in the related literature on the reference-dependent models that lead to changing preferences and tastes; see \cite{Peleg1973,Harris2001}. While the BPO is violated with TIC incentives, the SPE formulation essentially aims to revive the BPO for the control sequence by introducing adjustment terms to offset the TIC sources or constraining the future control choices. \cite{Bjoerk2014,Bjoerk2017} build an unified analytical framework for addressing general TIC problems by treating a TIC control problem as an intrapersonal game with players indexed by time carrying different objectives. The Nash equilibrium of the intrapersonal game, i.e., SPE, yields a time-consistent policy for the agent. Further, \cite{Bjoerk2017} derive a coupled system of Hamilton–Jacobi–Bellman (HJB) equations to characterize the equilibrium solution, including the SPE policy and the equilibrium value function. Although the framework of \cite{Bjoerk2017} nests the MV criterion and gives the same equations to be solved in \cite{Basak2010}, the latter has claimed more fruitful mathematical properties of the equilibrium solution while \cite{Bjoerk2017} has still remained some open problems about the well-posedness (see its last discussion section). This paper is focused on the dynamic MV asset allocation in incomplete markets while we use a probabilistic approach that allow us to consider a broader range of stochastic factor models and assert the well-posedness of the equilibrium solution.

\cite{Basak2010} indeed attempted the problem with stochastic volatility (SV) models (incomplete markets), but their characterization of the equilibrium solution and investigation of its existence and uniqueness require the model parameters to be sufficiently well-behaved. Specifically, explicit expressions of moments of various orders need to be available and sufficiently smooth to satisfy the analytical claims in their study. It largely limits the model choices for the SV or stochastic factors. As a result, (the only) existing studies on both TIC problem and incomplete markets (\cite{Basak2010,Dai2021}) were confined to the Cox--Ingersoll--Ross (CIR) model and the Ornstein--Uhlenbeck (OU) process for the SV that possess smooth and explicit density functions. While the CIR and OU models provide scarce analytical tractability, their specific assumptions may not always align with real-world scenarios.

In this paper, we adopt a probabilistic approach to redefine equilibrium portfolio choice, demonstrating its well-defined nature through the theory of backward stochastic differential equations (BSDEs). Consequently, our theoretical results, when applied to SV models in incomplete markets, significantly alleviate the constraints imposed on model parameters as observed in the existing literature. In contrast to previous approaches, we no longer necessitate density functions or explicit calculations of moments. Moreover, we only require the moments to be continuous, eliminating the need for differentiability. This advance enables our analytic framework to explore dynamic MV problems across a much broader range of SV models driven by the Chan–Karolyi–Longstaff–Sanders (CKLS) process, whose adjustable parameter provides a more adaptable fit to the observed stochastic factors. For the CKLS model in its most general form, parameterized by $p\in[0,1]$, neither the Fourier transform nor the density is known except for some special cases of the CIR ($p=\frac{1}{2}$) and OU ($p=0$) processes and thus the existing approaches are infeasible. Hence, our MV analysis with CKLS-based SV models enhances modeling flexibility and applicability.

It should be noted that under the considerations of open-loop controls, \cite{Hu2012,Hu2017} have shown the well-posedness of the equilibrium solution to TIC linear-quadratic control problems. \cite{Yan2019} subsequently adapts the open-loop control framework to MV problem for some SV models with well-posedness results. Our formulation in line with (extended) dynamic programming, however, belongs to the closed-loop (feedback) control framework, whose well-posedness is much more involved. We note some existing studies, such as \cite{Huang2017,Han2021}, consider closed-loop controls but restrictions could be found in their form presumed and the way they are perturbed. We get rid of these constraints as our optimal controls are first expressed in terms of the unknown (value) function, whose underlying BSDEs are studied directly. Our attempt is more aligned with the works of \cite{Basak2010,Bjoerk2014,Bjoerk2017,Dai2021} (solving for global optimality from the HJB equations) rather than the open- or closed-loop control frameworks of \cite{Hu2012,Huang2017,Yan2019,Han2021} (solving for local optimality with adjoint processes). Our primary goal is to extend the works of \cite{Basak2010,Dai2021} to more general SV models by leveraging a nonlocal BSDE approach. Noting the connection between BSDEs and parabolic partial differential equations (PDEs), we also refer the readers to some recent works on the well-posedness of the solutions to nonlocal PDEs (or equilibrium HJB equations), e.g., \cite{Wei2017,Lei2023,Lei2024}. The results of the aforementioned works cannot be directly applied in our case, though they are mathematically inspirational.

The main contributions of this paper is threefold. First, we extend the MV analysis in continuous time to incorporate with running MV objectives (on top of the MV criterion on the terminal portfolio wealth) and a broader class of stochastic factor models, while the risk aversion coefficients in both running and terminal MV objectives could be time-varying. Second, in such a setting, we manage to characterize the time-consistent MV investment policy by the summation of a myopic term and a hedging demand term, where these two terms are represented via the solutions to a linear Volterra integral equation of the second kind and a flow of BSDEs (or a nonlocal BSDE), respectively. Our framework possesses two major merits: (I) the equilibrium policy could be well-defined without explicit expressions and differentiability of the moments of the stochastic factor models (which are currently desired in the literature); (II) the generator of the nonlocal BSDE is generalized by relaxing the conventional Lipschitz condition to its stochastic counterpart. Third, for the aforementioned two equations, we prove the well-posedness of their solutions with well-designed Banach spaces and contractive mappings, and subsequently we develop a numerical scheme for solving them. 

The rest of this paper is organized as follows. Section \ref{sec:tic} first introduces the fundamental game-theoretic concepts of TIC stochastic control problems needed to analyze the dynamic MV asset allocation under our consideration, including equilibrium policies, equilibrium value functions, and the extended HJB system. Next, to better examine the MV equilibrium portfolio strategy, we initially limit our study to a complete market in Section \ref{sec:mvcm}, in which we derive an explicit representation for the myopic component of the portfolio choice. Moving to an incomplete market setting, Section \ref{sec:mvicm} proposes a probabilistic formulation of the equilibrium investment policy, providing a fully analytical yet simple characterization of the dynamic MV portfolios in a general incomplete-market economy. In Section \ref{subsec:wellposedness}, we establish the existence and uniqueness of the solutions to a new type of nonlocal BSDE system, providing soild evidence for the well-defined-ness of the equilibrium MV policy. The proofs of the results in Section \ref{sec:mv} (main theoretical results) are postponed to the appendix. The theoretical results are further applied to different stochastic investment opportunity sets offered by a broad range of SV models in Section \ref{subsec:apps}. Section \ref{sec:Num} presents a numerical scheme for solving those well-posed equations and illustrate the results with a few examples. Finally, we conclude this paper and further discuss some potential extensions, including MV problems with state-dependent risk aversion in complete and incomplete markets.

\section{Game-Theoretic Formulation for TIC Stochastic Control Problems} \label{sec:tic}
In this section, we first introduce the dynamic MV framework under our consideration. For the simplicity of notational use and the generality, we instead formulate a relatively general TIC stochastic control problems, which extends the framework of \cite{Bjoerk2017} in some aspects while nests all our MV analyses. Then, we define equilibrium policy and equilibrium value function, akin to that of \cite{Bjoerk2014,Bjoerk2017}, for the general problems. Subsequently, we introduce the equilibrum HJB equation that drive the policy and value function. Note that we postpone the specification of the state dynamics to latter context as it would generally not affect our discussion in this section. Though, it is good to know in advance that the controlled state dynamics we consider is of Markovian.

We first introduce the Markowitz MV operator as 
\begin{equation*}
	\mathbb{MV}^{\mathcal{F}_s}[\cdot]:=\rho(s)\mathbb{E}^{\mathcal{F}_s}[\cdot]-\frac{\gamma}{2}\mathbb{V}^{\mathcal{F}_s}[\cdot], \qquad s\in[0,T], 
\end{equation*}
where $\gamma>0$ is a risk-aversion coefficient while $\rho(s)\geq 0$ is any first-order differentiable function, which captures the time-varying tradeoff between conditional expectation $\mathbb{E}^{\mathcal{F}_s}$ (return) and conditional variance $\mathbb{V}^{\mathcal{F}_s}$ (risk) given filtration $\mathcal{F}_s$ at time $s$. By choosing appropriate $\rho$, we could capture the time-varying risk aversion, especially relative to the investment horizon. Our main context would treat it as a deterministic function in time while they could depend on the current state $X^u_s$ to accommodate with state-dependent risk aversion studied in \cite{Bjoerk2014a}.
The general form of MV objectives, which allows for running MV evaluations, is given by 
\begin{equation} \label{MVfunctional}
	J(s,y;u)=\int^T_s\eta(s,\tau)\mathbb{MV}^{\mathcal{F}_s}[\psi(X^u_\tau)]d\tau+\mu(s,T)\mathbb{MV}^{\mathcal{F}_s}[\psi(X^u_T)], \quad s\in[0,T], ~~ y\in\mathbb{R}^d, 
\end{equation}
where $\eta(s,\tau)$ and $\mu(s,T)$ represent the general discounting factor and the $d$-dimensional underlying dynamics $(X^u_\tau)_{\tau\in[s,T]}$ is a forward stochastic differential equation (FSDE) of the form
\begin{equation} \label{Statedynamics}
	dX^u_\tau=b(\tau,X^u_\tau,u_\tau)d\tau+\sigma(\tau,X^u_\tau,u_\tau)d\mathcal{W}_\tau,  
\end{equation}
driven by a $k$-dimensional Brownian motion $\{\mathcal{W}_\cdot\}$ and controlled by a $n$-dimensional control process $\{u_\cdot\}$. Under some suitable conditions, both \eqref{MVfunctional} and \eqref{Statedynamics} are well-defined. The two general discounting factors indicate how much future cash flows or benefits are worth in today's terms. Typically, two commonly used discounting factors are presented here.

\begin{itemize}
	\item Exponential discounting: $\eta(s,\cdot)=\mu(s,\cdot)=e^{-\lambda(\cdot-s)}$ with some constant $\lambda\geq 0$. In this case, the TIC sources comes only from the MV operators. 
	\item Non-exponential discounting: as long as either $\eta$ or $\mu$ is of non-exponential form, such as hyperbolic disounting \cite{Laibson1997}, it causes the optimal control problem TIC. Even when the exponential functions $\eta$ and $\mu$ adopt different $\eta(s, \tau)=e^{-\lambda_1(\tau-s)}$ and $\mu(s, T)=e^{-\lambda_2(T-s)}$,
    the corresponding problem would have become TIC as long as $\lambda_1 \neq \lambda_2$; see \cite{marin2010consumption,marin2011non}.  
\end{itemize}
The objective \eqref{MVfunctional} is economically significant. It encourages the allocation policy to be dynamically efficient throughout the whole horizon, in which we can incorporate with time-varying risk aversion and (present-biased) discounting factor. It is clear that \eqref{MVfunctional} degenerates to the usual MV criterion in \cite{Markowitz1952,Basak2010} if $\rho(s)=1$, $\eta(s,\tau)=0$, and $\mu(s,T)=1$. The optimal asset allocation problem is to search for a dynamically optimal policy that can optimze the objective \eqref{MVfunctional} in some sense (e.g., precommitment or consistent-planning). 

The objective of \eqref{MVfunctional} leads its stochastic control problem TIC in three ways: First, the MV criteria with $\mathbb{V}^{\mathcal{F}_s}[X^u_\cdot]=\mathbb{E}^{\mathcal{F}_s}[(X^u_\cdot)^2]-(\mathbb{E}^{\mathcal{F}_s}[X^u_\cdot])^2$ is nonlinear in the conditional expectation. Second, the real-time adjustment of risk and return tradeoff with $\rho$ potentially depending on the initial time $s$. Third, the use of general discounting factors, where $\eta$ and $\mu$ can vary with time $s$. 

The problem of finding a time-consistent investment policy with TIC objectives resembles an intrapersonal game, akin to studies on consumer behavior under TIC preferences; see \cite{Peleg1973,Harris2001}. Specifically, the time-consistent investor (planner) preconceives the investment policy of her future selves and optimally reacts to them. Consequently, her investment policy arises from a pure-strategy Nash equilibrium in this intrapersonal game.

Before we formalize the discussion above, we embed the stochastic control problem with \eqref{MVfunctional} into the problems with a larger class of objectives for notational simplicity. Let us now consider a relatively general TIC stochastic control problem with an objective at time $s$ and state value $y$ given by
\begin{equation} \label{GeneralTIC}
\begin{split}
    J(s,y;u) =& \mathbb{E}^{\mathcal{F}_s}\left[\int^T_s C\big(s,\tau,y,X^u_\tau,u(\tau,X^u_\tau)\big)d\tau + F\big(s,y,X^u_T\big)\right]\\
    &+\int^T_s H\big(s,\tau,y,\mathbb{E}^{\mathcal{F}_s}\big[\psi\big(X^u_\tau\big)\big]\big)d\tau + G\big(s,y,\mathbb{E}^{\mathcal{F}_s}\big[\psi\big(X^u_T\big)\big]\big),
\end{split}
\end{equation}
where the mapping $u:[0,T]\times\mathbb{R}^d\to\mathcal{U}$ with $\mathcal{U}\subseteq\mathbb{R}^m$ is an admissible control law in the sense that both the underlying dynamics $\{X^u_\cdot\}$ of \eqref{Statedynamics} and the TIC objective \eqref{GeneralTIC} are well-defined. In the above, $C$, $F$, $H$, $G$, and $\psi$ are all deterministic functions of suitable dimensions. It is clear that the MV objective functional \eqref{MVfunctional} can be written in the form of \eqref{GeneralTIC} and the degeneration will be presented in the next section. The two $\psi$ of \eqref{GeneralTIC} could be different, however, we make them identical to simplify our formulation. It is noteworthy that the TIC originates from two main sources. The first source is that the entire cost functional \eqref{GeneralTIC} is contingent on the initial spatial-temporal reference point $(s,y)$ of the sub-problem being addressed. For instance, a rational investor's preferences and tastes may be dynamically refined based on the time-to-maturity and the current wealth status. The second TIC source is the nonlinearity in the arguments of conditional expectations in $H$ and $G$, which violate the BPO.


The fundamental idea behind treating the TIC stochastic control problem as an intrapersonal game is to consider a game with a continuum of players (selves) over the time interval $[0,T]$. Each player $s$, $s\in [0,T]$ is guided by the TIC objective at time $s$ and selects an optimal strategy by assuming that future policies are preconceived and reacting optimally to them. Moreover, it is also assumed that each player $s$, $s\in[0,T]$ makes the optimal decision and only affects the controlled system over a minimal time elapse period from time $s$. An equilibrium policy of the TIC control problem is defined as the set of strategies chosen by each player, which remains an equilibrium in any subgame of this game, i.e., a subgame-perfect equilibrium (SPE). The following gives a mathematical definition of the above and a similar definition can be found in \cite{Bjoerk2014,Bjoerk2017}.

\begin{definition}[Equilibrium policy and Equilibrium value function] \label{Def:EquiPolicy}
Consider an admissible control policy $\widehat{u}$ as a candidate equilibrium policy. For an arbitrary admissible control law $u$ and a fixed real number $\Delta s>0$, given an initial point $(s,y)$, we define a perturbed policy $u_{\Delta s}$ by 
\begin{equation}
    u_{\Delta s}(\epsilon,y)={\left\{\begin{array}{l l}{u(\epsilon,y)}&{{\mathrm{for~}}s\leq \epsilon< s+\Delta s,y\in\mathbb{R}^{d},} \\
    ~ \\ 
    {\widehat{u}(\epsilon,y)}&{{\mathrm{for~}}s+\Delta s\leq \epsilon\leq T,y\in\mathbb{R}^{d}.}\end{array}\right.}
\end{equation}

\noindent If the candidate law $\widehat{u}$ and the perturbed one $u_{\Delta s}$ satisfy the inequality 
\begin{equation} \label{Local optimality} 
			\underset{\Delta s\downarrow 0}{\underline{\lim}}\frac{J(s,y;\widehat{u})-J(s,y;u_{\Delta s})}{\Delta s}\geq 0, 
		\end{equation} 
then we say $\widehat{u}$ is an equilibrium policy and the equilibrium value function is defined by $V(s,y)=J(s,y;\widehat{u})$.   
\end{definition}

\begin{remark}
    Before we move on to the next stage of analysis, it is worth providing some explanatory notes on the concept of equilibrium policy. First of all, the focus of this paper is not on refining and rigorously formulating TIC control problems. Therefore, we adopt a relatively early and original definition of equilibrium policy here. For a more rigorous discussion of the related TIC issues, existing literature has refined and polished this definition from multiple perspectives. Nevertheless, all subsequent works preserve its fundamental essence: an equilibrium policy must attain local optimality, as delineated by \eqref{Local optimality}. Here, we highlight several recent developments for readers interested in exploring these topics further: (1) a comparison between closed-loop and open-loop controls is discussed in \cite{Wei2017}; (2) \cite{He2021} differentiated between weak and strong equilibria in TIC problems, highlighting the absence of strong equilibrium strategies in a general setting. They introduced a new concept called regular equilibrium, demonstrating its implication for weak equilibrium and providing a sufficient condition for a weak equilibrium strategy to transition into a regular equilibrium.
\end{remark}

The equilibrium strategy and equilibrium value function are defined by a recursive equation, similar to a Bellman equation in the classical time-consistent setting, now interpreted as the optimal response function of each player to the actions of other players. To understand or derive the TIC recursive equation, one should start with a discrete setting as in \cite{Wei2017}. As the time mesh size approaches zero, the recursive equation transitions informally into a system of HJB equations (also known as extended HJB system or equilibrium HJB equation). For a rigorous derivation and analysis with discretization and limiting process, we refer the readers to \cite{Yong2012,Wei2017,He2021,bensoussan2013mean} for their comprehensive insights and discussions and we omit the lengthy but well-studied derivation here. The extended HJB system \eqref{HJBSys} is presented as a definition rather than a formal proposition.

\begin{definition}[Equilibrium HJB system] The extended HJB system of equations for $V(s,y)$, $f(t,s,x,y)$, and $\{g^\tau(s,y)\}_{\tau\in[0,T]}$ is defined as
follows: for $0\leq t \leq s \leq \tau \leq T$ and $x,y\in\mathbb{R}^d$, 
\begin{equation} \label{HJBSys}
    \left\{
    \begin{array}{rcl}
        \sup\limits_{a\in\mathcal{U}}\Big\{\mathbb{A}^a V(s,y)+C(s,s,y,y,a)+H(s,s,y,\psi(y)) && \\
        -\Delta^a_1-\int^T_s\Delta^a_2(\tau)d\tau -\Delta^a_2(T) \Big\}   &= 0,  \\
         
        \mathbb{A}^{\widehat{u}} f(t,s,x,y)+ C(t,s,x,y,\widehat{u}(s,y)) &= 0, \\
        \mathbb{A}^{\widehat{u}} g^\tau(s,y)  &= 0,
    \end{array}
    \right.
\end{equation}
with the terminal conditions: $V(T,y)=f(T,T,y,y)+G(T,y,\psi(y))$, $f(t,T,x,y)=F(t,x,y)$, and $g^\tau(\tau,y)=\psi(y)$. Moreover, $\mathbb{A}^a=\frac{\partial}{\partial s}+b(s,y,a)\frac{\partial}{\partial y}+\frac{1}{2}\sigma^2(s,y,a)\frac{\partial^2}{\partial y^2}$ is the conventional controlled infinitesimal operator, and the adjustment terms $\Delta^a_1$ and $\Delta^a_2(\tau)$ that used to revive the recursion are defined as  
\begin{equation*}
\left\{
\begin{array}{rcl}
    \Delta^a_1 & := & \mathbb{A}^u f(s,s,y,y)-\mathbb{A}^u f(t,s,x,y)|_{t=s,x=y}, \\
    \Delta^a_2(\tau) & := & \mathbb{A}^u H(s,\tau,y,g^\tau(s,y))-\frac{\partial H}{\partial g^\tau}(s,\tau,y,g^\tau(s,y))\cdot\mathbb{A}^u g^\tau(s,y), \quad 0\leq\tau\leq T,
\end{array}
\right. 
\end{equation*} 
where we noted $G(\cdot,\cdot,\cdot)=H(\cdot,T\cdot,\cdot)$. Then, $f$ and $g$ have the following probabilistic representations
\begin{equation} \label{ProbInter}
\begin{split}
    f(t,s,x,y)=&\mathbb{E}^{\mathcal{F}_s}\left[\int^T_s C\big(t,\tau,x,X^u_\tau,u(\tau,X^u_\tau)\big)d\tau + F\big(t,x,X^u_T\big)\right], \\ g^\tau(s,y)=&\mathbb{E}^{\mathcal{F}_s}\big[\psi\big(X^u_\tau\big)\big]
\end{split}
\end{equation} 
Furthermore, the function $\widehat{u}$ realizing the supremum in the $V$-equation of \eqref{HJBSys} is the equilibrium policy, as per Definition \ref{Def:EquiPolicy}, and the corresponding equilibrium value function $V$ is characterized by 
\begin{equation} \label{V}
    V(s,y)=f(s,s,y,y)+\int^T_s H(s,\tau,y,g^\tau(s,y))d\tau+G(s,y,g^T(s,y)).
\end{equation}
\end{definition}

It is obvious that \eqref{HJBSys} is a fully coupled system of PDEs in the sense that we need to simultaneously determinate the unknown functions $V$ , $f$ and $g^\tau$. More specifically, to solve the $V$--equation, one must know $f$ and $g^\tau$. However, these functions are determined by the equilibrium policy $\widehat{u}$, which is, in turn, defined by the sup-part of the $V$ equation. In fact, as noted in \cite{Bjoerk2014}, assuming $f$ and $g$ are both known, then in the $V$ equation, every term except for the $\mathbb{A}^aV$ term can be considered as part of a new objective function's running term. This leads to an interesting conclusion: any TIC problem can be transformed into a classical time-consistent one, even though the cost function in this new problem is mainly of theoretical interest and has little ``practical" application. The functions $f$ and $g$ in \eqref{HJBSys} are precisely the two factors mentioned earlier in Introduction to address TIC. Based on them, $\Delta^a_2$ is designed to mitigate the nonlinearity in conditional expectations, while $\Delta^a_1$ serves to balance the reference-dependent variations in the decision-maker's preferences. Apparently, when the TIC objective \eqref{GeneralTIC} has neither initial dependence nor nonlinearity in conditional expectations, the coupled system \eqref{HJBSys} aligns with the classical HJB equation; see \cite{Bjoerk2014,Bjoerk2017,Yong1999} for more details.

The derivation of the equilibrium HJB system \eqref{HJBSys} alone does not justify its mathematical connection with the TIC stochastic control. We need to resolve two issues: (a) \textbf{Sufficiency}: the solution of \eqref{HJBSys} indeed gives an equilibrium policy and an equilibrium value function; (b) \textbf{Necessity}: every equilibrium policy must maximize the Hamiltonian associated to TIC problem and the corresponding value function solves \eqref{HJBSys}. By similar arguments in \cite{Bjoerk2017}, it is easy to establish the verification theorem (Sufficiency) for the Markovian setting while the non-Markovian setting was attempted in \cite{hernandez2023me}. The Necessity issue is difficult and we refer the readers to its latest progress, such as \cite{Lindensjoe2019,hernandez2023me,He2021,Hamaguchi2021} and a comprehensive literature review of the field in \cite{He2022}.

When it comes to well-posedness of solutions of \eqref{HJBSys}, this is a complex challenge within the context of partial differential equations (PDEs). The complexity is particularly significant when the drift $b$ and volatility $\sigma$ in the state dynamics $\{X^u_\cdot\}$ of \eqref{Statedynamics} both contain control variable $u$, resulting in a full nonlinearity of \eqref{HJBSys}. Additionally, $V$, $f$, and $g$ are tightly coupled together, which presents further analytical difficulties. In the existing literature, \cite{Wei2017,Yong2012} provides an analytical method to examine such an HJB system, while \cite{hernandez2023me,Hamaguchi2020} proposes a corresponding probabilistic approach for analysis. However, these literature could only explore cases with controlled drift while strictly prohibiting any control from influencing the volatility. In fact, only when the controls could or would take effect on the magnitude of uncertainty, the stochastic problems differ from the deterministic ones. The recent works of \cite{Lei2023,lei2023well,Lei2024} break the bottleneck and allow the diffusion to be controllable. They developed new PDE techniques and established the well-posedness of \eqref{HJBSys} over a maximally defined interval, which implied global solvability if a very sharp a-priori estimate is available. Note that another popular approach to analyzing HJB equation (in time-consistent cases) using weak (or viscosity) solutions was not yet translated to the TIC setting, which requires further investigation.

\section{MV Analysis: Equilibrium Policy and Value Function} \label{sec:mv}
Before analyzing the dynamic MV portfolio problem in an incomplete market, we first consider a complete market in Section \ref{sec:mvcm}, which is interesting in its own right. Moreover, when we move to the incomplete market in Section \ref{sec:mvicm}, we could directly call the results from Section \ref{sec:mvcm} and explicitly characterize the hedging demand due to the market incompleteness. The equilibrium policy and equilibrium value function are driven by some differential equations, whose well-posedness is proved in Section \ref{subsec:wellposedness}. The general framework will be then applied to CKLS-type stochastic factor models, under which all technical assumptions of our framework are verified.

We now specify the basic setup of the markets. Let us consider one risky asset (stock/index) and one risk-free asset (bond) available in the market for trading in continuous-time without frictions. The dynamics of the stock/index price $\{S_\cdot\}$, the stochastic factor driving the stock/index price $\{R_\cdot\}$, and the investor's wealth process $\{W_\cdot\}$ are given by   
\begin{equation}  \label{Completedynamics}
\left\{
    \begin{array}{rcl}
    dS_s/S_s & = & \mu(s,R_{s})ds+{\sigma}(s,R_{s})dB_{s}, \\
    dR_{s} & = & m(s,R_{s})ds+n(s,R_{s})dB^R_{s}, \\
    dW_{s} & = & [r_0 W_{s}+u_{s}\beta(s,R_{s})]ds+u_s\sigma(s,R_{s})dB_s,
    \end{array}
\right. 
\end{equation} 
where the mean rate $\mu$ and the volatility rate $\sigma$ of the asset return are deterministic functions, $\beta(s,R_{s})=\mu(s,R_{s})-r_0$ is the excess return rate over the risk-free rate $r_0$, the drift $m$ and the diffusion coefficient $n$ of the stochastic factor are both deterministic functions, and $\{u_\cdot\}$ is the stochastic control process representing the money amount invested in the risky asset. In this paper, all the stochastic analyses  are conducted in a filtered probability space
$(\Omega,\mathcal{F}, \{\mathcal{F}_s\}_{s\in[0,T]}, \mathbb{P})$, on which two correlated Brownian motions, $B$
and $B^R$, with correlation $\varrho\in(-1,1)$ are defined; all the stochastic processes are assumed to be adapted to $\mathbb{F}:=\{\mathcal{F}_s\}_{s\in[0,T]}$, the augmented filtration generated by $B$ and $B^R$.

We could let $\beta$ and $\sigma$ in the $W$-dynamics to absorb the risk-free interest rate such that we have
\begin{equation} \label{Updatedwealth}
    d\widetilde{W}_{s}=u_{s}\widetilde{\beta}(s)ds+u_s\widetilde{\sigma}(s)dB_s, 
\end{equation}
where $\widetilde{W}_s=W_s e^{r_0(T-s)}$, $\widetilde{\beta}(s,R_{s})=\beta(s,R_{s})e^{r_0(T-s)}$, and $\widetilde{\sigma}(s,R_{s})=\sigma(s,R_{s})e^{r_0(T-s)}$. 
Without loss of generality, we could consider $r_0=0$ while we will continue to use the notations in \eqref{Completedynamics} for the ease of the notational burdens in subsequent analyses. Though, we present all the results with the discounting factor.

In the setting of $\eqref{Completedynamics}$, the market completeness can be achieved in two ways: first, the two randomness sources are equivalent in the sense that $|\varrho|\uparrow 1$; second, both $\mu$ (and thus $\beta$) and $\sigma$ are independent of $\{R_\cdot\}$. The first case is trivial and not the focus of this paper. We refer to the second case as the complete market under our consideration. Note that in the complete-market setting, the second line of \eqref{Completedynamics} (the $R$-dynamics) does not play a role in a stochastic control problem with the $W$-dynamics.


\subsection{MV Analysis in Complete Markets ($\mu\equiv \mu(s),~\sigma\equiv \sigma(s)$)} \label{sec:mvcm}
In complete markets, there is no hedging demand against fluctuations in the investment opportunities and thus we can anticipate a myopic investment policy. 
Throughout this paper, for the well-posedness of the stochastic control problem, we assume that $\sup_{0\leq s\leq T}|\rho(s)|\leq c$ and $\sup_{0\leq s\leq \tau\leq T}|\lambda(s,\tau)|\leq c$, where $\lambda(s,t)=\eta(s,t)/\mu(s,T)$.

The MV objective \eqref{MVfunctional} can be reformulated as: 
\begin{equation} \label{equiMV}
\begin{split}
    J(s,w;u)=&\mathbb{E}^{\mathcal{F}_s}\left[\int^T_s \lambda(s,\tau)\Phi(s,W_\tau)d\tau+\Phi(s,W_T)\right]\\
    &+\int^T_s \lambda(s,\tau)\Psi\left(\mathbb{E}^{\mathcal{F}_s}[W_\tau]\right)d\tau+\Psi\left(\mathbb{E}^{\mathcal{F}_s}[W_T]\right), 
\end{split}
\end{equation}
where $\Phi(s,w)=\rho(s)w-\frac{\gamma}{2}w^2$, and $\Psi(w)=\frac{\gamma}{2}w^2$. Our aim is then to search for the equilibrium MV policy that maximizes the objectives \eqref{equiMV} as per Definition \ref{Def:EquiPolicy}. To this end, let us examine the corresponding extended HJB system \eqref{HJBSys} associated with the dynamics setting \eqref{Completedynamics} with $\mu\equiv \mu(s),~\sigma\equiv \sigma(s)$ and the MV objective \eqref{equiMV}, where we should identify the followings
\begin{equation*}
\begin{split}
C(t,s,w,a)=\lambda(t,s)\left(\rho(s)w-\frac{\gamma}{2}w^2\right),&\quad
F(t,w)=\frac{\gamma}{2}w, \\
H(t,s,w)=\frac{\gamma}{2}\lambda(t,s)w^2,& \quad G(t,w)=\frac{\gamma}{2}w^2.
\end{split}
\end{equation*}
%
%
%
%
%
%

Moreover, we know the following components:
\begin{equation*} 
    \begin{cases}{l}
        \begin{aligned}
        C(t,s,w,a)+H(t,s,\psi(w))=&\lambda(t,s)\rho(t)w-\frac{\lambda(t,s)\gamma}{2}w^2+\frac{\lambda(t,s)\gamma}{2}\psi^2(w)\\
        =&\lambda(t,s)\rho(t)w,  
        \end{aligned}\\

        \begin{aligned}
            \Delta^a_1:=&\mathbb{A}^a f(s,s,w)-\mathbb{A}^a f(t,s,w)|_{t=s}\\
            =&f_t(t,s,w)|_{t=s},             
        \end{aligned}\\
        
        \begin{aligned}
        \Delta^a_2(\tau):=&\mathbb{A}^u H(s,\tau,g^\tau(s,w))-\frac{\partial H}{\partial g^\tau}(s,\tau,g^\tau(s,w))\cdot\mathbb{A}^u g^\tau(s,w)\\
        =&\frac{\gamma}{2}\lambda_t(t,\tau)|_{t=s}(g^\tau)^2(s,w) +\frac{\gamma}{2}\widetilde{\sigma}^2(s)a^2\lambda(s,\tau)(g^\tau_w)^2(s,w). 
        \end{aligned}
    \end{cases} 
\end{equation*} 
By noting that $\lambda_t(t,T)=0$ and $\lambda(t,T)=1$, one also has $\Delta^a_2(T)=\frac{\gamma}{2}\sigma^2(s)a^2(g^T_w)^2(s,w)$. Consequently, the $V$--equation of the extended HJB system \eqref{HJBSys} reads 
\begin{equation*}
    \begin{split}
        V_s(s,w)+\sup\limits_{a\in\mathcal{U}}\bigg\{&\frac{1}{2}\widetilde{\sigma}^2(s)a^2V_{ww}(s,w)+\widetilde{\beta}(s)a V_w(s,w)+\lambda(s,s)\rho(s)w\\
        & -f_t(t,s,w)|_{t=s} -\frac{\gamma}{2}\int^T_s\lambda_t(t,\tau)|_{t=s}(g^\tau)^2(s,w)d\tau\\
        &-\frac{\gamma}{2}\widetilde{\sigma}^2(s)a^2\int^T_s\lambda(s,\tau)(g^\tau_w)^2(s,w)d\tau-\frac{\gamma}{2}\widetilde{\sigma}^2(s)a^2 (g^T_w)^2(s,w)\bigg\}=0.
    \end{split}
\end{equation*}
The Hamiltonian (the supremum problem above) is quadratic and concave in $a$ (with $V_ww\le 0$ that will be satisfied with the later ansatz), which implies that the unique equilibrium policy,  solved from the first-order optimality condition, is given by
\begin{equation*}
    \widehat{u}(s,w)=-\frac{\beta(s)}{\sigma^2(s)}\frac{V_w(s,w)}{V_{ww}(s,w)-\gamma\int^T_s\lambda(s,\tau)(g^\tau_w)^2(s,w)d\tau-\gamma (g^T_w)^2(s,w)}e^{-r_0(T-s)}.
\end{equation*}

By integrating \eqref{Completedynamics} from $s$ to $\tau$ and substituting $W_\tau$ into the MV objective function \eqref{equiMV}, as well as examining the structure of the coefficients in \eqref{Completedynamics} and the terminal conditions of $V$, $f$, and $g$, we leverage the following ansatzs:
\begin{equation} \label{AnsatzComplete}
\begin{split}
    V(s,w)&=A(s)w+B(s), \\ 
    f(t,s,w)&=-\frac{\gamma}{2}L(t,s)w^2+M(t,s)w+N(t,s), \\
    g^\tau(s,w)&=a^\tau(s)w+b^\tau(s).
\end{split}
\end{equation} 
We also immediately obtain the $t$-derivative of $f$: $f_t(t,s,w)=-\frac{\gamma}{2}L_t(t,s)w^2+M_t(t,s)w+N_t(t,s)$. Subsequently, with \eqref{AnsatzComplete}, the equilibrium MV policy reads  
\begin{equation} \label{AnsatzComPoli}
    \widehat{u}(s,w)=\frac{1}{\gamma}\frac{\beta(s)}{\sigma^2(s)}\frac{A(s)}{\int^T_s \lambda(s,\tau)(a^\tau)^2(s)d\tau+ (a^T)^2(s)}e^{-r_0(T-s)}. 
\end{equation}

By substituting \eqref{AnsatzComplete} and \eqref{AnsatzComPoli} into the extended HJB system \eqref{HJBSys} and employing the separation of variables method, we obtain the following coupled system of nonlocal ODEs for $(A,B,L,M,N,a,b)$: for $0\leq t \leq s \leq \tau \leq T$,
\begin{equation} \label{NonlocalODESys}
\left\{
    \begin{array}{rcl}
        L_t(t,s)|_{t=s}-\int^T_s\lambda_t(t,\tau)|_{t=s}(a^\tau)^2(s)d\tau&=&0 \\
        A_s(s)-M_t(t,s)|_{t=s}+\lambda(s,s)\rho(s)-\gamma\int^T_s\lambda_t(t,\tau)|_{t=s}a^\tau(s)b^\tau(s)d\tau &=& 0, \\
        B_s(s)+\frac{1}{2\gamma}\frac{\beta^2(s)}{\sigma^2(s)}\frac{A^2(s)}{\int^T_s \lambda(s,\tau)(a^\tau)^2(s)d\tau+ (a^T)^2(s)}-N_t(t,s)|_{t=s}&&\\
        -\frac{\gamma}{2}\int^T_s\lambda_t(t,\tau)|_{t=s}(b^\tau)^2(s)d\tau
        &=& 0, \\
        (L_t)_s(t,s) + \lambda_t(t,s) &=& 0, \\
        (M_t)_s(t,s)-\frac{\beta^2(s)}{\sigma^2(s)}\frac{A(s)L_t(t,s)}{\int^T_s\lambda(s,\tau) (a^\tau)^2(s)d\tau+ (a^T)^2(s)}&&\\
        +\lambda_t(t,s)\rho(t)+\lambda(t,s)\rho_t(t) &=& 0, \\
        (N_t)_s(t,s)-\frac{1}{2\gamma}\frac{\beta^2(s)}{\sigma^2(s)}\left(\frac{A(s)}{\int^T_s \lambda(s,\tau)(a^\tau)^2(s)d\tau+ (a^T)^2(s)}\right)^2 L_t(t,s)&&\\
        +\frac{1}{\gamma}\frac{\beta^2(s)}{\sigma^2(s)}\frac{A(s)M_t(t,s)}{\int^T_s \lambda(s,\tau)(a^\tau)^2(s)d\tau+ (a^T)^2(s)} &=& 0, \\
        a^\tau_s(s) &=& 0, \\
        \displaystyle{b^\tau_s(s)+\frac{1}{\gamma}\frac{\beta^2(s)}{\sigma^2(s)}\frac{A(s)a^\tau(s)}{\int^T_s \lambda(s,\tau)(a^\tau)^2(s)d\tau+(a^T)^2(s)}} &=& 0
    \end{array}
\right. 
\end{equation} 
with $ A(T)=\rho(T)$, $B(T)=0$, $L_t(t,T)=0$, $M_t(t,T)=\rho_t(t)$, $N_t(t,T)=0$, $a^\tau(\tau)=1$, and $b^\tau(\tau)=0$. 

Before delving deeper into the nonlocal ODE system \eqref{NonlocalODESys}, let us make some initial observations: (I) \textbf{Coupling of ODEs}: The system \eqref{NonlocalODESys} is coupled in the sense that that the three unknown functions $A$, $M$, and $b$ are interdependent and cannot be solved individually; (II) \textbf{Nonlocality in time}: Unlike conventional ODE systems, where unknown functions depend on a single time variable $s$, this nonlocal ODE system involves functions that depend on pairs of temporal variables, such as $(t,s)$ or $(s;\tau)$. This setup complicates the system, as it cannot be treated as a collection of independent ODEs parameterized by time variables $t$ or $\tau$. The presence of nonlocal terms like $L_t(t,s)|_{t=s}$, $M_t(t,s)|_{t=s}$, and $\int^T_s \lambda(s,\tau)(a^\tau)^2(s)d\tau$, requires a simultaneous consideration of all parameters to yield complete solutions. This nonlocality desires a broader perspective, where global information rather than purely local data is necessary for solving the system. This characteristic is why we refer to it as a nonlocal system. Such a dual-time-variable structure, arising from TIC preferences, is also observed in existing literature, including a PDE format in \cite{Wei2017,Yong2012,Lei2023,Lei2024,lei2023well} and an SDE variant in \cite{Lin2002,Yong2006,Wang2019,Wang2020,hernandez2023me,Hamaguchi2021}. 


Next, the following lemma reveals the equivalence between the nonlocal ODE system \eqref{NonlocalODESys} and a well-studied Volterra-type integral equation. With this result, we can show that \eqref{NonlocalODESys} is uniquely solvable. 
\begin{lemma} \label{LemmaVolterra}
The coupled system of nonlocal ODEs \eqref{NonlocalODESys} is equivalent to a linear Volterra integral equation of the second kind: 
    \begin{equation} \label{Volterrainteq}
    \begin{split}
        &\displaystyle{A(s)=\Theta(s)+\int^T_sK(s,\tau)A(\tau)d\tau}
    \end{split}
\end{equation}
with  
\begin{equation*} 
\left\{
    \begin{array}{rcl}
        \Theta(s)&=&\displaystyle{\rho(T)-\int^T_s\left[\rho_t(\delta)+\rho(\delta)\int^T_\delta\lambda_t(\delta,\eta)d\eta+\rho_t(\delta)\int^T_\delta\lambda(\delta,\eta)d\eta+\lambda(\delta,\delta)\rho(\delta)\right]d\delta}, \\
        K(s,\tau)&=&\displaystyle{\int^\tau_s\frac{\beta^2(\tau)}{\sigma^2(\tau)}\left(\frac{\int^T_\tau\lambda_t(\delta,\sigma)d\sigma}{\int^T_\tau\lambda(\tau,\sigma)d\sigma+1}-\int^\tau_\delta\frac{\lambda_t(\delta,\epsilon)}{\int^T_\tau \lambda(\tau,\sigma)d\sigma+1}d\epsilon\right)d\delta},
    \end{array}
\right. 
\end{equation*} 
which admits a unique continuous solution $A(s)$ in $[0,T]$. 
\end{lemma}

The complete determination of function $A$ through this linear Volterra integral equation \eqref{Volterrainteq} decouples the nonlocal ODE system \eqref{NonlocalODESys}, allowing it to be fully solvable using classical ODE theory. Moreover, we can derive its explicit expression.  

\begin{theorem} \label{EPCom}
The resolvent representation for the solution of \eqref{Volterrainteq} satisfies  
\begin{equation*} \label{Sol_Volterra integral eq}
    \displaystyle{A(s)=\Theta(s)+\int^T_s\mathcal{R}(s,\tau)\Theta(\tau)d\tau=\rho(s)\left(\int^T_s\lambda(s,\tau)d\tau+1\right), }
\end{equation*}
where $\mathcal{R}$ is the resolvent kernel associated with $K$ of \eqref{Volterrainteq}. As a result, the equilibrium MV policy for the dynamic MV asset allocation with TIC risk-return preferences \eqref{equiMV} in a complete market \eqref{Completedynamics} is given by:
\begin{equation} \label{EPComplete}
    \widehat{u}(s,W_s)=\frac{\rho(s)}{\gamma}\frac{\mu(s)-r_0}{\sigma^2(s)}e^{-r_0(T-s)}, \quad s\in[0,T].  
\end{equation}
\end{theorem}

Theorem \ref{EPCom} provides a fully analytical characterization of the equilibrium investment policy of MV investor in a complete market. The certain investment opportunities lead to myopic demands only. We obtain a similar form of equilibrium MV policy as in \cite{Basak2010} even we have considered the running MV objectives. Moreover, the equilibrium policy \eqref{EPComplete} shows some key characteristics. First, the equilibrium policy \eqref{EPComplete} does not depend on $W_s$.
Second, in the case of $\rho(s)=1$ over the whole duration, which means that investors generally do not modify their preference of risk and return during the investment horizon, our result is consistent with the previous ones in \cite{Basak2010,Bjoerk2014a}. However, when $\rho(s)$ is increasing over time, it suggests that the investor prioritizes returns. The strategy with \eqref{EPComplete} implies that the investor should then increase her investments in the risky assets.


\subsection{MV Analysis in Incomplete Markets} \label{sec:mvicm}
We are now ready to study the incomplete market with the setup \eqref{Completedynamics}. In this case, trading in stocks and bonds cannot perfectly hedge against the fluctuations in the stochastic investment opportunity set. Hence, we anticipate the hedging demand added to the MV policy on top of the myopic component of \eqref{EPComplete} as in \cite{Basak2010,Dai2021}. Note that when there is no correlation between the stock return and the stochastic factor ($\varrho=0$), there is no hedging demand for the stochastic factor because trading in the stock cannot offset any fluctuations in the stochastic factor from the modeling perspective. 


In this section, we use BSDE theory to formulate and derive the explicit expression of the equilibrium MV policy in a general incomplete-market economy. Before we again examine the corresponding extended HJB system \eqref{HJBSys}, we note that the controlled infinitesimal operator $\mathbb{A}^a$ now transforms an arbitrary twice continuously differentiable function $\phi(s,w,r)$ as follows: 
\begin{equation*}
    \mathbb{A}^a\phi(s,w,r)=\frac{\partial\phi}{\partial s}+a\widetilde{\beta}\frac{\partial\phi}{\partial w}+m\frac{\partial\phi}{\partial r}+\frac{1}{2}\left(a^2\widetilde{\sigma}^2\frac{\partial^2\phi}{\partial w^2}+n^2\frac{\partial^2\phi}{\partial r^2}+2\varrho a\widetilde{\sigma} n\frac{\partial^2\phi}{\partial w\partial r}\right).
\end{equation*}

With the dynamic setting of \eqref{Completedynamics} and the MV objective \eqref{equiMV}, the corresponding extended HJB system \eqref{HJBSys} reads
\begin{equation*} 
    \begin{cases}
        \begin{aligned}
            C(t,s,w,r,a)+H(t,s,\psi(w,r))= &\lambda(t,s)\rho(t)w-\frac{\lambda(t,s)\gamma}{2}w^2+\frac{\lambda(t,s)\gamma}{2}\psi^2(w,r)\\
            =&\lambda(t,s)\rho(t)w,
        \end{aligned} \\ 
        
        \begin{aligned}
        \Delta^a_1:=&\mathbb{A}^a f(s,s,w,r)-\mathbb{A}^a f(t,s,w,r)|_{t=s}\\
        =&f_t(t,s,w,r)|_{t=s},             
        \end{aligned}\\

        \begin{aligned}
            \Delta^a_2(\tau):= & \mathbb{A}^u H(s,\tau,g^\tau(s,w,r))-\frac{\partial H}{\partial g^\tau}(s,\tau,g^\tau(s,w,r))\cdot\mathbb{A}^u g^\tau(s,w,r) \\
            =&\frac{\gamma}{2}\lambda_t(t,\tau)|_{t=s}(g^\tau)^2(s,w,r)+\frac{\gamma}{2}\widetilde{\sigma}^2(s,r)a^2\lambda(s,\tau)(g^\tau_w)^2(s,w,r)\\
            &+\frac{\gamma}{2}n^2(s,r) \lambda(s,\tau)(g^\tau_r)^2(s,w,r)+\gamma\varrho(n\widetilde{\sigma})(s,r)a\lambda(s,\tau)g^\tau_wg^\tau_r(s,w,r). 
        \end{aligned}
    \end{cases}
\end{equation*}

        

         
Similarly, one also has $\Delta^a_2(T)$ by noting $\lambda(t,T)=1$. Consequently, the $V$-equation of the extended HJB system \eqref{HJBSys} satisfies 
\begin{equation*}
    \begin{split}
        & {V_s(s,w,r)+\frac{1}{2}n^2(s,r)V_{rr}(s,w,r)+m(s,r)V_r(s,w,r)-\frac{\gamma}{2}n^2(s,r)\int^T_s \lambda(s,\tau)(g^\tau_r)^2(s,w,r)d\tau}\\
        & \qquad {-\frac{\gamma}{2}n^2(s,r)(g^T_{r})^2(s,w,r)+\lambda(s,s)\rho(s)w-f_t(t,s,w,r)|_{t=s}-\frac{\gamma}{2}\int^T_s\lambda_t(t,\tau)|_{t=s}(g^\tau)^2(s,w,r)d\tau} \\
        & \qquad +\sup\limits_{a\in\mathcal{U}}\bigg\{\frac{1}{2}\widetilde{\sigma}^2(s,r)a^2V_{ww}(s,w,r)+\varrho (n\widetilde{\sigma})(s,r)aV_{wr}(s,w,r)+\widetilde{\beta}(s,r)a V_w(s,w,r) \\
        & \qquad -\frac{\gamma}{2}\widetilde{\sigma}^2(s,r)a^2\int^T_s\lambda(s,\tau)(g^\tau_w)^2(s,w,r)d\tau -\gamma\varrho(n\widetilde{\sigma})(s,r)a\int^T_s\lambda(s,\tau)g^\tau_w g^\tau_r(s,w,r)d\tau\\    
        & \qquad -\frac{\gamma}{2}\widetilde{\sigma}^2(s,r)a^2 (g^T_w)^2(s,w,r)-\gamma\varrho(n\widetilde{\sigma})(s,r)ag^T_w g^T_r(s,w,r) \bigg\}=0, \\
    \end{split}
\end{equation*}
whose supremum problem implies the equilibrium investment policy taking the form
\begin{equation*}
    \begin{split}
        \widehat{u}(s,w,r)  = &\frac{1}{\widetilde{\sigma}^2(s,r)V_{ww}(s,w,r)-\gamma\widetilde{\sigma}^2(s,r)\int^T_s\lambda(s,\tau)(g^\tau_w)^2(s,w,r)d\tau-\gamma \widetilde{\sigma}^2(s,r)(g^T_w)^2(s,w,r)}\\
        & \bigg[-\widetilde{\beta}(s,r)V_w(s,w,r)-\varrho(n\widetilde{\sigma})(s,r)V_{wr}(s,w,r)\\
        &+\gamma\varrho(n\widetilde{\sigma})(s,r)\int^T_s\lambda(s,\tau)g^\tau_wg^\tau_r(s,w,r)d\tau+\gamma\varrho(n\widetilde{\sigma})(s,r)g^T_wg^T_r(s,w,r)\bigg].
    \end{split}
\end{equation*}

From the previous analysis of MV problems in a complete market, we know that $V$ and $g$ are both similarly separable in $w$ and satisfy 
\begin{equation} \label{AnsatzIncom}
\left\{
    \begin{array}{l}                    \displaystyle{V(s,w,r)=A(s)w+B(s,r)=\rho(s)\left(\int^T_s\lambda(s,\tau)d\tau+1\right)w+B(s,r),} \\ 
    \displaystyle{g^\tau(s,w,r)=a^\tau(s)w+b^\tau(s,r)=w+\mathbb{E}_{s,r}\left[\int^\tau_s\widetilde{\beta}(s,R_s)\widehat{u}(s,R_s)ds\right],}
    \end{array}
\right. 
\end{equation}
under the assumption of $\mathbb{E}_{s,r}\big[\int^T_s\widetilde{\sigma}^2_s\widehat{u}^2_sds\big]<\infty$, which can be verified by our theoretical results of nonlocal BSDEs in Section \ref{subsec:wellposedness} and the specific financial applications in Section \ref{subsec:apps}. Consequently, the equilibrium policy is independent of $w$ and reads 
\begin{equation} \label{PDEequipolicy}
    \begin{split}
        \widehat{u}(s,r) = \frac{1}{\gamma}\frac{\widetilde{\beta}(s,r)\rho(s)}{\widetilde{\sigma}^2(s,r)}-\frac{\varrho}{\widetilde{\sigma}(s,r)}\frac{\int^T_s\lambda(s,\tau)n(s,r)b^\tau_r(s,r)d\tau+n(s,r)b^T_r(s,r)}{1+\int^T_s\lambda(s,\tau)d\tau}.
    \end{split}
\end{equation}
It is noteworthy that $b^\tau(s,r)$ represents an conditional expectation. Thanks to the well-known Feynman-Kac formula, we can reformulate \eqref{PDEequipolicy} into a probabilistic setting as follows:
\begin{equation} \label{EPincom}
    \begin{split}
        \widehat{u}(s,R_s) = \frac{1}{\gamma}\frac{\mu(s,R_s)-r_0}{\sigma^2(s,R_s)}\rho(s) e^{-r_0(T-s)}-\frac{\varrho}{\sigma(s,R_s)}\frac{\int^T_s\lambda(s,\tau) Z^\tau_sd\tau+Z^T_s}{1+\int^T_s\lambda(s,\tau)d\tau}e^{-r_0(T-s)},
    \end{split}
\end{equation} 
where $\{R_\cdot\}$ and $\{Z^\tau_\cdot\}_{\tau\in[0,T]}$ are parts of the solution $\{(R_\cdot,Y^\tau_\cdot,Z^\tau_\cdot)\}_{\tau\in[0,T]}$ of the following family of FBSDEs parameterized by $\tau\in[0,T]$:
\begin{equation} \label{FBSDEsys}
\left\{
    \begin{array}{l}
        \displaystyle{dR_s=m(s,R_s)ds+n(s,R_s)dB^R_s}, \\
        \displaystyle{dY^\tau_s=-\left(\frac{\rho(s)}{\gamma}\left(\frac{\mu(s,R_s)-r_0}{\sigma(s,R_s)}\right)^2-\varrho\frac{\mu(s,R_s)-r_0}{\sigma(s,R_s)}\frac{\int^T_s\lambda(s,\tau) Z^\tau_sd\tau+Z^T_s}{1+\int^T_s\lambda(s,\tau) d\tau}\right)ds+Z^\tau_s dB^R_s,} \\
        R_t=R, \quad Y^\tau_\tau=0, \quad 0\leq t \leq s \leq \tau \leq T, \quad R\in\mathbb{R}. 
    \end{array}
\right. 
\end{equation} 
It is clear that the FBSDE \eqref{FBSDEsys} is decoupled in the sense that one can first solve $\{R_\cdot\}$ in this system and then independently examine the solvability of the BSDE. Suppose that the process $\{R_\cdot\}$ is known, the generator of BSDE is an adapted stochastic process rather than a deterministic function, which allows for an explicit dependence upon the paths of $\{B^R_\cdot\}$. Although we can study it as a general BSDE with random coefficients, it is noteworthy that its specific linear structure with respect to $\{Z^\tau_\cdot\}_{\tau\in[0,T]}$ makes it more manageable when viewed as a BSDE with stochastic Lipschitz conditions. 

\begin{remark}
    To better understand the advantages of our probabilistic approach compared to earlier studies in \cite{Basak2010,Dai2021}, we consider a special parameter setting of our MV formulation: $\rho(s)=1$ and $\lambda(t,s)=0$. For this degenerate case, \cite{Basak2010} also noted that the unknown function $b^T(s,r)$ represents a conditional expectation. According to \eqref{AnsatzIncom} and \eqref{PDEequipolicy}, it satisfies the following equation
\begin{equation} \label{SDU}
    b^T(s,r)=\mathbb{E}_{s,r}\left[\int^T_s\left(\frac{1}{\gamma}\frac{\beta^2(\epsilon,R_\epsilon)}{\sigma^2(\epsilon,R_\epsilon)}-\frac{\varrho\beta(\epsilon,R_\epsilon)n(\epsilon,R_\epsilon)}{\sigma(\epsilon,R_\epsilon)}b^T_r(\epsilon,R_\epsilon)\right)d\epsilon\right].
\end{equation}
Clearly, in order to provide a fully analytical characterization of the optimal investment policy \eqref{PDEequipolicy} in terms of the model parameters, one needs to eliminate the dependence on $b_r$ inside the conditional expectation \eqref{SDU}. By the Feynman–Kac theorem \cite{karatzas2014brownian}, \cite{Basak2010} obtains the following PDE: 
\begin{equation} \label{newPDE}
    b^T_s(s,r)+\left(m(s,r)-\frac{\varrho\beta(s,r)n(s,r)}{\sigma(s,r)}\right)b^T_r(s,r)+\frac{1}{2}n(s,r)b^T_{rr}(s,r)+\frac{1}{\gamma}\frac{\beta^2(s,r)}{\sigma^2(s,r)}=0,
\end{equation}
with $b^T(T,r)=0$, where the coefficient of $b^T_r$ is obtained by merging the original coefficient $m$ with the one preceding $b^T_r$ from the non-homogeneous term. Consequently, again by Feynman–Kac theorem, the solution of \eqref{newPDE} can be reinterpreted as 
\begin{equation} \label{simpleSDU} 
    b^T(s,r)=\widetilde{\mathbb{E}}_{s,r}\left[\int^T_s\frac{1}{\gamma}\frac{\beta^2(\epsilon,R_\epsilon)}{\sigma^2(\epsilon,R_\epsilon)}d\epsilon\right],
\end{equation}
where $\widetilde{\mathbb{E}}_{s,r}[\cdot]$ denotes the expectation under a new probability measure $\widetilde{\mathbb{P}}$ such that $\{R_\cdot\}$ now follows dynamics with a modified drift:
\begin{equation*}
    \displaystyle{dR_{s}=\left(m(s,R_{s})-\frac{\varrho\beta(s,R_s)n(s,R_s)}{\sigma(s,R_s)}\right)ds+n(s,R_{s})d\widetilde{B}^R_{s}},  
\end{equation*}
where $\{\widetilde{B}_\cdot\}$ is a Brownian motions under $\widetilde{\mathbb{P}}$. By strategically employing the Feynman–Kac theorem twice, it is clear that \cite{Basak2010} effectively eliminates the dependence on $b_r$ on the right hand side of \eqref{SDU}. Introducing a new probability measure and state dynamics yields a fully analytical expression \eqref{simpleSDU} for the conditional expectation $b$ of the optimal investment. This facilitates explicit computation of the optimal dynamic MV portfolios in a straightforward manner. However, this sophisticated technique—which relies on changing the probability measure of the conditional expectation to simplify the form of the random variables—depends heavily on three key assumptions: (I) \eqref{PDEequipolicy} is linear with respect to $b^T_r$ such that the unknown can be absorbed into coefficients of the differential operator of PDE \eqref{newPDE}. (II) The conditional expectation $b^T$ must be sufficiently smooth in $r$ for the Feynman-Kac theorem to be applicable and for its first-order derivative $b^T_r$ to be meaningful. (II) More importantly, the conditional expectation must be explicitly expressible to verify if it is regular enough to be twice continuously differentiable such that the equilibrium policy \eqref{PDEequipolicy} is well-defined; see \cite[Proposition 2 and Remark 8]{Basak2010}. This is why earlier studies \cite{Basak2010} were confined to exploring the CIR model and the OU process---both models offer clear and smooth expressions for moments of all orders. 
\end{remark}

Unlike earlier literature, which simplified the conditional expectation \eqref{SDU} by altering the underlying probability measure, our approach, inspired by the standard result in the theory of BSDEs \cite{Peng1992,Pardoux1992,Peng2011}, directly shapes $\{(b^\tau(s,R_s),$ $n(s,R_s) b^\tau_r(s,R_s))\}_{\tau\in[0,T]}$ of \eqref{AnsatzIncom} and \eqref{PDEequipolicy} as the unique solution $\{(Y^\tau_s,Z^\tau_s)\}_{\tau\in[0,T]}$ to the BSDE system \eqref{FBSDEsys}. This BSDE-based formulation \eqref{EPincom}-\eqref{FBSDEsys} of equilibrium MV strategy offers three significant advantages over the state-of-the-art approaches:
\begin{enumerate}
    \item \textbf{Offering well-posedness even without explicit forms of moments of stochastic factor models}. While the CKLS process \eqref{CKLSSys} is widely employed in mathematical finance to model various stochastic investment opportunities, it often lacks analytical tractability, except in rare cases where neither the Fourier transform nor the density is known. Previous studies \cite{Basak2010} heavily relied on explicit density expressions and the ability to compute various orders of moments, limiting their results to specific models like the CIR and OU models. However, in our BSDE-based analysis, density functions are not required, and we do not yet have an explicit formula for moment generating functions (MGFs). This advantage allows for the exploration of a broader range of stochastic factor models; see Section \ref{subsec:apps} for more details. 
    \item \textbf{Assigning interpretable significance to the equilibrium policy even if the required conditional expectation is merely continuous but not differentiable}. By presenting a probabilistic argument rather than an analytic one, our approach significantly eases the underlying assumptions, allowing for a more flexible and broader approach. Earlier studies not only required precise expressions for the expectation (i.e., $Y_\cdot$), but also that these expectations be differentiable to the first/second order. Otherwise, the equilibrium policy \eqref{CKLSpolicy} would lack a clear definition. However, as the Feynman--Kac formula in \cite{Pardoux1992,Yong1999} indicates, even when a given PDE admits only a continuous viscosity solution instead of a differentiable classical solution, we can still derive a corresponding BSDE, allowing us to define suitable interpretations for $\{(Y_\cdot,Z_\cdot)\}$. In this sense, we have expanded the range of interpretations for equilibrium policy \eqref{CKLSpolicy} in MV research.  
    \item \textbf{Allowing for the consideration of a nonlinear generator}. Our BSDE formulation only requires the generator to satisfy the stochastic Lipschitz condition (to be detailed) and do not require linearity with respect to $(Y^\tau_\cdot,Z^\tau_\cdot)$. This flexibility broadens the range of control scenarios where our method can be applied. In the MV asset allocation problem, the focus is on controlling the wealth management process. This inherently affects the drift and volatility, resulting in a corresponding BSDE generator that is linear with respect to $\{Z^\tau_\cdot\}$. However, if we consider a more general control problem, where the control process has greater flexibility in interacting with the control system, the corresponding BSDE would likely be nonlinear, though still adhering to the stochastic Lipschitz condition. 
\end{enumerate}

It is important to note that formulating the equilibrium MV policy as the solution to an FBSDE \eqref{FBSDEsys} without proving the existence and uniqueness of this FBSDE solution is incomplete. However, this is precisely where the difficulty of the problem lies. The two notable features of \eqref{FBSDEsys} elevate it beyond the classical FBSDE analysis framework: (1) Nonlocality and entanglement induced by integration over parameters, $\int^T_s\lambda(s,\tau) Z^\tau_sd\tau$. (2) Stochastic generator and stochastic Lipschitz conditions caused by SV models in incomplete markets, $(\mu_s-r_0)/\sigma_s$. We have analyzed the myopic part of \eqref{EPincom} in the previous section. Next, we will demonstrate that under suitable conditions, the FBSDE \eqref{FBSDEsys} is well-posed, leading to a well-defined equilibrium policy \eqref{EPincom}. Our theoretical results are applicable to a wide range of SV modes.

\subsection{Well-Posedness of Nonlocal BSDEs} \label{subsec:wellposedness}
In this subsection, we investigate the existence and uniqueness of nonlocal BSDEs (nonlocality in the parameter $\tau$) of the form: 
\begin{equation} \label{NonlocalBSDEDiff}
\left\{
    \begin{array}{rcl}
     dY^\tau_s&=&-h\left(\tau,s,Y^\tau_s,\int^T_s \phi(s,\tau) Y^\tau_sd\tau,Y^T_s,Z^\tau_s,\int^T_s\varphi(s,\tau) Z^\tau_sd\tau,Z^T_s\right)ds+Z^\tau_sdB_s, \\

     Y^\tau_\tau&=&\xi^\tau, \quad 0 \leq s \leq \tau \leq T, \quad \xi^\tau\in L^2_{\mathcal{F}_\tau}(\Omega;\mathbb{R}^d), 
    \end{array}
\right. 
\end{equation} 
where $\tau$ is considered as an external parameter which takes value in $[0,T]$. It is clear that \eqref{NonlocalBSDEDiff} is more general than \eqref{FBSDEsys}.

For each fixed $\tau$, the associated BSDE indexed by it has different terminal time $\tau$, different $\mathcal{F}_\tau$-terminal random variable $\xi^\tau$ at the endpoint, and different $\{\mathcal{F}_s\}_{s\in[0,\tau]}$-progressively measurable random function (denoted by $h^\tau$), 
\begin{equation*}
    h(\omega,\tau,s,y_1,y_2,y_3,z_1,z_2,z_3):\Omega\times\{\tau\}\times[0,\tau]\times\mathbb{R}^{3d}\times\mathbb{R}^{3d\times k}\rightarrow\mathbb{R}^{k}
\end{equation*}
as its generator. For simplicity, we often suppress the randomness $\omega$ and introduce $\bm{0}=(0,0,0,0,0,0)^\top$. In addition to their mutual entanglement between the two BSDEs associated $\tau$ and $T$, each BSDE parameterized by $\tau\in[0,T]$ is also coupled and interconnected through the two integral terms in the generator.
\begin{figure}[!ht]
	\centering
	\includegraphics[width=0.4\textwidth]{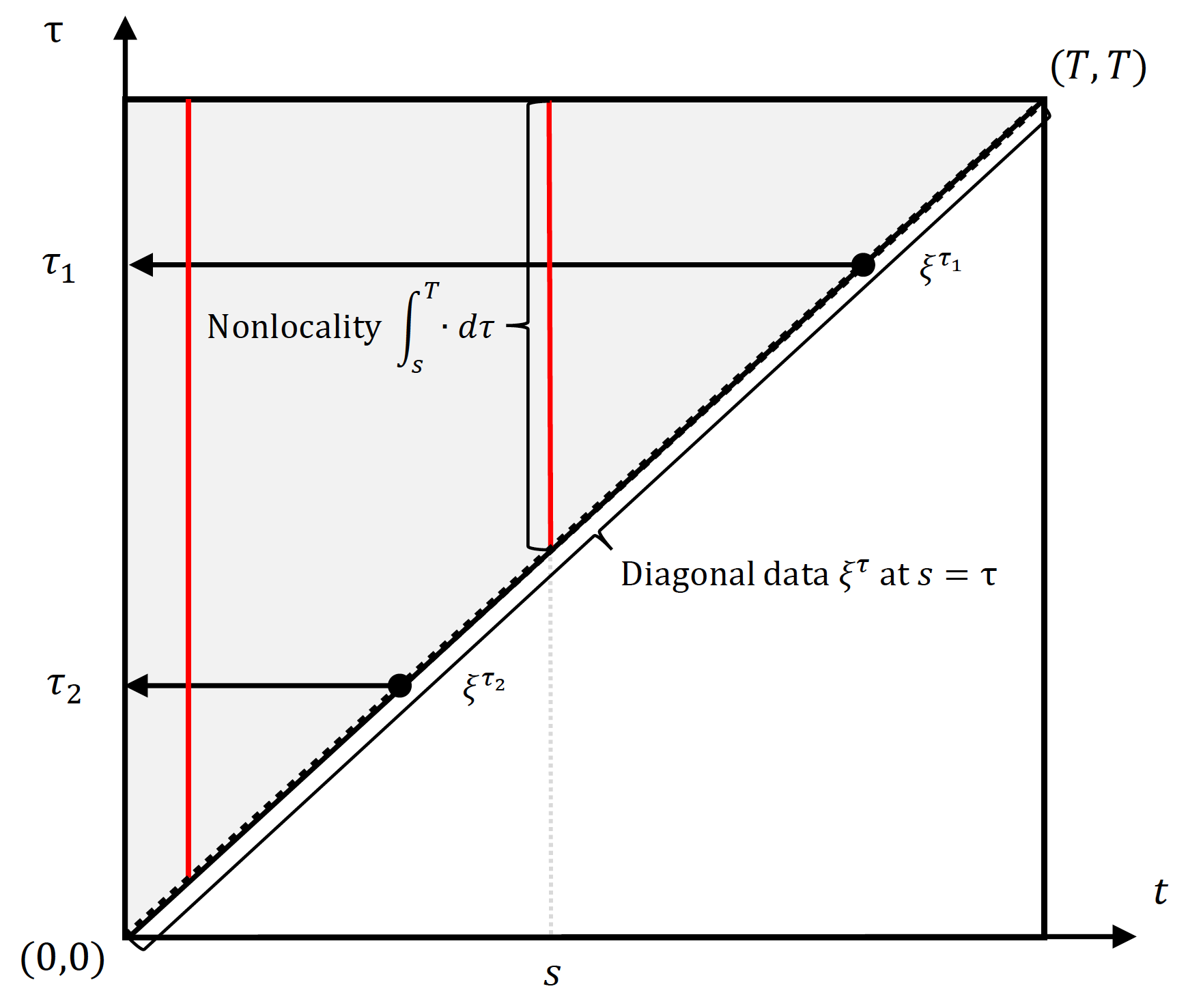}
	\caption{Nonlocal BSDE System}
	\label{fig:nBSDE}
\end{figure}

As shown in Figure \ref{fig:nBSDE}, the data on the diagonal $s=\tau$ of the nonlocal BSDE system \eqref{NonlocalBSDEDiff} has been provided. For any fixed $\tau\in[0,T]$, if the generator $h$ does not contain the tricky integral terms, the system can be simply regarded as a family of BSDEs parameterized by $\tau$. However, if nonlocal terms exist, when solving the BSDE backward from the terminal time $\tau$, the evolution of each time point $s$ depends on the unknown random field $\{(Y^\tau_s, Z^\tau_s)\}$ with $\tau\in[s,T]$, which is marked as a red line in Figure \ref{fig:nBSDE}. We aim to find a random field $\{(Y^\tau_\cdot,Z^\tau_\cdot)\}_{\tau\in[0,T]}$ solving \eqref{NonlocalBSDEDiff} in the sense of Definition \ref{nBSDESol}, which is only defined on the shaded triangular region in Figure \ref{fig:nBSDE}, instead of on the whole rectangle. To this end, we assume that $\phi$ and $\varphi$ are both real and bounded functions with $\max\big\{\sup_{0\leq s\leq \tau\leq T}|\phi(s,\tau)|,\sup_{0\leq s\leq \tau\leq T}|\varphi(s,\tau)|\big\}\leq c$. Furthermore, we assume that the generator $h^\tau$ ($0\leq \tau\leq T$) satisfies the following stochastic Lipschitz conditions.
\begin{definition}[Stochastic Lipschitz Conditions]
    For $0\leq s\leq \tau\leq T$, let $\beta>0$ be a fixed real number, $a_s:=u_s+v^2_s+l^2_s+k^2_s$, where $u_s$, $v_s$, $l_s$, and $k_s$ be given non-negative $\{\mathcal{F}_s\}$-progressively measurable process satisfying $\mathbb{E}\big[e^{\beta\int^T_0a_sds}\big]<\infty$. Furthermore, for $\tau$, $\tau^\prime\in[0,T]$, $y_i$, $y^\prime_i\in\mathbb{R}^d$, $z_i$, $z^\prime_i\in\mathbb{R}^{d\times k}$ $(i=1,2,3)$, and $s\in[0,\tau\wedge\tau^\prime]$, 

    \begin{equation} \label{LipschitzCons}
\left\{
    \begin{array}{rcl}
     |h(\tau,s,y_1,y_2,y_3,z_1,z_2,z_3)-h(\tau,s,y^\prime_1,y^\prime_2,y^\prime_3,z^\prime_1,z^\prime_2,z^\prime_3)|&\leq& u_s\sum\limits_i|y_i-y^\prime_i|+v_s\sum\limits_i|z_i-z^\prime_i|, \\
    |h(\tau,s,y_1,y_2,y_3,z_1,z_2,z_3)-h(\tau^\prime,s,y_1,y_2,y_3,z_1,z_2,z_3)|&\leq& \theta(|\tau-\tau^\prime|)\big(l_s+k_s\sum\limits_i|y_i|+\sum\limits_i|z_i|\big)
    \end{array}
\right. 
\end{equation} 
for a continuous and monotone increasing function $\theta:[0,\infty)\to[0,\infty)$ with $\theta(0)=0$.
\end{definition}
It is noteworthy that the generator $h^\tau$ of \eqref{NonlocalBSDEDiff} can be extended to a more general case with slight modifications, without increasing the proof’s complexity; for example, substituting nonlocal terms of \eqref{NonlocalBSDEDiff} with $\int^T_s \phi(s,\tau,Y^\tau_s)d\tau$ and $\int^T_s\varphi(s,\tau,Z^\tau_s)d\tau$.

\begin{definition} \label{nBSDESol}
A family of processes $\{(Y^\tau_\cdot,Z^\tau_\cdot)\}_{\tau\in[0,T]}\in\mathcal{M}^c(\beta,a.,T)$ (defined in detail later) is called an (adapted) solution of \eqref{NonlocalBSDEDiff} if the following holds: 
    \begin{equation} \label{NonlocalBSDEInt}
    \begin{split}
        Y^\tau_t=&\xi^\tau+\int^\tau_t h\bigg(\tau,s,Y^\tau_s,\int^T_s \phi(s,\tau) Y^\tau_sd\tau,Y^T_s,Z^\tau_s,\int^T_s\varphi(s,\tau) Z^\tau_sd\tau,Z^T_s\bigg)ds\\
        &-\int^\tau_t Z^\tau_sdB_s, \quad 0\leq t\leq \tau\leq T, \quad \mathbb{P}-a.s.
    \end{split}
\end{equation}
Equation \eqref{NonlocalBSDEInt} is said to have a unique adapted solution if for any two adapted solutions $\{(Y^\tau,Z^\tau)\}_{\tau\in[0,T]}$ and $\{(\widetilde{Y}^\tau,\widetilde{Z}^\tau)\}_{\tau\in[0,T]}$, it must hold that 
$$\mathbb{P}\big\{Y^\tau_s=\widetilde{Y}^\tau_s,\;\forall s\in[0,\tau]\;\mathrm{~and}\;Z^\tau_s=\widetilde{Z}^\tau_s,\;\mathrm{a.e.}\;s\in[0,\tau]\big\}=1, \quad \forall \tau\in[0,T]. $$ 
\end{definition}

\vspace{0.2cm}

Before studying nonlocal BSDEs \eqref{NonlocalBSDEDiff}, we first introduce some appropriate norms and Banach spaces such that 
\begin{itemize}
    \item The elements $(Y^\tau_s,Z^\tau_s)_{s\in[0,\tau]}$ do not exist in isolation from the family of processes $\{(Y^\tau_\cdot,Z^\tau_\cdot)\}_{\tau\in[0,T]}$ parameterized by $\tau$. It is required that the mapping from parameters to BSDEs, i.e. $[0,T]\ni\tau\mapsto(Y^\tau_s,Z^\tau_s)_{s\in[0,\tau]}\in\{(Y^\tau_\cdot,Z^\tau_\cdot)\}_{\tau\in[0,T]}$, has a certain degree of regularity in the direction of $\tau$.
    \item For each fixed $\tau\in[0,T]$, the pair of processes $(Y^\tau_s,Z^\tau_s)_{s\in[0,\tau]}$ needs to be regular enough in the temporal direction of $s$ to adapt the setting of stochastic Lipschitz condition. Moreover, the conditions need to be sufficiently weak to encompass a wide range of stochastic factor models as examples of applications in the theory of BSDEs.
    \item While it is standard to make use of Banach's fixed point theorem to prove the well-posedness of differential equations by establishing a contractive mapping, the constructed Banach spaces need to be closed with respect to this contraction in the sense that the mapping is a self-mapping.
\end{itemize}

In real-world financial applications, like the CKLS-type stochastic factor models discussed in the following section, the Lipschitz coefficients of \eqref{LipschitzCons} are not only stochastic but also lack a strictly defined positive lower bound and a finite upper bound uniformly for all arguments. In the existing literature, \cite{Bender2000BSDESWS,li2023bsdes,pardoux2014stochastic} have focused on relaxing the classical deterministic Lipschitz conditions in the standard setting of \cite{Peng1992,Pardoux1992,Peng2011}, established corresponding well-posedness results under stochastic conditions. One must strengthen other conditions while relaxing the Lipschitz continuity, rendering stronger integrability conditions on the driver as well as on the solutions. These integrability conditions make it possible to replace the uniform Lipschitz condition with a stochastic one. Inspired by earlier works of \cite{Bender2000BSDESWS,li2023bsdes,pardoux2014stochastic}, we propose some suitable norms and spaces to investigate our nonlocal BSDE system \eqref{NonlocalBSDEDiff}. In what follows, we list them in detail.
\begin{enumerate}
    \item $L^2(\tau;\mathbb{R}^d)$ is the set of all $\mathcal{F}_\tau$-measurable $\mathbb{R}^d$-valued random vectors $\xi$ satisfying $\|\xi\|^{2}:=\mathbb{E}\big[|\xi|^{2}\big]< \infty$. The family of random variables $\{\xi^\tau\}_{\tau\in[0,T]}$ parameterized by $\tau$ is regarded as a mapping from $\tau\in[0,T]$ to $\xi^\tau\in L^2(\tau;\mathbb{R}^d)$. Furthermore, the set $\mathbb{L}^2_T(\mathbb{R}^d)$ consists of all continuous $\mathbb{F}$-adapted processes $\{\xi^\tau\}_{\tau\in[0,T]}$ ($\{\xi^\tau\}_T$ for short) with $\|\xi^\tau\|_{c,T}^{2}:=\mathbb{E}\left[\operatorname*{sup}_{\tau\in[0,T]}|\xi^\tau|^{2}\right]<\infty.$ It is clear that the real-valued mapping $\tau\mapsto\|\xi^\tau\|$ is continuous and $\sup_{\tau\in[0,T]}\|\xi^\tau\|<\infty$ if $\{\xi^\tau\}_T\in\mathbb{L}^2_T(\mathbb{R}^d)$.  
    \item $S^2([0,\tau];\mathbb{R}^d)$ is the set of all $\{\mathcal{F}_s\}$-adapted, $\mathbb{R}^d$-valued, and continuous processes $(Y_s)_{s\in[0,\tau]}$ satisfying $\|Y_\cdot\|_{S^{2}_\tau}:=(\mathbb{E}[\operatorname*{sup}_{s\in[0,\tau]}|Y_{s}|^{2}])^{1/2}< \infty$. Next, the family of processes $\{Y^\tau_\cdot\}_{\tau\in[0,T]}$ parameterized by $\tau$ is considered as a mapping from $\tau\in[0,T]$ to $Y^\tau_\cdot\in S^2([0,\tau];\mathbb{R}^d)$. Moreover, $\{Y^\tau_\cdot\}_{\tau\in[0,T]}$ ($\{Y^\tau_\cdot\}_T$ for short) is said to be continuous with respect to $\tau$ if $\lim_{\tau^\prime\to\tau}\|Y^{\tau^\prime}_\cdot\|_{S^{2}_{\tau^\prime}}=\|Y^\tau_\cdot\|_{S^{2}_\tau}$. Furthermore, the set $C([0,T];S^2([0,\cdot];\mathbb{R}^d))$ consists of all continuous $\{Y^\tau_\cdot\}_T$ and 
    \begin{equation*}    
    \mathbb{S}^2_T(\mathbb{R}^d):=\left\{\left.\{Y^\tau_\cdot\}_T\in C([0,T];S^2([0,\cdot];\mathbb{R}^d))~\right|\sup_{\tau\in[0,T]}\|Y^\tau_\cdot\|_{S^2_\tau}<\infty\right\}. 
    \end{equation*}
    \item $M^2([0,\tau];\mathbb{R}^{d\times k})$ is the set of all $\{\mathcal{F}_s\}$-progressively measurable, $\mathbb{R}^{d\times k}$-valued processes $(Z_s)_{s\in[0,\tau]}$ such that $\|Z_\cdot\|_{M^{2}_\tau}:=\{\mathbb{E}[(\int_{0}^{\tau}|Z_{s}|^{2}ds)]\}^{1/2}< \infty$. Similarly, the set $C([0,T];M^2([0,\cdot];\mathbb{R}^d)$ consists of all processes $\{Z^\tau\}_{\tau\in[0,T]}$ with $\lim_{\tau^\prime\to\tau}\|Z^{\tau^\prime}_\cdot\|_{M^{2}_{\tau^\prime}}=\|Z^\tau_\cdot\|_{M^{2}_\tau}$ and  
    \begin{equation*}
        \mathbb{M}^2_T(\mathbb{R}^d):=\left\{\left.\{Z^\tau_\cdot\}_T\in C([0,T];M^2([0,\cdot];\mathbb{R}^d))~\right|\sup_{\tau\in[0,T]}\|Z^\tau_\cdot\|_{M^2_\tau}<\infty\right\}. 
    \end{equation*}
\end{enumerate}
Furthermore, we introduce some Banach spaces to adapt the generator of BSDEs with stochastic Lipschitz conditions. 
\begin{enumerate}\addtocounter{enumi}{3}
    \item $L^2(\beta,a_\cdot,\tau;\mathbb{R}^d)$ is the set of all $\mathcal{F}_\tau$-measurable $\mathbb{R}^d$-valued random vectors $\xi$ satisfying $\|\xi\|_{\beta,a_\cdot,\tau}^{2}:=\mathbb{E}\big[e^{\beta\int_{0}^{\tau}a_{r}dr}|\xi|^{2}\big]< \infty$. Moreover, $\mathbb{L}^2_T(\beta,a_\cdot;\mathbb{R}^d)$ consists of all continuous $\mathbb{F}$-adapted processes $\{\xi^\tau\}_T$ with $\|\xi^\tau\|_{c,\beta,a_\cdot,T}^{2}:=\mathbb{E}\big[\operatorname*{sup}_{\tau\in[0,T]}e^{\beta\int_{0}^{\tau}a_{r}dr}|\xi^\tau|^{2}\big]<\infty$. It indicates that $\tau\mapsto\|\xi^\tau\|_{\beta,a_\cdot,\tau}$ is continuous and $\sup_{\tau\in[0,T]}\|\xi^\tau\|_{\beta,a_\cdot,\tau}<\infty$.
    

    \item $L^2(\beta,a_\cdot,[0,\tau];\mathbb{R}^d)$ (resp. $L^2(\beta,a_\cdot,[0,\tau];\mathbb{R}^{d\times k})$) is the set of all $\{\mathcal{F}_s\}$-progressively measurable $\mathbb{R}^d$-valued (resp. $\mathbb{R}^{d\times k}$-valued) processes $(Z_s)_{s\in[0,\tau]}$ with $\|Z_\cdot\|_{\beta,a_\cdot,\tau}^{2}:=\mathbb{E}\big[\int_{0}^{\tau}e^{\beta\int_{0}^{s}a_{r}dr}|Z_{s}|^{2}\mathrm{d}s\big]<\infty$. Moreover, $C([0,T];L^2(\beta,a_\cdot,[0,\cdot];\mathbb{R}^d))$ is the set of all $\{Z^\tau_\cdot\}_T$ satisfying $\lim_{\tau^\prime\to\tau}\|Z^{\tau^\prime}_\cdot\|_{\beta,a_\cdot,\tau^\prime}=\|Z^\tau_\cdot\|_{\beta,a_\cdot,\tau}$ and 
    \begin{equation*}
        \mathbb{L}^2_T(\beta,a_\cdot,[0,\cdot];\mathbb{R}^d):=\left\{\left.\{Z^\tau_\cdot\}_T\in C([0,T];L^2(\beta,a_\cdot,[0,\cdot];\mathbb{R}^d))~\right|\sup_{\tau\in[0,T]}\|Z^\tau_\cdot\|^2_{\beta,a_\cdot,\tau}<\infty\right\}. 
    \end{equation*}
    
    \item $\overline{L}^2(\beta,a_\cdot,[0,\tau];\mathbb{R}^d)$ is the set of all $\{\mathcal{F}_s\}$-progressively measurable $\mathbb{R}^d$-valued processes $(Y_s)_{s\in[0,\tau]}$ with
    $\|\sqrt{a}. Y_\cdot\|_{\beta,a_\cdot,\tau}^{2}=\mathbb{E}\big[\int_{0}^{\tau}e^{\beta\int_{0}^{s}a_{r}dr}a_{s}|Y_{s}|^{2}ds\big]<\infty$. Moreover, $C([0,T];\overline{L}^2(\beta,a_\cdot,[0,\cdot];\mathbb{R}^d))$ is the set of all $\{Y^\tau_\cdot\}_T$ satisfying the condition $\lim_{\tau^\prime\to\tau}\|\sqrt{a}Y^{\tau^\prime}_\cdot\|_{\beta,a_\cdot,\tau^\prime}=\|\sqrt{a}Y^\tau_\cdot\|_{\beta,a_\cdot,\tau}$ and
    \begin{equation*}
        \overline{\mathbb{L}}^2_T(\beta,a_\cdot,[0,\cdot];\mathbb{R}^d):=\left\{\left.\{Y^\tau_\cdot\}_T\in C([0,T];\overline{L}^2(\beta,a_\cdot,[0,\cdot];\mathbb{R}^d))~\right|\sup_{\tau\in[0,T]}\|\sqrt{a}.Y^\tau_\cdot\|^2_{\beta,a_\cdot,\tau}<\infty\right\}. 
    \end{equation*}

    \item $H^2(\beta,a_\cdot,[0,\tau];\mathbb{R}^d)$ is the set of all $\{\mathcal{F}_s\}$-progressively measurable $\mathbb{R}^k$-valued processes $(h_s)_{s\in[0,\tau]}$ with $\|h_{\cdot}\|_{\beta,a_\cdot,H,\tau}^{2}:=\mathbb{E}\big[\big(\int_{0}^{\tau}e^{\frac{\beta}{2}\int_{0}^{s}a_{r}dr}|h_{s}|ds\big)^{2}\big]<\infty$. Then, 
    \begin{equation*}
    \begin{split}
        \mathbb{H}^{2}_T(\beta,a_\cdot,[0,\cdot];\mathbb{R}^d):=\bigg\{\{h^\tau_\cdot\}_T:[0,T]\ni\tau\mapsto (h^\tau_s)_{s\in[0,\tau]}\in H^2(\beta,a_\cdot,[0,\tau];\mathbb{R}^d)~~~\text{and}~ & \\ 
        \sup_{\tau\in[0,T]}\|h^\tau_\cdot\|^2_{\beta,a_\cdot,H,\tau}<\infty&\bigg\}. 
    \end{split}
    \end{equation*}
  
    \item $L^{2,c}(\beta,a_\cdot,[0,\tau];\mathbb{R}^d)$ is the set of all $\{\mathcal{F}_s\}$-adapted, $\mathbb{R}^d$-valued, and continuous processes $(Y_s)_{s\in[0,\tau]}$ with $\|Y_\cdot\|_{\beta,a_\cdot,c,\tau}^{2}:=\mathbb{E}\left[\operatorname*{sup}_{s\in[0,\tau]}\left(e^{\beta\int_{0}^{s}a_{r}dr}|Y_{s}|^{2}\right)\right]<\infty$. Moreover, $C([0,T];L^{2,c}(\beta,a_\cdot,[0,\cdot];\mathbb{R}^d))$ is the set of all $\{Y^\tau_\cdot\}_T$ satisfying $\lim_{\tau^\prime\to\tau}\|Y^{\tau^\prime}_\cdot\|_{\beta,a_\cdot,c,\tau^\prime}=\|Y^\tau_\cdot\|_{\beta,a_\cdot,c,\tau}$ and  

    \begin{equation*}
        \mathbb{L}^{2,c}_T(\beta,a_\cdot,[0,\cdot];\mathbb{R}^d):=\left\{\left.\{Y^\tau_\cdot\}_T\in C([0,T];L^{2,c}(\beta,a_\cdot,[0,\cdot];\mathbb{R}^d))~\right|\sup_{\tau\in[0,T]}\|Y^\tau_\cdot\|^2_{\beta,a_\cdot,c,\tau}<\infty\right\}. 
    \end{equation*}
    
\end{enumerate}

Based on (4)-(8), we further develop the following space.
\begin{enumerate}\addtocounter{enumi}{8}
    \item $\mathbb{M}_T(\beta,a_\cdot):=\overline{{\mathbb{L}}}^{2}_T(\beta,a_\cdot,[0,\cdot];\mathbb{R}^{d})\times \mathbb{L}^{2}_T(\beta,a_\cdot,[0,\cdot];\mathbb{R}^{d\times k})$ is a Banach space with the norm $$\|\{(Y^\tau_\cdot,Z^\tau_\cdot)\}\|_{\beta,a_\cdot}^{2}:=\sup\limits_{\tau\in[0,T]}\left\{\|\sqrt{a.}Y^\tau_\cdot\|_{\beta,a_\cdot,\tau}^{2}+\|Z^\tau_\cdot\|_{\beta,a_\cdot,\tau}^{2}\right\}.$$ 
    \item $\mathbb{M}^{c}_T(\beta,a_\cdot):=\big(\overline{{\mathbb{L}}}^{2}_T(\beta,a_\cdot,[0,\cdot];\mathbb{R}^{d})\cap \mathbb{L}^{2,c}_T(\beta,a_\cdot,[0,\cdot];\mathbb{R}^{d})\big)\times \mathbb{L}^{2}_T(\beta,a_\cdot,[0,\cdot];\mathbb{R}^{d\times k})$ is a subspace of $\mathbb{M}_T(\beta,a_\cdot)$ and with 
    $$\|\{(Y^\tau_\cdot,Z^\tau_\cdot)\}\|_{\beta,a_\cdot,c}^{2}:=\sup\limits_{\tau\in[0,T]}\left\{\|Y^\tau_\cdot\|_{\beta,a_\cdot,c,\tau}^{2}+\|{\sqrt{a}}.Y^\tau_\cdot\|_{\beta,a_\cdot,\tau}^{2}+\|Z^\tau_\cdot\|_{\beta,a_\cdot,\tau}^{2}\right\}.$$
\end{enumerate}

We first present the solvability of a simplified BSDE \eqref{SimpleSDEInt}, which provides a solid foundation for our analysis of nonlinear BSDEs via a fixed-point argument.
\begin{lemma} \label{WellposednessSimpleBSDE}
Let $\beta>0$, terminal data $\{\xi^\tau\}_T\in\mathbb{L}^2_T(\beta,a_\cdot;\mathbb{R}^d)$, and the generators $\{h^\tau(s,\bm{0})\}_T\in \mathbb{H}^{2}_T(\beta,a_\cdot,[0,\cdot];\mathbb{R}^d)$ satisfy uniformly the stochastic Lipschitz conditions \eqref{LipschitzCons}. If $\{(\bar{Y}^\tau_\cdot,\bar{Z}^\tau_\cdot)\}_{\tau\in[0,T]}\in\mathbb{M}_T(\beta,a_\cdot)$, then  
\begin{equation} \label{SimpleSDEInt}
\begin{split}
    Y^\tau_t=&\xi^\tau+\int^\tau_t h\bigg(\tau,s,\bar{Y}^\tau_s,\int^T_s \phi(s,\tau)\bar{Y}^\tau_sd\tau,\bar{Y}^T_s,\bar{Z}^\tau_s,\int^T_s\varphi(s,\tau)\bar{Z}^\tau_sd\tau,\bar{Z}^T_s\bigg)ds\\
    &-\int^\tau_s Z^\tau_sdB_s, 
\end{split}
\end{equation}
admits a unique solution $\{(Y^\tau_\cdot,Z^\tau_\cdot)\}_{\tau\in[0,T]}\in\mathbb{M}^c_T(\beta,a_\cdot)$. 
\end{lemma}

Furthermore, one can take advantage of \eqref{SimpleSDEInt} to establish a mapping $\Lambda(\{(\bar{Y}^\tau_\cdot,\bar{Z}^\tau_\cdot)\}_T)=\{(Y^\tau_\cdot,Z^\tau_\cdot)\}_T$ and prove that it is a contraction over the Banach space $\mathbb{M}_T(\beta,a_\cdot)$ for some suitable $\beta>0$. Consequently, the unique fixed point of $\Lambda$ is exactly the solution of our nonlocal BSDE \eqref{NonlocalBSDEDiff}. 

\begin{theorem} \label{WellposednessBSDE}
    Suppose that 
    \begin{equation} \label{Conbeta}
        \frac{12\max\{c^2T^2,1\}}{\beta}+\frac{24\max\{c^2T^2,1\}}{\beta^2}<1, 
    \end{equation}
    terminal data $\{\xi^\tau\}_T\in\mathbb{L}^2_T(\beta,a_\cdot;\mathbb{R}^d)$, and the generators $\{h^\tau(s,\bm{0})\}_T\in \mathbb{H}^{2}_T(\beta,a_\cdot,$ $[0,\cdot];\mathbb{R}^d)$ satisfy uniformly the stochastic Lipschitz conditions \eqref{LipschitzCons} Then, BSDE \eqref{NonlocalBSDEDiff} admits a unique solution $\{(Y^\tau_\cdot,Z^\tau_\cdot)\}_{\tau\in[0,T]}\in\mathbb{M}^c_T(\beta,a_\cdot)$. 
\end{theorem}

Given a nonlocal BSDE \eqref{NonlocalBSDEDiff} with fixed $c$ and $T$, it is always available to choose a large enough $\beta$ such that the condition \eqref{Conbeta} holds such that $\Lambda$ is a contractive mapping. In the next section, we will apply the well-posedness result obtained in Theorem \ref{WellposednessBSDE} to analyze the solvability of MV problem \eqref{equiMV}-\eqref{Completedynamics} and well-definedness of our equilibrium investment policy \eqref{EPincom} within various models with stochastic investment opportunities and potentially incomplete markets.

\subsection{MV Problems with CKLS-SV Models} \label{subsec:apps}
In this subsection, we consider the MV asset allocation problem with various stochastic investment opportunity sets. We provide explicit closed-form solutions to equilibrium investment policies under appropriate conditions. We now consider an incomplete-market setting in which the stock price follows
a Chan–Karolyi–Longstaff–Sanders (CKLS)-type stochastic volatility (SV) model of the form
\begin{equation} \label{SVMs}
\left\{
    \begin{array}{l}
     \displaystyle{\frac{dS_{s}}{S_{s}}=(r+\delta R_{s}^{\frac{1+2\kappa\alpha}{2\alpha}})ds+R_{s}^{\frac{1}{2\alpha}}dB_s}, \\ 
     \displaystyle{dR_s=\left(a-bR_s\right)ds+\sigma R^p_sdB^R_s},
    \end{array}
\right. 
\end{equation} 
where $\kappa:=1-p$, $\alpha\neq 0$, $\delta\in\mathbb{R}$, and $\{R_\cdot\}$ follows a CKLS process, which is a generalization of the Cox–Ingersoll–Ross (CIR) and the Ornstein–Uhlenbeck (OU) processes corresponding to the two special cases of $p=\frac{1}{2}$ and $p=0$, respectively. The CKLS process provides a flexible framework for modeling SV in financial markets. Except for specific cases, such as the CIR and OU case, closed-form solutions for the Fourier transform or density function are not available for the CKLS process. This limits the ability to derive analytical results and makes analytic methods in \cite{Basak2010} unavailable. Despite its lack of analytical tractability, the CKLS process has shown good empirical performance in fitting historical financial data. It has been widely used in empirical studies to analyze and forecast asset prices and volatility. Using equation \eqref{SVMs}, it is clear that the instantaneous Sharpe ratio can be expressed as $\delta R^\kappa_s$. Moreover, if letting $p=1/2$, then the case of $\alpha=-1$ corresponds to the SV model employed by \cite{chacko2005dynamic} and the case of $\alpha=1$ reduces to the Heston model in \cite{heston1993closed}. Furthermore, if $p=0$ and $\alpha\uparrow\infty$, \eqref{SVMs} is specified in the OU example of \cite{Basak2010}. 

By our probabilistic approach, we will show that the equilibrium investment policy \eqref{EPincom} is well-defined for a wide range of CKLS-type SV models (i.e., a range of parameter $p$). 
Under the model setting of \eqref{SVMs}, the nonlocal BSDE \eqref{FBSDEsys} becomes
\begin{equation} \label{CKLSSys}
\left\{
    \begin{array}{rcl}
     dY^\tau_s&=&\displaystyle{\left(-\frac{\delta^2\rho(s)}{\gamma}R^{2\kappa}_s+\varrho\delta R^{\kappa}_s\frac{\int^T_s\lambda(s,\tau)Z^\tau_sd\tau+Z^T_s}{1+\int^T_s\lambda(s,\tau)d\tau}\right)ds+Z^\tau_s dB^R_s,} \\

     Y^\tau_\tau&=&0, \quad 0\leq s \leq \tau \leq T. 
    \end{array}
\right. 
\end{equation} 
It is noteworthy that the SV $\{R_\cdot\}$ of \eqref{SVMs} and the nonlocal BSDE $\{(Y^\tau_\cdot,Z^\tau_\cdot)\}_{\tau\in[0,T]}$ \eqref{CKLSSys} are driven by the same Brownian motion $B^R$. It promises that the generator and stochastic Lipschitz conditions are all adapted to the same filtration. 


     



        


\subsubsection*{I. The Degenerated Case when $\bm{p=1}$.}
First of all, we consider a degenerated case of \eqref{SVMs}, where $p=1$, $a\geq 0$, $b<0$, $R_0>0$, $\sigma>0$, and $\kappa=0$. Then the BSDE of \eqref{CKLSSys} is equipped with a deterministic generator with a standard Lipschitz condition 
\begin{equation*}
    \{h^\tau(s,\bm{0})\}_T=-\frac{\delta^2\rho(s)}{\gamma}, \quad a_s=\varrho^2\delta^2.
\end{equation*} 
Consequently, under the condition of $\int^T_0\rho(s)ds<\infty$, our well-posedness result in Theorem \ref{WellposednessBSDE} promises that \eqref{CKLSSys} admits a unique solution $\{(Y^\tau_\cdot,Z^\tau_\cdot)\}_{\tau\in[0,T]}$, thereby ensuring that the equilibrium investment policy \eqref{EPincom} is well-defined; see Theorem \ref{AppThm} below for further details.  


\subsubsection*{II. General CKLS Case with $\bm{p\in(0,1)}$.}
Next, we turn to investigate the case, where $p\in(0,1)$ while the generator of \eqref{CKLSSys} satisfies a stochastic Lipschitz condition \eqref{LipschitzCons}. Suppose that $b\in\mathbb{R}$, $a$, $R_0$, $\sigma>0$, it is clear that 
\begin{equation*}
    \{h^\tau(s,\bm{0})\}_T=-\frac{\delta^2\rho(s)}{\gamma}R^{2\kappa}_s, \quad a_s=\frac{\varrho^2\delta^2}{\big(1+\int^T_s\lambda(s,\tau)d\tau\big)^2} R^{2\kappa}_s.
\end{equation*} 

In order to apply our well-posedness results for nonlocal BSDE \eqref{NonlocalBSDEDiff} to analyze \eqref{CKLSSys}, we need to verify the condition $\{h^\tau(s,\bm{0})\}_T\in \mathbb{H}^{2}_T(\beta,a_\cdot,[0,\cdot];\mathbb{R}^d)$, which is equivalent to show that
\begin{equation} \label{Conh1}
    \mathbb{E}\left[\left(\int_{0}^{T}e^{\frac{\beta}{2}\int_{0}^{s}\frac{\varrho^2\delta^2}{\left(1+\int^T_r\lambda(r,\tau)d\tau\right)^2} R^{2\kappa}_rdr}\left(\frac{\beta}{2}\frac{\varrho^2\delta^2}{\big(1+\int^T_s\lambda(s,\tau)d\tau\big)^2} R^{2\kappa}_s\right)ds\right)^{2}\right]<\infty.
\end{equation}
To see this, we can perform a straightforward analysis if the MGF of $R^{2\kappa}_s$ exists for all $s\in[0,T]$:
\begin{equation} \label{Conh2}
    \begin{split}
        & \mathbb{E}\left[\left(\int_{0}^{T}e^{\frac{\beta}{2}\int_{0}^{s}\frac{\varrho^2\delta^2}{\left(1+\int^T_r\lambda(r,\tau)d\tau\right)^2} R^{2\kappa}_rdr}\left(\frac{\beta}{2}\frac{\varrho^2\delta^2}{\big(1+\int^T_s\lambda(s,\tau)d\tau\big)^2} R^{2\kappa}_s\right)ds\right)^{2}\right]\\
        =&\mathbb{E}\left[\left(e^{\frac{\beta}{2}\int_{0}^{T}\frac{\varrho^2\delta^2}{\left(1+\int^T_s\lambda(s,\tau)d\tau\right)^2} R^{2\kappa}_sds}-1\right)^{2}\right] \\
        \leq & \mathbb{E}\left[e^{\int_{0}^{T}\frac{\beta\varrho^2\delta^2}{\left(1+\int^T_s\lambda(s,\tau)d\tau\right)^2} R^{2\kappa}_sds}\right] \\
        \leq &\frac{1}{T}\int^T_0\mathbb{E}\left[e^{\frac{\beta\varrho^2\delta^2T}{\left(1+\int^T_s\lambda(s,\tau)d\tau\right)^2} R^{2\kappa}_s}\right]ds\leq \frac{1}{T}\int^T_0\mathbb{E}\left[e^{\beta\varrho^2\delta^2T R^{2\kappa}_s}\right]ds< \infty.
    \end{split}
\end{equation}
It is well-known that the MGF of the CKLS process is finite only on part of the real line; see \cite{aly2015moment}. It means that there exists a critical point $\mu^*_s>0$ (possibly depends upon $s$) such that $\mathbb{E}\left[e^{\mu R^{2\kappa}_s}\right]$ is finite in $\mu\in(0,\mu^*_s)$ and explodes at $\mu^*_s$. Since the real function $s\mapsto\mu^*_s$ is decreasing, it is necessary to set the condition $\beta\rho^2\delta^2T<\mu^*_T$. We will provide more details about it in Theorem \ref{AppThm}. 

\vspace{0.4cm} 

\noindent\textbf{CIR process with $\bm{p=\frac{1}{2}}$.} Remarkably, it is interesting to consider a special case of $p=\frac{1}{2}$, which corresponds to the CIR process. In this case, the stochastic process $\{R_\cdot\}$ admits an explicit formula of future distribution. Indeed, let us consider 
\begin{equation*}
        \mathbb{E}\left[e^{\beta\varrho^2\delta^2 T R_s}\right]=\mathbb{E}\left[e^{\beta\varrho^2\delta^2 T \frac{\chi_s}{2C_s}}\right] \quad \text{with} \quad \kappa=p=\frac{1}{2},\quad C_s=\frac{2b}{\left(1-e^{-bs}\right)\sigma^2}
    \end{equation*}
where $\chi_s$ is the non-central chi-squared distribution with $4a/\sigma^2$ degrees of freedom and non-centrality parameter $2C_sR_0e^{-bs}$. Considering the critical moment $\mu^*_s$ of MGF of the non-central chi-squared distribution $\chi_s$ equals to $\frac{1}{2}$, it is required that $\beta\varrho^2\delta^2T/(2C_s)<1/2$. Consequently, if  
    \begin{equation} \label{ConCIR}
        \beta\varrho^2\delta^2T<\frac{2b}{\left(1-e^{-bT}\right)\sigma^2}. 
    \end{equation}
then $\mathbb{E}\left[e^{\beta\varrho^2\delta^2 T R_s}\right]<\infty$ for all $s\in[0,T]$. Moreover, since the conditional expectation is continuous with respect to $s\in[0,T]$, the condition \eqref{ConCIR} is sufficient enough to promise the integral of \eqref{Conh2} is well-defined (integrable).  

Next, let us offer a counterexample to illustrate that a condition akin to \eqref{ConCIR} is essential for the analysis in \eqref{Conh2}. In fact, as \cite{Korn2004} indicates, one can find $\mathbb{E}\left[e^{\beta\varrho^2\delta^2 T R_s}\right]=\infty$ if $b>0$ and 
\begin{equation} \label{counterexample}
        \beta\varrho^2\delta^2T\geq \frac{2b^3T^2e^{bT}}{\sigma^2[2e^{bT}-(1+bT)^2-1]}.   
    \end{equation}
It is obvious that $\frac{2b}{\left(1-e^{-bT}\right)\sigma^2}<\frac{2b^3T^2e^{bT}}{\sigma^2[2e^{bT}-(1+bT)^2-1]}$ for any $b$, $\sigma$, $T>0$. This shows that \eqref{ConCIR} is robust enough to avoid such constructed counterexamples in \cite{Korn2004}. However, it does not mean that \eqref{ConCIR} cannot be further improved. In fact, to utilize our well-posedness result in Theorem \ref{WellposednessBSDE}, we only need to verify \eqref{Conh2}. Ensuring this verification to its fullest extent lies beyond the scope of our study.

\subsubsection*{III. The OU Case when $\bm{p=0}$.} 
Finally, letting $\alpha\to\infty$, we essentially consider a mean-reverting OU process for \eqref{SVMs} and the risky asset price dynamics becomes
\begin{equation*} \label{OUprocess}
\left\{
    \begin{array}{rcl}
        dS_{s}/S_{s}&=&(r+\delta R_{s})ds+ dB_s, \\ 
        dR_s&=&\left(a-bR_s\right)ds+\sigma dB^R_s,
    \end{array}
\right. 
\end{equation*}
then it is clear to identify that 
\begin{equation*}
    \{h^\tau(s,\bm{0})\}_T=-\frac{\delta^2\rho(s)}{\gamma}R^{2}_s, \quad a_s=\frac{\varrho^2\delta^2}{\big(1+\int^T_s\lambda(s,\tau)d\tau\big)^2} R^{2}_s
\end{equation*} 

Similar to the calculation in \eqref{Conh1}-\eqref{Conh2}, one needs to consider the MGF of $\{R^2_\cdot\}$. Since the OU process is a Gaussian process, it is clear that each $R_s$ is normally distributed with expectation $R_0e^{-bs}+\frac{a}{b}(1-e^{-bs})$ and variance $\frac{\sigma^2}{2b}(1-e^{-2bs})$. Hence, the MGF of $R^2_s$ satisfies
\begin{equation} \label{ConExpOU}
    \mathbb{E}\left[e^{\mu R^2_s}\right]=\frac{1}{\sqrt{1-2\mu\mathbb{V}[R_s]}}e^{\frac{\mu\mathbb{E}^2[R_s]}{1-2\mu\mathbb{V}[R_s]}}=\frac{1}{\sqrt{1-2\mu\frac{\sigma^2}{2b}(1-e^{-2bs})}}e^{\frac{\mu(R_0e^{-bs}+\frac{a}{b}(1-e^{-bs}))^2}{1-2\mu\frac{\sigma^2}{2b}(1-e^{-2bs})}}
\end{equation}
For convergence, $\mu\frac{\sigma^2}{2b}(1-e^{-2bs})<\frac{1}{2}$. Combined with \eqref{Conh2}, let $\mu=\beta\varrho^2\delta^2T$, it is required that 
\begin{equation} \label{ConOU}
    \beta\varrho^2\delta^2T<\frac{b}{\left(1-e^{-2bT}\right)\sigma^2}
\end{equation}
such that $\mathbb{E}\left[e^{\beta\varrho^2\delta^2 T R^2_s}\right]<\infty$ for all $s\in[0,T]$. Similarly, due to the continuity of \eqref{ConExpOU} in $s$, one can find that the condition \eqref{ConOU} is good enough to ensure \eqref{Conh2}.

\subsubsection*{The Equilibrium MV Policy under CKLS-SV Models} 
Finally, let us summarize the three cases above and end this section with the following theorem.
\begin{theorem} \label{AppThm}
    Let $p\in[0,1]$. If $\beta>0$ satisfies \eqref{Conbeta} and 
    \begin{equation} \label{GeneralCon}
        \beta\varrho^2\delta^2 T<\frac{b}{\sigma^2\kappa(1-e^{-2b\kappa T})}, 
    \end{equation}
holds, then \eqref{CKLSSys} admits a unique solution $\{(Y^\tau_\cdot,Z^\tau_\cdot)\}_{\tau\in[0,T]}\in\mathbb{M}^c_T(\beta,a_\cdot)$. Furthermore, the equilibrium MV investment policy for the MV objectives \eqref{equiMV} under the CKLS-SV model
\eqref{SVMs} is given by 
\begin{equation} \label{CKLSpolicy}
    \begin{split}
        \widehat{u}(s,R_s) = \frac{\delta\rho(s)}{\gamma}R^{\kappa-\frac{1}{2\alpha}}_s e^{-r_0(T-s)}-\varrho R^{-\frac{1}{2\alpha}}_s\frac{\int^T_s\lambda(s,\tau)Z^\tau_sd\tau+Z^T_s}{1+\int^T_s\lambda(s,\tau)d\tau}e^{-r_0(T-s)}.
    \end{split}
\end{equation}
\end{theorem}

The inequality in \eqref{GeneralCon} from \cite{aly2015moment} prevents a moment explosion of the CKLS process. Given our well-posedness result in Theorem \ref{WellposednessBSDE}, proving the statements in Theorem \ref{AppThm} becomes straightforward. Next, let us make some interesting observations for the equilibrium MV policy \eqref{CKLSpolicy} under the CKLS models: (I) the general condition of \eqref{GeneralCon} for $p\in[0,1]$ is consistent with the special cases of the CIR model \eqref{ConCIR} with $p=\frac{1}{2}$ and the OU process \eqref{ConOU} with $p=1$; (II) the value on the right hand side of \eqref{GeneralCon} can be infinite as $p\to 1$ (i.e. $\kappa\to 0$), but this still makes sense, indicating that there are no restrictions on the parameters on the left hand side of the condition \eqref{GeneralCon}. This flexibility is due to our analysis of this specific example (Case $p=1$), which revealed that the corresponding BSDE reduces to a standard Lipschitz condition, as opposed to a stochastic Lipschitz condition. Thus, there is no need to limit the parameter ranges to maintain the validity of its MGF function. The first component of \eqref{CKLSpolicy} arises from the myopic demand, while the second stems from the hedging demand. The correlation parameter $\varrho$ captures the extent of market incompleteness in the economy. This effect is most pronounced in complete markets, where $\varrho = \pm 1$, and disappears with zero correlation, $\varrho = 0$. Notably, the parameter $\rho(s)$, which balances risk and return in real time, directly impacts the myopic component and indirectly affects the hedging term via the FBSDE system \eqref{CKLSSys}. In contrast, the weighting parameter $\lambda(s,\tau)$, which balances past and future considerations, exclusively influences the hedging component.

\section{Numerical Studies} \label{sec:Num}
In this section, we numerically illustrate our theoretical framework, especially for solving the nonlocal BSDE \eqref{CKLSSys} and implementing the equilibrium MV policies of the form \eqref{CKLSpolicy}. To this end, we first develop a numerical scheme for \eqref{CKLSSys}. Thanks to the well-posedness results we established in Theorem \ref{AppThm} (or Theorem \ref{WellposednessBSDE}), the convergence of the numerical scheme to the unique BSDE solution is guaranteed. We then compare the performance between different equilibrium policies under different models. The assumptions in Theorem \ref{AppThm} have been verified for all problem setups in this section.

\subsection{Numerical Scheme}
The solution to the BSDE (\ref{CKLSSys}) can be approximated with a direct implementation of a least-squares-regression-based backward Euler Scheme as in \cite{gobet_regression-based_2005} with some modifications; see Algorithm \ref{alg:0} below. The horizon $T$ is discretized into $N+1$ time points with equal interval $\Delta_t = \frac{T}{N}$.
For each time point $t$ and index $\tau$, we take $\mathbb{R}^K$-valued determinitic function bases $q^{\tau}_{t}$, whose elements are given by a sequence of Laguerre polynomials of size $K$. The solution of the BSDE at $(\tau, t)$ is then approximated by $\hat{Y}^{\tau, N, K, M}_t(x) \coloneqq p^{\tau, Y}_{t} \cdot q^{\tau}_{t}(x)$ and $\hat{Z}^{\tau, N, K, M}_t(x) \coloneqq p^{\tau, Z}_{t} \cdot q^{\tau}_{t}(x)$. 

\begin{algorithm}[!ht]
\caption{Least-Squares-Regression-based Backward Euler Method}
\label{alg:0}
\begin{algorithmic}[1] 
\STATE {\bfseries Input:} Number of trajectories, M, number of Picard iterations, I.
\STATE {\bfseries Output:} Solution to the BSDE:$Y$, $Z$.
\STATE Initialize coefficients $p^{\cdot, \cdot}_{\cdot, \cdot}$;
\STATE Generate $M$ sample CKLS process trajectories, $\{R_i\}_{i = 1, 2, \cdots, M}$ with Wiener process, $\{B_i\}_{i = 1, 2, \cdots, M}$;
\FOR {$t \leftarrow N-1, N-2, \cdots, 0$}
\FOR {$\tau \leftarrow N-1, N-2, \cdots, t+1$}
\FOR {$k \leftarrow 1, 2, \cdots, I$}
\STATE 
\begin{align*}
    p^{\tau, Y}_{t, k}, p^{\tau, Z}_{t, k} \leftarrow 
\mathop{\arg\min}\limits_{p^Y, p^Z} \frac{1}{M} \sum_{m=1}^{M} [&\hat{Y}^{\tau, N, K, M}_{t+1} - p^{\tau, Y}_{t} \cdot q^{\tau}_{t}(R_{m, t\Delta_t}) + \Delta_t h \\
   &- p^{\tau, Z}_{t} \cdot q^{\tau}_{t}(R_{m, t\Delta_t}) dB_{m, t\Delta_t}]^2, \text{ where}
\end{align*} 
\begin{align*}
h = h(&\tau, t, 
p^{\tau, Y}_{t, k-1} \cdot q^{\tau}_{t},
\sum_{n = t}^{N} \phi(t, n) \hat{Y}^{n, N, K, M}_t,
\hat{Y}^{\tau, N, K, M}_N,\\
&p^{\tau, Z}_{t, k-1} \cdot q^{\tau}_{t},
\sum_{n = t}^{N} \varphi(t, n) \hat{Z}^{n, N, K, M}_t,
\hat{Z}^{\tau, N, K, M}_N
)    
\end{align*} \label{loss_in_algo}
\ENDFOR
\STATE $\hat{Y}^{\tau, N, K, M}_t \leftarrow (p^{\tau, Y}_{t, I} \cdot q^{\tau}_{t})$
\STATE $\hat{Z}^{\tau, N, K, M}_t \leftarrow (p^{\tau, Z}_{t, I} \cdot q^{\tau}_{t})$
\ENDFOR
\ENDFOR
\State Return $\hat{Y}^{\cdot, N, K, M}_\cdot$ and $\hat{Z}^{\cdot, N, K, M}_\cdot$
\end{algorithmic}
\end{algorithm}

Loosely speaking, the algorithm (\ref{alg:0}) can be viewed as executing the algorithm in \cite{gobet_regression-based_2005} for each $\tau$ simultaneously. With reference to Figure \ref{fig:nBSDE}, the estimation starts from the top-right tip of the shaded triangle and iterates from right to left ($t = N, \cdots, 0$) and for each time $t$, from top to bottom ($\tau = N, \cdots, t$). The validity of the fixed-point argument at each $(t, \tau)$ has been proven in \cite{gobet_regression-based_2005}.
Besides, the loss function in line (\ref{loss_in_algo}) can be read as the mean-squared-error of the difference between the estimated $Y$ at next time step and the current estimated $Y$ plus current estimated $dY$. The function $h$ in the algorithm is given by (\ref{NonlocalBSDEDiff}). Note that, for simplicity in presentation, the dependence of $(Y,Z)$ estimates on $R_m$ is omitted.

\subsection{Numerical Result} \label{Sect_NR}
Three sets of experiments were conducted to illustrate how our result can be applied to solve for the equilibrium MV policy in an incomplete market setting, where the stock price follows the CKLS-SV model of (\ref{SVMs}). We consider three problems, where the simulated data come from the CKLS-SV model of (\ref{SVMs}) with the SV process $R_s$ being the CIR process (\textbf{Problem A}), the OU process (\textbf{Problem B}), and the CKLS process with $p$ equals to values other than 0.5 (\textbf{Problem C}), respectively. In the first set of experiments (\textbf{Problems A} and \textbf{B}), we adopt the equilibrium MV policies under the corresponding true SV models, while the objective is to show that our nonlocal BSDE approach yields consistent results with that of the existing literature \cite{Basak2010} under the special cases. In the second set of experiments (\textbf{Problem C}), we evaluate the equilibrium MV policies under different SV models, whose parameters are estimated with the simulated data, to illustrate the impact of choice of market model on the equilibrium MV policies. Finally, we examine the change of the discounting factor in the objective from constant $1$ to exponential discounting to show that the hedging demand in \textbf{Problem B} is dependent on discounting factor.


\subsubsection{Consistency with Analytical Solution}
The market environment specifications in both \textbf{Problems A} and \textbf{B} are the same as the ones in \cite{Basak2010} and we set $N=10$. Specifically, with objective function, 
\begin{equation}\label{numericalObj_w/oDiscount}
    J(s,w;u)=\mathbb{E}^{\mathcal{F}_s}\left[W_T-2W_T^2\right]+2\left(\mathbb{E}^{\mathcal{F}_s}[W_T]\right)^2,
\end{equation}
where $W_T$ denotes the terminal wealth at the end of horizon, \textbf{Problems A} is to solve for the equilibrium MV policies under the SV model whose dynamics is described by the following,
\begin{equation*}
\left\{
    \begin{array}{l}
     \displaystyle{\frac{dS_s}{S_s} = (0.03 + 0.0811)ds  + R_s^{-\frac{1}{2}}dB_s}, \\ 
     \displaystyle{dR_s = (9.4251-0.3374 R_s)ds + 0.6503 R_s^{\frac{1}{2}}dB^R_s},
    \end{array}
\right. 
\end{equation*} 
and \textbf{Problems B} is to solve for that under the SV model with the following dynamics,
\begin{equation*}
\left\{
    \begin{array}{l}
     \displaystyle{\frac{dS_s}{S_s} = (0.0014 + \delta R_s)ds  + 1dB_s}, \\ 
     \displaystyle{dR_s = (0.021276-0.27R_s)ds + 0.065 dB^R_s}.
    \end{array}
\right. 
\end{equation*}
These two dynamics are obtained by substituting $r= 0.03, \kappa = 0.5, \delta = 0.0811, \alpha = -1$ (\textbf{Problems A}) and $r= 0.0014, \kappa = 1, \delta = 1, \alpha = 1$ (\textbf{Problems B}) into (\ref{SVMs}) in Section \ref{subsec:apps}, while the objective function is obtained by substituting $\gamma = 4, C(\cdot) = H(\cdot) = 0, F(s,y,x) = x, G(s,y,x) = x^2$ into the general TIC objective function (\ref{GeneralTIC}) in Section (\ref
{sec:tic}).

Algorithm \ref{alg:0} is run to solve the BSDE (\ref{CKLSSys}) corresponding to \textbf{Problems A} and \textbf{B} and the numerical solutions (``Ours") obtained were compared with the analytical solution \eqref{simpleSDU} given in \cite{Basak2010}. Figures \ref{fig:CIR_Numerical} and \ref{fig:OU_Numerical} illustrate two forms of comparison performed: a) compare the control given by each policies for a arbitrary SV process trajectory and b) estimate and compare the conditional expected values of the objective, $\mathbb{E}_s[J(s, 1; u)]$, for $s = 0, 0.1, \cdots, 1$. The conditional expected values of objective are estimated from $M = 10000$ simulated wealth trajectories $\{W^i_t\}_{i = 1,\cdots, M}$ with the following expression:
\begin{equation*}
    \mathbb{E}_s[J(s, 1; u)] \approx \frac{1}{M}\sum_{i = 1}^M \frac{W^i_T}{W^i_s} - \frac{\gamma}{2M}\sum_{i = 1}^M \left( \frac{W^i_T}{W^i_s} - \frac{1}{M}\sum_{j = 1}^M \frac{W^j_T}{W^j_s} \right) ^ 2.
\end{equation*}
Note that such an approximation is only valid when the dynamics of the SV models are Markovian and independent of the state variable.

 A slight deviation in the control for the early time interval ($t < 0.3$) is observed in Figure (\ref{fig:CIR_Numerical}) (a) and (\ref{fig:OU_Numerical}) (a), but it is expected due to the accumulated discretization error from the backward solving procedure for the BSDE. However, the two curves of conditional expectations are completely overlapped, which suggests that such a tiny deviation in control has no significant impact on the overall performance of the policy across the horizon. The results confirm the accuracy of the algorithm.

\begin{figure}[h]
    \centering
    \subfigure[]{
    \includegraphics[width=0.4\textwidth]{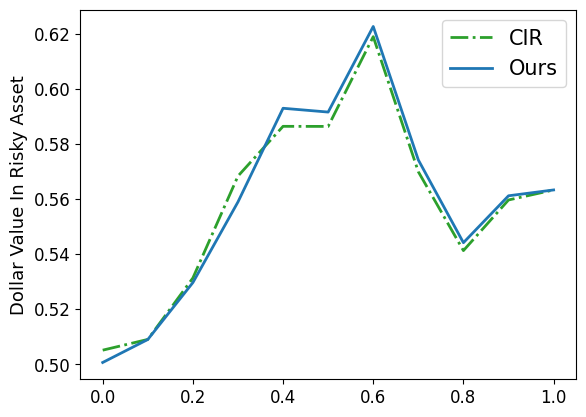}} 
    \subfigure[]{\includegraphics[width=0.4\textwidth]{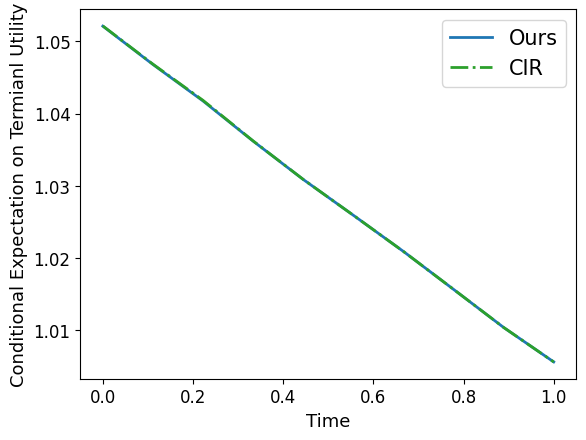}}
    \caption{[\textbf{Problem A}] (a) Equilibrium policies of an arbitrary SV process trajectory; (b) the conditional expectation on terminal utility corresponding to the equilibrium MV policies.}
    \label{fig:CIR_Numerical}
\end{figure}

\begin{figure}[h]
    \centering
    \subfigure[]{\includegraphics[width=0.4\textwidth]{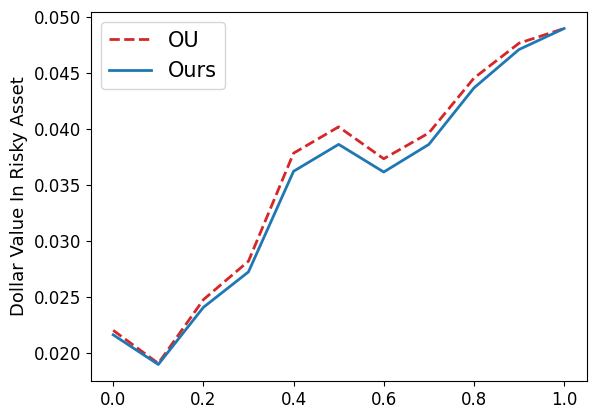}} 
    \subfigure[]{\includegraphics[width=0.4\textwidth]{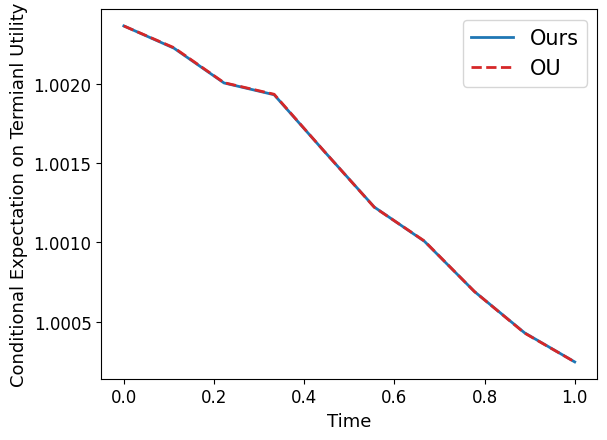}}
    \caption{[\textbf{Problem B}] (a) Equilibrium policies of an arbitrary SV process trajectory; (b) the conditional expectation on terminal utility corresponding to the equilibrium MV policies.}
    \label{fig:OU_Numerical}
\end{figure}

\subsubsection{CKLS-SV Model}
In the second set of experiments, it is assumed that the SV process $R_s$ in (\ref{SVMs}) follows a CKLS process with $p \neq 0.5$. The control policies of two agents are compared: one agent knows the true parameters for the CKLS model and solve for her equilibrium MV policy numerically by applying Algorithm \ref{alg:0} to \textbf{Problem C}, $\mathcal{P}_C(p)$, which consider a SV model whose dynamics is described by the following,
\begin{equation*}
\left\{
    \begin{array}{l}
     \displaystyle{\frac{dS_s}{S_s} = (0.03 + 0.0811 R_s^{\frac{1-2p}{2}})ds  + R_s^{-\frac{1}{2}}dB_s}, \\ 
     \displaystyle{dR_s = (9.4251-0.3374 R_s)ds + 0.6503 R_s^pdB^R_s},
    \end{array}
\right. 
\end{equation*}
with the same objective function as \textbf{Problem A} and \textbf{B}.
Meanwhile, the another agent believes the SV process $R_s$ follows a CIR process and tries to calibrate the corresponding CIR-SV model parameters according to the $M = 10000$ trajectories of the SV process observed (simulated), $\{R^i_t\}_{i = 1,\cdots, M}$, and adopts the analytical equilibrium policy with the calibrated parameters. That said, the later agent obtains her MV policy by computing the analytical policy given in \cite{Basak2010} to a problem with SV model dynamics described as follows,
\begin{equation*}
\left\{
    \begin{array}{l}
      \displaystyle{\frac{dS_s}{S_s} = (0.03 + 0.0811)ds  + R_s^{-\frac{1}{2}}dB_s}, \\ 
     \displaystyle{dR_s = (\hat{a}^{\{R^i_t\}}-\hat{b}^{\{R^i_t\}} R_s)ds + \hat{\sigma}^{\{R^i_t\}} R_s^pdB^R_s},
    \end{array}
\right. 
\end{equation*}
where $\hat{a}^{\{R^i_t\}}, \hat{b}^{\{R^i_t\}},$ and $ \hat{\sigma}^{\{R^i_t\}}$ are ordinary least squares estimator of CIR model parameters according to the observed trajectories, $\{R^i_t\}_{i = 1,\cdots, M}$.

For $p \in \{0.1, 0.3, \cdots, 0.9\}$, the equilibrium control policies of these two agents are shown in Figure \ref{fig:CKLS_Numerical} (a). It is observed that the equilibrium policies computed using explicit form given in \cite{Basak2010} with calibrated CIR parameters do not differ significantly from each other across various $p$, while the equilibrium policies obtained from Algorithm \ref{alg:0} are distinguishable from each other. This suggests that the parameter calibration of CIR model is not sufficiently responsive to the change of $p$ value and cannot adjust the equilibrium policy to the correct position. As a result, the estimated expected values of terminal utility are significantly lowered when $p \neq 0.5$ as illustrated in Figure \ref{fig:CKLS_Numerical} (b). Hence, modeling with calibrated CIR model is not accurate if the underlying model in the real market is a CKLS model with $p \neq 0.5$. A CKLS model is needed in general and this paper provides the first theoretical framework which can handle nonlocal BSDE arises from dynamic MV asset allocation problem under a general stochastic factor models.

\begin{figure}
    \centering
    \subfigure[]{\includegraphics[width=0.8\linewidth]{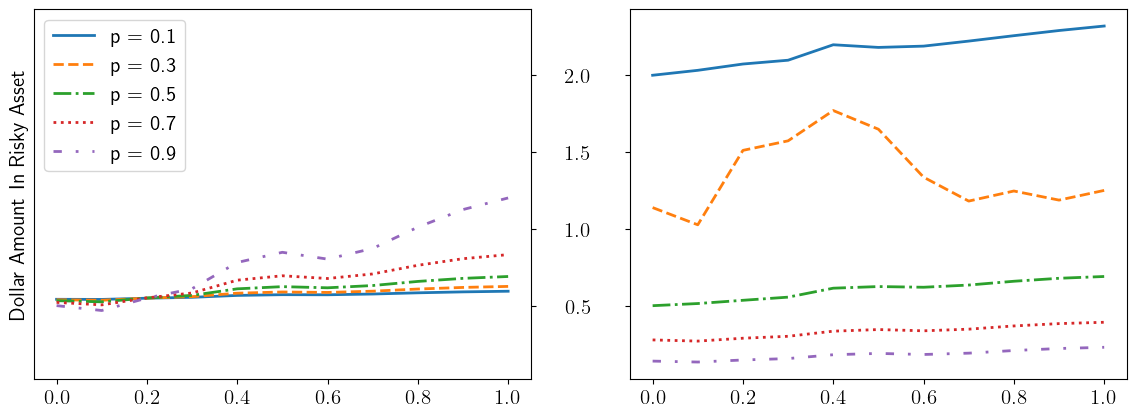}} 
    \subfigure[]{\includegraphics[width=0.8\linewidth]{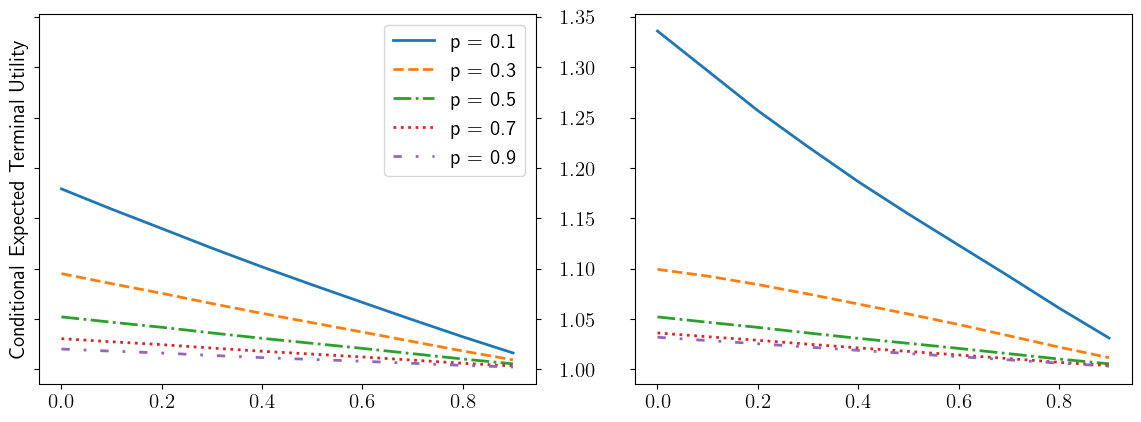}} 
    \caption{[\textbf{Problem C}] (a) Equilibrium control policy by explicit CIR-model-based formula with calibrated CIR parameters (left) and by numerical solution with CKLS-SV model (right), and (b) the corresponding conditional expected terminal utility across the horizon.}
    \label{fig:CKLS_Numerical}
\end{figure}

\subsubsection{Exponential Discounting Factor}

In the two sets of numerical experiments above, the objective function does not take discounting factor into account. To study the impact of discounting factor on hedging demand, the \textbf{Problem B} is studied again with the following objective function 
\begin{equation}\label{objectiveDiscount}
    J(s,w;u)=\int^T_s \eta(s, \tau) \mathbb{MV}^{\mathcal{F}_s}[W^u_\tau]d\tau + \mu(s,T)\mathbb{MV}^{\mathcal{F}_s}[W^u_T],
\end{equation}
where $\eta(s,\tau) \coloneqq e^{-\lambda_{\text{coef}}(\tau-s)}$ and $\mu(s,T) \coloneqq e^{-\lambda_{\text{coef}}(T-s)}$ for $\lambda_{\text{coef}} \in \{0.2, 0.5, 0.8\}$.
By applying Algorithm \ref{alg:0}, the equilibrium MV policies are obtained, denoted by $u_{\text{dist}}$, and they are compared with the equilibrium MV policy, $u_{B}$, which solves \textbf{Problem B} objective (\ref{numericalObj_w/oDiscount}). Recall that the equilibrium MV policy can be broken down into a myopic term and a hedging terms (see \eqref{EPincom} or \eqref{CKLSpolicy}) and the myopic term is independent of both discounting factor $\lambda(t,s)$ and BSDE solution $Y(t,s)$ and $Z(t,s)$. Hence, the differences in equilibrium MV policies in Figure \ref{fig:running_diff} reveal the differences in their hedging demands. The relative differences in hedging demand are computed as $\frac{u_{B}(s, R_s) - u_{\text{dist}}(s, R_s)}{\textbf{Hdeging}(u_{B}(s, R_s))}$ at each time point $s$ and plotted in Figure \ref{fig:running_diff} (b) after being averaged over $10000$ trajectories of simulated SV process $\{R^i_s\}_{i=1}^{10000}$.

\begin{figure}[!ht]
    \centering
    \subfigure[]{\includegraphics[width=0.4\textwidth]{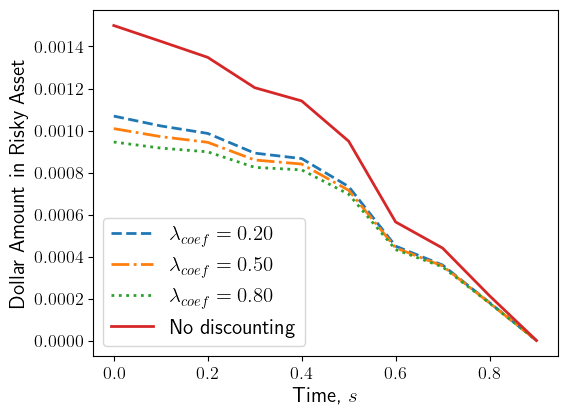}} 
    \subfigure[]{\includegraphics[width=0.4\textwidth]{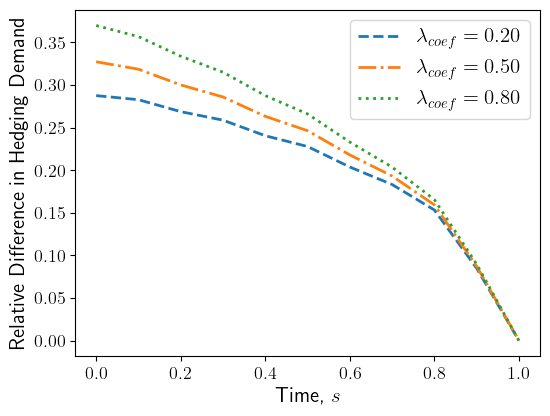}} 
    \caption{(a) A trajectory of hedging demand for various discounting factors, (b)average relative differences in hedging demands between objectives with exponential discounting factor (\ref{objectiveDiscount}) and objective without discounting (\ref{numericalObj_w/oDiscount}).}
    \label{fig:running_diff}
\end{figure}

It can be observed in Figure \ref{fig:running_diff} that the difference decreases to $0$ as time approaches the maturity $T$, which is in line with the intuition that when getting closer to the maturity, uncertainty vanishes and there is less need to hedge. 
In addition, it is also observed that smaller $\lambda_{\text{coef}}$ in the discounting factor leads to smaller average difference. This is reasonable for two reason. On one hand, if we divide the objective (\ref{objectiveDiscount}) by the positive weight function $\mu(s,T)$, the objective function becomes
\begin{equation*}
    \frac{J(s,w;u)}{\mu(s, T)}=\int^T_s \frac{\eta(s, \tau)}{\mu(s, T)} \mathbb{MV}^{\mathcal{F}_s}[W^u_\tau]d\tau + \mathbb{MV}^{\mathcal{F}_s}[W^u_T].
\end{equation*}
Since $\lim_{\lambda_{\text{coef}}\rightarrow-\infty} \frac{\eta(s, \tau)}{\mu(s, T)} = \lim_{\lambda_{\text{coef}}\rightarrow-\infty}e^{-\lambda_{\text{coef}}(\tau - T)} \equiv 0$ and $\mu(s,T) > 0 $ for all $s \in [0, T]$, the objective (\ref{objectiveDiscount}) is the equivalent to the objective without discounting (\ref{numericalObj_w/oDiscount}) when $\lambda_{\text{coef}}$ tends to negative infinity. On the other hand, when $\lambda_{\text{coef}} > 0$, smaller $\lambda_{\text{coef}}$ implies the future values are discounted at a slower rate and less weight is put on MV criterion in the near future as compared to the weight on terminal MV criterion, meaning the agent is longsighted and facing greater uncertainty. Thus, it is reasonable for the equilibrium MV policy hedges for a larger amount. When $\lambda_{\text{coef}} \leq 0$, $\eta(s, \tau)$ and $\mu(s,T)$ are no longer discounting factor (as they are greater than one), but it is still the case that the more negative $\lambda_{\text{coef}}$ is, the agent is more longsighted and hedges for a larger amount.

Through these three sets of experiments, our numerical algorithm is shown to be feasible in terms of its computational efficiency. Its accuracy is also validated with the analytical solution in the literature when the stochastic factor process is modeled by the CIR process or the OU process. Meanwhile, the numerical results illustrate the irreplaceability of CLKS-SV model and the impact of discounting factor on hedging demand.





\section{Concluding Remarks} \label{sec:cons}
In conclusion, we employ a game-theoretic approach to dynamic MV asset allocation in a general incomplete-market economy for investors with time-varying risk-return trade-off and for the objectives with general discounting factors in the MV criteria throughout the investment horizon. We provided a straightforward and fully analytical characterization of the equilibrium MV policy in a continuous-time incomplete market. Our probabilistic analysis has significant advantages. It establishes the well-definedness of the equilibrium MV policy without requiring explicit conditional expectations or moment-generating functions, enabling a much broader exploration of stochastic-volatility models. This approach stays effective even when conditional expectations are not differentiable, significantly relaxing constraints from previous studies. Additionally, we examine a new nonlocal BSDE system, proving its well-posedness through tailored Banach spaces and fixed-point theory, demonstrating that this framework can be applied to a wide range of financial scenarios. The system is also accompanied by an algorithm for numerical solutions.
 

Beyond the main body of this paper, we discuss several appropriate extensions inspired by \cite{Basak2010}: (1) Discrete-time \& multiple-stock formulation: We can examine the MV asset allocation problem within a discrete-time framework that encompasses multiple stocks and state variables. (2) Stochastic interest rates \& local volatility model: It is possible to incorporate stochastic interest rates and refine our setting to a complete market using the constant elasticity of variance (CEV) model for stock prices. (3) Additionally, by following \cite{Bjoerk2014a}, we aim to explore an intriguing topic: the dynamic MV portfolio problem with state-dependent risk aversion in both complete and incomplete market settings. Through this exploratory discussion, we seek to highlight the potential for further investigation and development within our MV analytical framework and nonlocal BSDE analysis.


\subsection*{\textbf{State-dependent MV problems}}
In our previous analysis of initial-time-dependent MV problems in complete and incomplete markets, one can find that the equilibrium policies \eqref{EPComplete} and \eqref{EPincom} are both independent of the current wealth $W_s$. This aspect somewhat needs improvement in practical financial applications. Inspired by \cite{Bjoerk2014a}, we consider a MV model with state-dependent risk aversion. Let us introduce the Markowitz MV operator defined as follows:  
\begin{equation*}
    \mathbb{SV}^{\mathcal{F}_s}_{w}[\cdot]:=\rho(s)\mathbb{E}^{\mathcal{F}_s}[\cdot]-\frac{1}{w}\mathbb{V}^{\mathcal{F}_s}[\cdot], 
\end{equation*}
This operator not only considers both returns and risks but also has the capability to dynamically adjust the relative importance of risk and return based on the timing $s$ and wealth amount $w$ of investor decision-making. Next, we consider 
\begin{equation*}
    J(s,w,r;u)=\int^T_s\eta(s,\tau)\mathbb{SV}^{\mathcal{F}_s}_{w}[W_\tau]d\tau+\mu(s,T)\mathbb{SV}^{\mathcal{F}_s}_{w}[W_T], 
\end{equation*}
where $\eta(s,\tau)$ and $\mu(s,T)$ represent the general nonexponential discounting. In a similar argument, one has 
\begin{equation*}
    \begin{split}
        J(s,w,r;u)=&\mathbb{E}_{s}\left[\int^T_s \lambda(s,\tau)\overline{\Phi}(s,w,W_\tau)d\tau+\overline{\Phi}(s,w,W_T)\right]\\
        & +\int^T_s \lambda(s,\tau)\overline{\Psi}\left(\mathbb{E}_{s}[W_\tau]\right)d\tau+\overline{\Psi}\left(\mathbb{E}_{s}[W_T]\right),         
    \end{split}
\end{equation*}
where $\overline{\Psi}(s,W,w)=\rho(s)Ww-\frac{\gamma}{2}w^2$, and $\Psi(w)=\frac{\gamma}{2}w^2$. Compared with the previous analysis of MV problems with the argument $w$ separable, it is more convenient to directly remove the terms $\mathbb{A}^aV(s,y)$, $\mathbb{A}^af(s,s,y,y)$, $H(s,s,y,\psi(y))$, $\int^T_s\mathbb{A}^aHd\tau$, and $\mathbb{A}^aG$ in the $V$-equation of \eqref{HJBSys} according to \eqref{V}. Consequently, in this case, the extended HJB system \eqref{V} reads  
\begin{equation} \label{xHJBSys}
\left\{
    \begin{array}{l}
        \lambda(s,s)\rho(s)w^2-\frac{\gamma}{2}\lambda(s,s)w^2+\displaystyle{\sup\limits_{a\in\mathcal{U}}\Big\{\mathbb{A}^a f(t,s,W,w,r)|_{t=s,W=w}} \\
        \displaystyle{\qquad +\gamma\int^T_s\lambda(s,\tau)g^\tau(s,w,r)\cdot\mathbb{A}^a g^\tau(s,w,r) d\tau +\gamma g^T(s,w,r)\cdot\mathbb{A}^a g^T(s,w,r) \Big\}=0, } \\


        \displaystyle{ \mathbb{A}^{\widehat{u}} f(t,s,W,w,r)+\lambda(t,s)\rho(t)Ww-\frac{\lambda(t,s)\gamma}{2}w^2 = 0,} \\

        ~ \\

        \displaystyle{ \mathbb{A}^{\widehat{u}} g^\tau(s,w,r) = 0,} 

        
    \end{array}
\right. 
\end{equation} 
with the terminal conditions: $V(T,w,r)=\rho(T)w^2$, $f(t,T,W,w,r)=\rho(t)Ww-\frac{\gamma}{2}w^2$, and $g^\tau(\tau,w,r)=w$. Inspired by the terminal conditions of $f$ and $g^\tau$, it is natural to make the Ansatz: 
\begin{equation} \label{xAnsatz}
\left\{
    \begin{array}{l}

        f(t,s,W,w,r)=c(t,s,r)Ww-\frac{\gamma}{2}b(t,s,r)w^2, \\

        ~ \\

        g^\tau(s,w,r)=a^\tau(s,r)w,  
    \end{array}
\right. 
\end{equation} 
Then \eqref{xHJBSys} and \eqref{xAnsatz} give us the equilibrium policy 
\begin{equation} \label{xEquiPolicy}
    \begin{split}
        \widehat{u}(s,w,r)=&\frac{1}{\gamma\sigma^2(s,r)}\frac{1}{b(s,s,r)}\Bigg\{\beta(s,r) c(s,s,r)-\gamma\beta(s,r) b(s,s,r) \\
        &   +\varrho (n\sigma)(s,r) c_r(s,s,r)-\gamma \varrho (n\sigma)(s,r) b_r(s,s,r)) \\
        &   +\gamma\int^T_s\lambda(s,\tau)\Big[\beta(s,r) (a^\tau)^2(s,r)+\varrho (n\sigma)(s,r) (a^\tau a^\tau_r)(s,r)\Big]d\tau \\
        &   + \gamma \Big[\beta(s,r) (a^T)^2(s,r)+\varrho (n\sigma)(s,r) (a^Ta^T_r)(s,r)\Big]\Bigg\}w \\ 
        & := \phi(s,r)w. 
    \end{split}
\end{equation}
It is worth noting that the dependence of \eqref{xEquiPolicy} upon the current wealth $w$ makes it significantly different from the previous ones \eqref{EPComplete} and \eqref{EPincom}. By substituting \eqref{xEquiPolicy} into \eqref{Completedynamics}, the wealth dynamics of $W_s$ looks like a geometric Brownian motion. Furthermore, by computing the conditional expectations of $W_\tau$ and $W^2_\tau$ controlled by the process $\phi(s,R_s)W_s$, 
the probabilistic interpretations \eqref{ProbInter} of $f$ and $g^\tau$ indicate  
\begin{equation*} 
\left\{
    \begin{array}{l}
        \begin{aligned}
        c(t,s,r)=& \int^T_s\lambda(t,\tau)\rho(t)a^\tau(s,r)d\tau+\rho(t)a^T(s,r) \\
        =&\int^T_s\lambda(t,\tau)\rho(t) \widetilde{Y}^\tau_sd\tau+\rho(t) \widetilde{Y}^T_s,            
        \end{aligned}
        \\

        ~ \\ 
        \begin{aligned}
        b(t,s,r) =& \int^T_s\lambda(t,\tau)\mathbb{E}_{s,r}\left[e^{2\int^\tau_s(r_0+\beta_\epsilon\phi_\epsilon-\frac{1}{2}\sigma^2_\epsilon\phi^2_\epsilon)d\epsilon+2\int^\tau_s\sigma_\epsilon\phi_\epsilon dB_\epsilon}\right]d\tau\\
        &+\mathbb{E}_{s,r}\left[e^{2\int^T_s(r_0+\beta_\epsilon\phi_\epsilon-\frac{1}{2}\sigma^2_\epsilon\phi^2_\epsilon])d\epsilon+2\int^T_s\sigma_\epsilon\phi_\epsilon dB_\epsilon}\right]\\
                    =& \int^T_s\lambda(t,\tau)Y^\tau_sd\tau+Y^T_s,      
        \end{aligned}
        ~ \\
        ~ \\

        \displaystyle{a^\tau(s,r)=\mathbb{E}_{s,r}\left[e^{\int^\tau_s(r_0+\beta_\epsilon\phi_\epsilon-\frac{1}{2}\sigma^2_\epsilon\phi^2_\epsilon)d\epsilon+\int^\tau_s\sigma_\epsilon\phi_\epsilon dB_\epsilon}\right]=\widetilde{Y}^\tau_s,} 
    \end{array}
\right. 
\end{equation*} 
where $Y^\tau_s$ and $\widetilde{Y}^\tau_s$ are the conditional expectations of Doléans-Dade exponentials of semimartingales $X_\tau=\int^\tau_s(2r_0+2\beta_\epsilon\phi_\epsilon+\sigma^2_\epsilon\phi^2_\epsilon)d\epsilon+\int^\tau_s2\sigma_\epsilon\phi_\epsilon dB_\epsilon$ and $\widetilde{X}_\tau=\int^\tau_s(r_0+\beta_\epsilon\phi_\epsilon) d\epsilon+\int^\tau_s\sigma_\epsilon\phi_\epsilon dB_\epsilon$, respectively. Consequently, under some suitable assumptions, one has 
\begin{equation}  \label{xFBSDE}
\left\{
    \begin{array}{l}
        \displaystyle{dR_{\epsilon}=m(\epsilon,R_\epsilon) d\epsilon+n(\epsilon,R_\epsilon) dB^R_\epsilon,} \\

        ~ \\
        \displaystyle{dY^\tau_\epsilon=-Y^\tau_\epsilon\Big[2r_0+2\beta(\epsilon,R_\epsilon)\phi(\epsilon,R_\epsilon)+\sigma^2(\epsilon,R_\epsilon)\phi^2(\epsilon,R_\epsilon)\Big]d\epsilon+Z^\tau_\epsilon dB^R_\epsilon,} \\
        
        ~ \\
        \displaystyle{d\widetilde{Y}^\tau_\epsilon=-\widetilde{Y}^\tau_\epsilon\Big[r_0+\beta(\epsilon,R_\epsilon)\phi(\epsilon,R_\epsilon)\Big] d\epsilon+\widetilde{Z}^\tau_\epsilon dB^R_\epsilon,}
    \end{array}
\right. 
\end{equation} 
with $(R_s,Y^\tau_\tau,\widetilde{Y}^\tau_\tau)=(r, 1, 1)$. It is worth noting that the FBSDE system \eqref{xFBSDE} is coupled with each other via the term $\phi(s,R_s)$ which is defined by replacing $a^\tau(s,r)$ by $\widetilde{Y}^\tau_s$, $(na^\tau_r)(s,r)$ by $\widetilde{Z}^\tau_s$, $b(t,s,r)$ by $\int^T_s\lambda(t,\tau)Y^\tau_sd\tau+Y^T_s$, $(nb_r)(t,s,r)$ by $\int^T_s\lambda(t,\tau)Z^\tau_sd\tau+Z^T_s$, $c(t,s,r)$ by $\int^T_s\lambda(t,\tau)\rho(t)\widetilde{Y}^\tau_sd\tau+\rho(t)\widetilde{Y}^T_s$, and $n(s,r)c_r(t,s,r)$ by $\int^T_s\lambda(t,\tau)\rho(t)\widetilde{Z}^\tau_sd\tau+\rho(t)\widetilde{Z}^T_s$ in \eqref{xEquiPolicy}. 

The generator of this BSDE \eqref{xFBSDE} is highly complex, as indicated by \eqref{xEquiPolicy}, and it does not meet the (stochastic) Lipschitz conditions that we examined in earlier sections. Therefore, we currently cannot prove its well-posedness. This will be left for our future research work. However, we can investigate a degenerate form of \eqref{xFBSDE}, specifically the case without randomness. For the special case, the FBSDE system \eqref{xFBSDE} will reduce to a system of ODEs. Furthermore, if $\rho(s)=1$ and $\lambda(t,s)=0$, it coincides with the ODE system of Proposition 4.5 in \cite{Bjoerk2014a}. By noting that $a_r=b_r=0$ and $Z_\cdot=\widetilde{Z}_\cdot=0$ in deterministic investment opportunities, $Y^\tau_s=e^{\int^\tau_s[2r_0+2\beta(\epsilon)\phi(\epsilon)+\sigma^2(\epsilon)\phi^2(\epsilon)]d\epsilon}$, and $\widetilde{Y}^\tau_s=e^{\int^\tau_s[r_0+\beta(\epsilon)\phi(\epsilon)] d\epsilon}$, the ODE system reduced from \eqref{xFBSDE} is equivalent to a deterministic integral equation 
\begin{equation} \label{SM integral eq}
    \begin{split}
        &\displaystyle{\phi(s)=\frac{\beta(s)}{\gamma\sigma^2(s)}\frac{\int^T_s\lambda(s,\tau)\rho(s)e^{\int^\tau_s[r_0+\beta(\epsilon)\phi(\epsilon)]d\epsilon}d\tau+\rho(s)e^{\int^T_s[r_0+\beta(\epsilon)\phi(\epsilon)]d\epsilon}}{\int^T_s\lambda(s,\tau)e^{2\int^\tau_s\left[r_0+\beta(\epsilon)\phi(\epsilon)+\frac{1}{2}\sigma^2(\epsilon)\phi^2(\epsilon)\right]d\epsilon}d\tau+e^{2\int^T_s\left[r_0+\beta(\epsilon)\phi(\epsilon)+\frac{1}{2}\sigma^2(\epsilon)\phi^2(\epsilon)\right]d\epsilon}}} \\
        & ~ \\         
        & \qquad\qquad  
        \displaystyle{+\frac{\beta(s)}{\sigma^2(s)}\frac{\int^T_s\lambda(s,\tau)e^{2\int^\tau_s[r_0+\beta(\epsilon)\phi(\epsilon)]d\epsilon}d\tau+e^{2\int^T_s[r_0+\beta(\epsilon)\phi(\epsilon)]d\epsilon}}{\int^T_s\lambda(s,\tau)e^{2\int^\tau_s\left[r_0+\beta(\epsilon)\phi(\epsilon)+\frac{1}{2}\sigma^2(\epsilon)\phi^2(\epsilon)\right]d\epsilon}d\tau+e^{2\int^T_s\left[r_0+\beta(\epsilon)\phi(\epsilon)+\frac{1}{2}\sigma^2(\epsilon)\phi^2(\epsilon)\right]d\epsilon}}-\frac{\beta(s)}{\sigma^2(s)}}
    \end{split}
\end{equation}
which reduces to the integral equation of Theorem 4.6 in \cite{Bjoerk2014a} if $\lambda(t,s)=0$, $\rho(s)=1$, and both $\beta(s)$ and $\sigma(s)$ are both independent of $s$. By a similar argument in \cite{Bjoerk2014a}, one can prove that \eqref{SM integral eq} admits a unique solution in $C([0,T])$. 

Conversion of numerical scheme similar to Algorithm \ref{alg:0} is observed in our numerical experiments and the solution is stable with respect to the forward process as shown in the following figures (\ref{fig:stateDpt}).

\begin{figure}[!ht]
    \centering
    \subfigure[]{\includegraphics[width=0.4\textwidth]{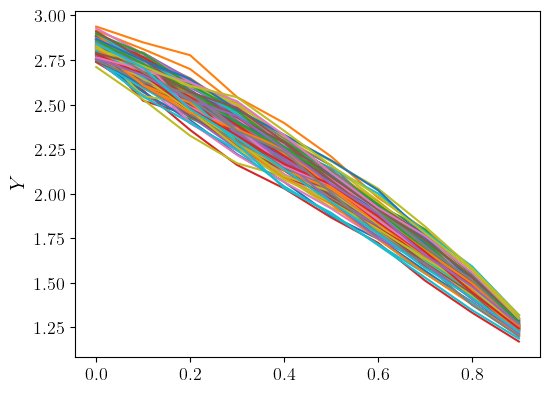}} 
    \subfigure[]{\includegraphics[width=0.4\textwidth]{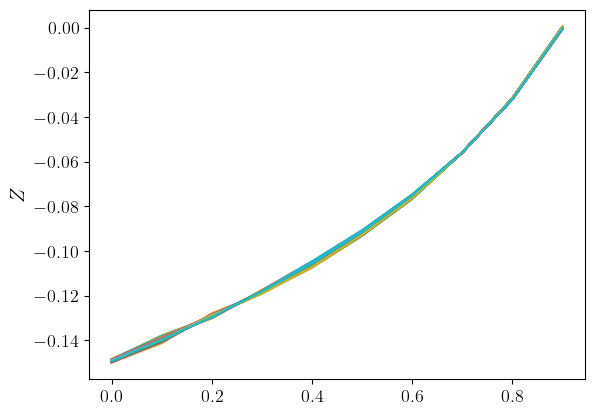}} 
    \subfigure[]{\includegraphics[width=0.4\textwidth]{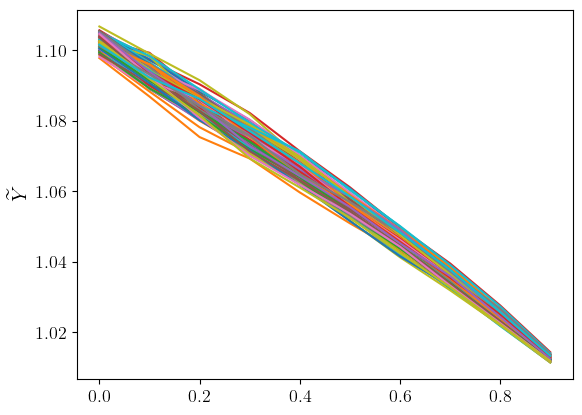}} 
    \subfigure[]{\includegraphics[width=0.4\textwidth]{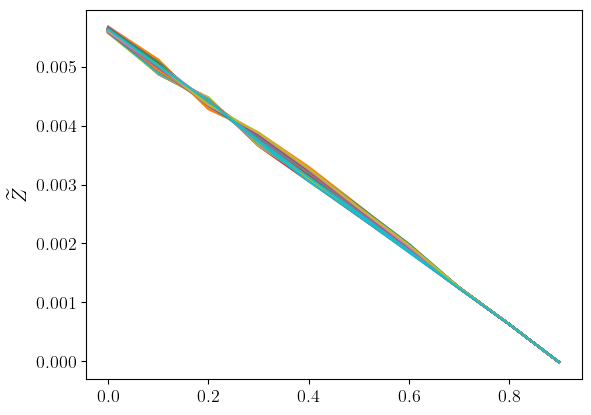}} 
    \caption{Plots of $Y, \widetilde{Y},Z$ and $\widetilde{Z}$ for 100 arbitrarily generated trajectories of stochastic volatility process, $\{R_t^i\}_{i=1}^{100}$}
    \label{fig:stateDpt}
\end{figure}







\bibliographystyle{spmpsci}      

\appendix
\normalsize
\section{Proofs of Statements} \label{app:pf}
\begin{proof}[Proof of Lemma \ref{LemmaVolterra} \& Theorem \ref{EPCom}]
From the nonlocal ODE system \eqref{NonlocalODESys}, it is clear that $a^\tau(s)=1$ and $L_t(t,s)=\int^T_s\lambda_t(t,s)ds$. Consequently, the coupled system \eqref{NonlocalODESys} is solvable if we can solve the following system  
 \begin{equation} \label{SimpleODEs}
\left\{
    \begin{array}{l}
        \displaystyle{A_s(s)-M_t(t,s)|_{t=s}+\lambda(s,s)\rho(s)-\gamma\int^T_s\lambda_t(t,\tau)|_{t=s}b^\tau(s)d\tau = 0,} \\



        ~ \\

        \displaystyle{(M_t)_s(t,s)-\frac{\beta^2(s)}{\sigma^2(s)}\frac{\int^T_s\lambda_t(t,\tau)d\tau}{\int^T_s\lambda(s,\tau)d\tau+1}A(s)+\lambda_t(t,s)\rho(t)+\lambda(t,s)\rho_t(t) = 0,} \\





        ~ \\

        \displaystyle{b^\tau_s(s)+\frac{1}{\gamma}\frac{\beta^2(s)}{\sigma^2(s)}\frac{1}{\int^T_s \lambda(s,\tau)d\tau+1}A(s) = 0,} \\




    \end{array}
\right. 
\end{equation} 
with $A(T)=\rho(T)$, $M_t(t,T)=\rho_t(t)$, and $b(\tau;\tau)=0$. Next, by substituting 
\begin{equation*} 
\left\{
    \begin{array}{l}
        \displaystyle{M_t(s,s)=\rho_t(s)-\int^T_s\frac{\beta^2(\eta)}{\sigma^2(\eta)}\left(\frac{\int^T_\eta\lambda_t(s,\tau)d\tau}{\int^T_\eta\lambda(\eta,\tau)d\tau+1}A(\eta)\right)d\eta+\int^T_s\lambda_t(s,\eta)\rho(s)d\eta+\int^T_s\lambda(s,\eta)\rho_t(s)d\eta}, \\

        ~ \\

        \displaystyle{b^\tau(s)=\frac{1}{\gamma}\int^\tau_s\frac{\beta^2(\eta)}{\sigma^2(\eta)}\frac{1}{\int^T_\eta \lambda(\eta,\tau)d\tau+1}A(\eta)d\eta}
    \end{array}
\right. 
\end{equation*} 
into the equation of $A(s)$ in \eqref{SimpleODEs}, one obtains 
\begin{equation*}
    \begin{split}
        &\displaystyle{A(s)=\rho(T)+\int^T_s\left[\int^T_\delta\frac{\beta^2(\eta)}{\sigma^2(\eta)}\left(\frac{\int^T_\eta\lambda_t(\delta,\tau)d\tau}{\int^T_\eta\lambda(\eta,\tau)d\tau+1}-\int^\eta_\delta\frac{\lambda_t(\delta,\tau)}{\int^T_\eta \lambda(\eta,\tau)d\tau+1}d\tau\right)A(\eta)d\eta\right.} \\
        &\qquad\qquad\qquad\qquad\qquad\qquad 
        \displaystyle{\left.-\rho_t(\delta)-\int^T_\delta\lambda_t(\delta,\eta)\rho(\delta)d\eta-\int^T_\delta\lambda(\delta,\eta)\rho_t(\delta)d\eta+\lambda(\delta,\delta)\rho(\delta)\right]d\delta}
    \end{split}
\end{equation*}
Finally, by exchanging the integration order of $\eta$ and $\delta$, it is straightforward to derive the linear Volterra integral equation \eqref{Volterrainteq} of the second kind. Thanks to the classical theory of Volterra integral equations, there exists a unique resolvent kernel $R$ associated with $K$ of \eqref{Volterrainteq} such that $A(s)=\Theta(s)+\int^T_sR(s,\tau)\Theta(\tau)d\tau$. Moreover, it is easy to verify that $A(s)=\rho(s)\big(\int^T_s\lambda(s,\tau)d\tau+1\big)$ is the unique solution of \eqref{SimpleODEs} and \eqref{Volterrainteq}. 
\end{proof}

\begin{proof}[Proof of Lemma \ref{WellposednessSimpleBSDE}]
According to Figure \ref{fig:nBSDE}, it is obvious that the elements $(Y^\tau_s,Z^\tau_s)_{s\in[0,\tau]}$ in the family of processes $\{(Y^\tau,Z^\tau)\}_{\tau\in[0,T]}$ parameterized by $\tau$ interact with each other. It is required that the mappings $[0,T]\ni\tau\mapsto(Y^\tau_s,Z^\tau_s)_{s\in[0,\tau]}\in\{(Y^\tau_\cdot,Z^\tau_\cdot)\}_{\tau\in[0,T]}$ and $[0,\tau]\ni s\mapsto(Y^\tau_s,Z^\tau_s)$ have a certain degree of regularity in the direction of $\tau$ and $s$, respectively. 

\vspace{0.3cm}
    
\noindent\textbf{(Regularities along the temporal direction of $\bm{s}$)} For any fixed $0\leq \tau\leq T<\infty$, we first prove \eqref{SimpleSDEInt} admits a unique solution $(Y^\tau_s,Z^\tau_s)_{s\in[0,\tau]}\in S^2([0,\tau];\mathbb{R}^d)\times M^2([0,\tau];\mathbb{R}^{d\times k})$. Due to the stochastic Lipschitz condition \eqref{LipschitzCons}, one has
\begin{equation*}
    \begin{aligned} 
& \mathbb{E}\left[\left(\int_0^\tau\left|h^\tau\left(s,\bar{Y}^\tau_s,\int^T_s \phi(s,\tau)\bar{Y}^\tau_sd\tau,\bar{Y}^T_s,\bar{Z}^\tau_s,\int^T_s\varphi(s,\tau)\bar{Z}^\tau_sd\tau,\bar{Z}^T_s\right)\right| \mathrm{d} s\right)^2\right] \\ 
\leq & C\left(\mathbb{E}\left[\left(\int_0^\tau|h^\tau(s,\bm{0})| \mathrm{d} s\right)^2\right]+\mathbb{E}\left[\left(\int_0^\tau u_s\left|\bar{Y}^\tau_s\right| \mathrm{d} s\right)^2\right]+c\mathbb{E}\left[\left(\int_0^\tau u_s\int^T_s\left|\bar{Y}^\tau_s\right|d\tau \mathrm{d} s\right)^2\right]\right.\\
& \qquad\qquad
+\mathbb{E}\left[\left(\int_0^\tau u_s\left|\bar{Y}^T_s\right| \mathrm{d} s\right)^2\right]
+\mathbb{E}\left[\left(\int_0^\tau v_s\left|\bar{Z}^\tau_s\right| \mathrm{d} s\right)^2\right]+c\mathbb{E}\left[\left(\int_0^\tau v_s\int^T_s\left|\bar{Z}^\tau_s\right|d\tau \mathrm{d} s\right)^2\right]\\
& \qquad\qquad
\left.+\mathbb{E}\left[\left(\int_0^\tau v_s\left|\bar{Z}^T_s\right| \mathrm{d} s\right)^2\right]\right)\\
\leq & C\left(\mathbb{E}\left[\left(\int_0^T|h^\tau(s,\bm{0})| \mathrm{d} s\right)^2\right]+\mathbb{E}\left[\left(\int_0^\tau u_s\left|\bar{Y}^\tau_s\right| \mathrm{d} s\right)^2\right]+c\mathbb{E}\left[\left(\int_0^T u_s\int^T_s\left|\bar{Y}^\tau_s\right|d\tau \mathrm{d} s\right)^2\right]\right.\\
& \qquad\qquad
+\mathbb{E}\left[\left(\int_0^T u_s\left|\bar{Y}^T_s\right| \mathrm{d} s\right)^2\right]+\mathbb{E}\left[\left(\int_0^\tau v_s\left|\bar{Z}^\tau_s\right| \mathrm{d} s\right)^2\right]+c\mathbb{E}\left[\left(\int_0^T v_s\int^T_s\left|\bar{Z}^\tau_s\right|d\tau \mathrm{d} s\right)^2\right]\\
& \qquad \qquad \left.+\mathbb{E}\left[\left(\int_0^T v_s\left|\bar{Z}^T_s\right| \mathrm{d} s\right)^2\right]\right) \\
\leq & C \mathbb{E}\left[\left(\int_0^T e^{\frac{\beta}{2} \int_0^s a_r \mathrm{~d} r}|h^\tau(s,\bm{0})| \mathrm{d} s\right)^2\right] +C \mathbb{E}\left[\left(\int_0^\tau \sqrt{a_s} e^{-\frac{\beta}{2} \int_0^s a_r \mathrm{~d} r} \sqrt{a_s}\left|\bar{Y}^\tau_s\right| e^{\frac{\beta}{2} \int_0^s a_r \mathrm{~d} r} \mathrm{~d} s\right)^2\right] \\
& +C \mathbb{E}\left[\left(\int_0^T \sqrt{a_s} e^{-\frac{\beta}{2} \int_0^s a_r \mathrm{~d} r} \int^T_s\sqrt{a_s}\left|\bar{Y}^\tau_s\right| e^{\frac{\beta}{2} \int_0^s a_r \mathrm{~d} r}d\tau \mathrm{~d} s\right)^2\right]\\
& +C \mathbb{E}\left[\left(\int_0^T \sqrt{a_s} e^{-\frac{\beta}{2} \int_0^s a_r \mathrm{~d} r} \sqrt{a_s}\left|\bar{Y}^T_s\right| e^{\frac{\beta}{2} \int_0^s a_r \mathrm{~d} r} \mathrm{~d} s\right)^2\right] \\
& +C \mathbb{E}\left[\left(\int_0^\tau v_s e^{-\frac{\beta}{2} \int_0^s a_r \mathrm{~d} r}\left|\bar{Z}^\tau_s\right| e^{\frac{\beta}{2} \int_0^s a_r \mathrm{~d} r} \mathrm{~d} s\right)^2\right]+C \mathbb{E}\left[\left(\int_0^T v_s e^{-\frac{\beta}{2} \int_0^s a_r \mathrm{~d} r}\int^T_s\left|\bar{Z}^\tau_s\right| e^{\frac{\beta}{2} \int_0^s a_r \mathrm{~d} r} d\tau \mathrm{~d} s\right)^2\right] \\ 
& +C \mathbb{E}\left[\left(\int_0^T v_s e^{-\frac{\beta}{2} \int_0^s a_r \mathrm{~d} r}\left|\bar{Z}^T_s\right| e^{\frac{\beta}{2} \int_0^s a_r \mathrm{~d} r} \mathrm{~d} s\right)^2\right],  
\end{aligned} 
\end{equation*}
where $C$ represents a generic constant which could be different from line to line, and $\max\big\{\sup_{0\leq s\leq \tau\leq T}|\phi(s,\tau)|,\sup_{0\leq s\leq \tau\leq T}|\varphi(s,\tau)|\big\}\leq c$. Thanks to the H\"{o}lder's inequality and Jensen's inequality, one has 
\begin{equation} \label{estimgenerator}
    \begin{aligned} 
& \mathbb{E}\left[\left(\int_0^\tau\left|h^\tau\left(s,\bar{Y}^\tau_s,\int^T_s \phi(s,\tau)\bar{Y}^\tau_sd\tau,\bar{Y}^T_s,\bar{Z}^\tau_s,\int^T_s\varphi(s,\tau)\bar{Z}^\tau_sd\tau,\bar{Z}^T_s\right)\right| \mathrm{d} s\right)^2\right] \\  
\leq & C\mathbb{E}\left[\left(\int_0^T e^{\frac{\beta}{2} \int_0^s a_r \mathrm{~d} r}|h^\tau(s,\bm{0})| \mathrm{d} s\right)^2\right] +C \mathbb{E}\left[\int_0^\tau a_s e^{-\beta \int_0^s a_r \mathrm{~d} r} ds \int^\tau_0 a_s\left|\bar{Y}^\tau_s\right|^2 e^{\beta \int_0^s a_r \mathrm{~d} r} \mathrm{~d} s\right] \\
& +C\mathbb{E}\left[\int_0^T a_s e^{-\beta \int_0^s a_r \mathrm{~d} r} ds \int^T_0 (T-s)\left(\int^T_s a_s\left|\bar{Y}^\tau_s\right|^2 e^{\beta \int_0^s a_r \mathrm{~d} r}d\tau\right)ds\right]\\
&+C \mathbb{E}\left[\int_0^T a_s e^{-\beta \int_0^s a_r \mathrm{~d} r}ds \int^T_0 a_s\left|\bar{Y}^T_s\right|^2 e^{\int_0^s a_r \mathrm{~d} r} \mathrm{~d} s\right] \\
& +C \mathbb{E}\left[\int_0^\tau a_s e^{-\beta \int_0^s a_r \mathrm{~d} r} ds \int^\tau_0 \left|\bar{Z}^\tau_s\right|^2 e^{\beta \int_0^s a_r \mathrm{~d} r} \mathrm{~d} s\right]\\
&+C\mathbb{E}\left[\int_0^T a_s e^{-\beta \int_0^s a_r \mathrm{~d} r} ds \int^T_0 (T-s)\left(\int^T_s \left|\bar{Z}^\tau_s\right|^2 e^{\beta \int_0^s a_r \mathrm{~d} r}d\tau\right)ds\right] \\ 
& +C \mathbb{E}\left[\int_0^T a_s e^{-\beta \int_0^s a_r \mathrm{~d} r}ds \int^T_0 \left|\bar{Z}^T_s\right|^2 e^{\int_0^s a_r \mathrm{~d} r} \mathrm{~d} s\right]. \end{aligned} 
\end{equation}

Next, let us first consider one of terms of \eqref{estimgenerator}
\begin{equation*}
    \begin{split}
        & \mathbb{E}\left[\int_0^T a_s e^{-\beta \int_0^s a_r \mathrm{~d} r} ds \int^T_0 (T-s)\left(\int^T_s a_s\left|\bar{Y}^\tau_s\right|^2 e^{\beta \int_0^s a_r \mathrm{~d} r}d\tau\right)ds\right]\\
         \leq &\frac{1}{\beta}\mathbb{E}\left[\int^T_0 (T-s)\left(\int^T_s a_s\left|\bar{Y}^\tau_s\right|^2 e^{\beta \int_0^s a_r \mathrm{~d} r}d\tau\right)ds\right] \\
         = &\frac{1}{\beta}\mathbb{E} \left[\int^T_0\left(\int^\tau_0(T-s)a_s\left|\bar{Y}^\tau_s\right|^2 e^{\beta \int_0^s a_r \mathrm{~d} r}ds\right)d\tau\right]\leq \frac{T}{\beta}\int^T_0\left(\mathbb{E}\left[\int^\tau_0 a_s\left|\bar{Y}^\tau_s\right|^2 e^{\beta \int_0^s a_r \mathrm{~d} r}ds\right]\right)d\tau \\
         =&\frac{T}{\beta}\int^T_0\lVert \sqrt{a_\cdot}\bar{Y}^\tau_\cdot\rVert^2_\beta d\tau \leq \frac{T^2}{\beta}\sup\limits_{\tau\in[0,T]}\lVert \sqrt{a_\cdot}\bar{Y}^\tau_\cdot\rVert^2_\beta<\infty
    \end{split}
\end{equation*}
Remarkably, $\int^T_0\lVert \sqrt{a_\cdot}\bar{Y}^\tau_\cdot\rVert^2_\beta d\tau$ is integrable since the mapping $\tau\mapsto\lVert \sqrt{a_\cdot}\bar{Y}^\tau_\cdot\rVert^2_\beta$ is continuous and $\sup_{\tau\in[0,T]}\sqrt{a_\cdot}\bar{Y}^\tau_\cdot\rVert^2_\beta<\infty$. A similar argument is suitable for others of \eqref{estimgenerator}. As a result, the process $M^\tau_{s}:=\mathbb{E}[\xi^\tau+\int_{0}^{\tau}\,h^\tau_s\mathrm{d}s|{\mathcal{F}}_{s}]$ for $s\in[0,\tau]$ is a square-integrable martingale. Thanks to the martingale representation theorem, there exists a unique process $(Z^\tau_s)_{s\in[0,\tau]}$ in $M^2([0,\tau];\mathbb{R}^{d\times k})$ such that 
\begin{equation*}
    M^\tau_{s}=M^\tau_{0}+\int_{0}^{s}Z^\tau_{r}dB_{r},\quad s\in[0,\tau].
\end{equation*}
Subsequently, let 
\begin{equation*}
Y^\tau_{s}:=\mathbb{E}\left[\left.\xi^\tau+\int_{s}^{\tau}h^\tau_sds\right|\mathcal{F}_s\right],\quad s\in[0,\tau].
\end{equation*}
It is clear that the constructed $(Y^\tau_s,Z^\tau_s)_{s\in[0,\tau]}\in S^2([0,\tau];\mathbb{R}^d)\times M^2([0,\tau];\mathbb{R}^{d\times k})$ solves \eqref{SimpleSDEInt}. Moreover, due to the linearity of the simple nonlocal BSDE \eqref{NonlocalBSDEInt}, it can be shown that the solution is unique for the solutions satisfying $Z^\tau_s\in M^2([0,\tau];\mathbb{R}^{d\times k})$. 

Next, we will show that the solution $(Y^\tau_s,Z^\tau_s)_{s\in[0,\tau]}$ belongs to the subspace $\big(\overline{L}^{2}(\beta,a_\cdot,[0,\tau];\mathbb{R}^{d})\cap L^{2,c}(\beta,a_\cdot,[0,\tau];\mathbb{R}^{d})\big)\times L^{2}(\beta,a_\cdot,[0,\tau];\mathbb{R}^{d\times k})$. First of all, one has 
\begin{equation*}
    \begin{aligned} & \left|Y^\tau_t\right| e^{\frac{\beta}{2} \int_0^t a_r \mathrm{~d} r}\\
    =&\left|\mathbb{E}\left[\left.\xi^\tau+\int_{t}^\tau h^\tau\left(s,\bar{Y}^\tau_s,\int^T_s \phi(s,\tau)\bar{Y}^\tau_sd\tau,\bar{Y}^T_s,\bar{Z}^\tau_s,\int^T_s\varphi(s,\tau)\bar{Z}^\tau_sd\tau,\bar{Z}^T_s\right) \mathrm{d} s \right| \mathcal{F}_t\right]\right| e^{\frac{\beta}{2} \int_0^t a_r \mathrm{~d} r} 
\\  \leq &\mathbb{E}\left[\left.|\xi^\tau| e^{\frac{\beta}{2} \int_0^\tau a_r \mathrm{~d} r} \right\rvert\, \mathcal{F}_t\right]+\mathbb{E}\left[\left.\int_{t}^\tau|h^\tau(s,\bm{0})| e^{\frac{\beta}{2} \int_0^s a_r \mathrm{~d} r} \mathrm{~d} s \right\rvert\, \mathcal{F}_t\right] \\
& \qquad 
+\mathbb{E}\left[\int_t^\tau u_s\left|\bar{Y}^\tau_s\right| e^{\frac{\beta}{2} \int_0^t a_r \mathrm{~d} r} \mathrm{d} s\right]+c\mathbb{E}\left[\int_t^\tau u_s\int^T_s\left|\bar{Y}^\tau_s\right|e^{\frac{\beta}{2} \int_0^t a_r \mathrm{~d} r} d\tau \mathrm{d} s\right]+\mathbb{E}\left[\int_t^\tau u_s\left|\bar{Y}^T_s\right| e^{\frac{\beta}{2} \int_0^t a_r \mathrm{~d} r} \mathrm{d} s\right] \\
& \qquad 
+\mathbb{E}\left[\int_t^\tau v_s\left|\bar{Z}^\tau_s\right| e^{\frac{\beta}{2} \int_0^t a_r \mathrm{~d} r} \mathrm{d} s\right]+c\mathbb{E}\left[\int_t^\tau v_s\int^T_s\left|\bar{Z}^\tau_s\right| e^{\frac{\beta}{2} \int_0^t a_r \mathrm{~d} r} d\tau \mathrm{d} s\right]+\mathbb{E}\left[\int_t^\tau v_s\left|\bar{Z}^T_s\right| e^{\frac{\beta}{2} \int_0^t a_r \mathrm{~d} r} \mathrm{d} s\right]  
\end{aligned}
\end{equation*}

By the H\"{o}lder's inequality, one obtains 
\begin{equation} \label{estimY}
    \begin{aligned} & \left|Y^\tau_t\right| e^{\frac{\beta}{2} \int_0^t a_r \mathrm{~d} r}=\left|\mathbb{E}\left[\left.\xi^\tau+\int_{t}^\tau h^\tau\left(s,\bar{Y}^\tau_s,\int^T_s \phi(s,\tau)\bar{Y}^\tau_sd\tau,\bar{Y}^T_s,\bar{Z}^\tau_s,\int^T_s\varphi(s,\tau)\bar{Z}^\tau_sd\tau,\bar{Z}^T_s\right) \mathrm{d} s \right| \mathcal{F}_t\right]\right| e^{\frac{\beta}{2} \int_0^t a_r \mathrm{~d} r} \\
& \leq \mathbb{E}\left[\left.|\xi^\tau| e^{\frac{\beta}{2} \int_0^\tau a_r \mathrm{~d} r} \right\rvert\, \mathcal{F}_t\right]+\mathbb{E}\left[\left.\int_{t}^\tau|h^\tau(s,\bm{0})| e^{\frac{\beta}{2} \int_0^s a_r \mathrm{~d} r} \mathrm{~d} s \right\rvert\, \mathcal{F}_t\right] \\
& \qquad\qquad 
+\mathbb{E}\left[\left.\left(\int_t^\tau a_s\left|\bar{Y}^\tau_s\right|^2 e^{\beta \int_0^s a_r \mathrm{~d} r} \mathrm{d} s\right)^\frac{1}{2} \cdot \left(\int_t^\tau a_s e^{-\beta \int_0^s a_r \mathrm{~d} r} \mathrm{d} s \cdot e^{\beta \int_0^t a_r \mathrm{~d} r}\right)^\frac{1}{2}\right\rvert\ \mathcal{F}_t\right] \\
& \qquad\qquad 
+c\mathbb{E}\left[\left.\left(\int_t^\tau (T-s) \int^T_s a_s\left|\bar{Y}^\tau_s\right|^2e^{\beta \int_0^s a_r \mathrm{~d} r} d\tau \mathrm{d} s\right)^\frac{1}{2}\cdot\left(\int_t^\tau a_s e^{-\beta \int_0^s a_r \mathrm{~d} r} \mathrm{d} s \cdot e^{\beta \int_0^t a_r \mathrm{~d} r}\right)^\frac{1}{2}\right\rvert\ \mathcal{F}_t\right] \\
& \qquad\qquad 
+\mathbb{E}\left[\left.\left(\int_t^\tau a_s\left|\bar{Y}^T_s\right|^2 e^{\beta \int_0^s a_r \mathrm{~d} r} \mathrm{d} s\right)^\frac{1}{2} \cdot \left(\int_t^\tau a_s e^{-\beta \int_0^s a_r \mathrm{~d} r} \mathrm{d} s \cdot e^{\beta \int_0^t a_r \mathrm{~d} r}\right)^\frac{1}{2}\right\rvert\ \mathcal{F}_t\right] \\
& \qquad\qquad  
+\mathbb{E}\left[\left.\left(\int_t^\tau \left|\bar{Z}^\tau_s\right|^2 e^{\beta \int_0^s a_r \mathrm{~d} r} \mathrm{d} s\right)^\frac{1}{2} \cdot \left(\int_t^\tau a_s e^{-\beta \int_0^s a_r \mathrm{~d} r} \mathrm{d} s \cdot e^{\beta \int_0^t a_r \mathrm{~d} r}\right)^\frac{1}{2}\right\rvert\ \mathcal{F}_t\right] \\
& \qquad\qquad 
+c\mathbb{E}\left[\left.\left(\int_t^\tau (T-s) \int^T_s \left|\bar{Z}^\tau_s\right|^2e^{\beta \int_0^s a_r \mathrm{~d} r} d\tau \mathrm{d} s\right)^\frac{1}{2}\cdot\left(\int_t^\tau a_s e^{-\beta \int_0^s a_r \mathrm{~d} r} \mathrm{d} s \cdot e^{\beta \int_0^t a_r \mathrm{~d} r}\right)^\frac{1}{2}\right\rvert\ \mathcal{F}_t\right] \\
& \qquad\qquad 
+\mathbb{E}\left[\left.\left(\int_t^\tau \left|\bar{Z}^T_s\right|^2 e^{\beta \int_0^s a_r \mathrm{~d} r} \mathrm{d} s\right)^\frac{1}{2} \cdot \left(\int_t^\tau a_s e^{-\beta \int_0^s a_r \mathrm{~d} r} \mathrm{d} s \cdot e^{\beta \int_0^t a_r \mathrm{~d} r}\right)^\frac{1}{2}\right\rvert\ \mathcal{F}_t\right]  
\end{aligned}
\end{equation}

Moreover, it holds that  
\begin{equation*}
    \begin{split}
        & \mathbb{E}\left[\left(\int_t^\tau (T-s) \int^T_s a_s\left|\bar{Y}^\tau_s\right|^2e^{\beta \int_0^s a_r \mathrm{~d} r} d\tau \mathrm{d} s\right)^\frac{1}{2}\cdot\left(\int_t^\tau a_s e^{-\beta \int_0^s a_r \mathrm{~d} r} \mathrm{d} s \cdot e^{\beta \int_0^t a_r \mathrm{~d} r}\right)^\frac{1}{2}\right]\\
        & \leq \mathbb{E}\left[\left(\frac{1}{\beta}\int_t^\tau (T-s) \int^T_s a_s\left|\bar{Y}^\tau_s\right|^2e^{\beta \int_0^s a_r \mathrm{~d} r} d\tau \mathrm{d} s\right)^\frac{1}{2}\right] \\
        & \leq \mathbb{E}\left[\left(\frac{1}{\beta}\int_0^T \left(\int^\tau_0 (T-s) a_s\left|\bar{Y}^\tau_s\right|^2e^{\beta \int_0^s a_r \mathrm{~d} r} ds\right) d\tau \right)^\frac{1}{2}\right] \\
        & \leq \left(\frac{T}{\beta}\int_0^T \mathbb{E}\left(\int^\tau_0  a_s\left|\bar{Y}^\tau_s\right|^2e^{\beta \int_0^s a_r \mathrm{~d} r} ds\right) d\tau \right)^\frac{1}{2} \\
        & \leq \left(\frac{T}{\beta}\int_0^T \lVert \sqrt{a_\cdot}\bar{Y}^\tau_\cdot\rVert^2_\beta d\tau \right)^\frac{1}{2} \\
        & \leq \left(\frac{T^2}{\beta}\sup\limits_{\tau\in[0,T]}\lVert \sqrt{a_\cdot}\bar{Y}^\tau_\cdot\rVert^2_\beta\right)^\frac{1}{2}
    \end{split}
\end{equation*}

From a similar argument for other terms of \eqref{estimY} along with the assumptions of terminal condition $\{\xi^\tau\}_{\tau\in[0,T]}$, generator $\{h_s\}_{s\in[0,T]}$, and $\{(\bar{Y}^\tau_\cdot,\bar{Z}^\tau_\cdot)\}_{\tau\in[0,T]}$, Doob’s martingale inequality, and Jensen’s inequality, it follows that 
\begin{equation} \label{ContinuityYtau}
    \|Y^\tau_\cdot\|_{\beta,c}^{2}=\mathbb{E}\left[\sup_{0\leq s\leq \tau}\left(e^{\beta\int_{0}^{s}a_{r}dr}|Y^\tau_{s}|^{2}\right)\right]<\infty.
\end{equation}

Next, we need to introduce the following stopping time: 
\begin{equation*}
    \tau_{n}=\operatorname*{inf}{\biggl\{}t\in[0,\tau]:\int_{0}^{t}e^{\beta\int_{0}^{s}a_rdr}|Z^\tau_{s}|^{2}ds \geq n{\biggr\}}\wedge\tau.
\end{equation*}

It follows Proposition 2.6 in \cite{li2023bsdes} that 
\begin{equation*}
    \begin{array}{l}{{\mathbb{E}\left[\int_{0}^{\tau_{n}}e^{\beta\int_{0}^{s}a_ r dr}|Z^\tau_{s}|^{2}\mathrm{d}s\right]+\beta\mathbb{E}\left[\int_{0}^{\tau_{n}}e^{\beta\int_{0}^{s}a_r\mathrm{d}r}a_{s}|Y^\tau_{s}|^{2}\mathrm{d}s\right]}}{{\mathrm{~}\leq\mathbb{E}\left[e^{\beta\int_{0}^{\tau_{n}}a_{r}\mathrm{d}r}|Y^\tau_{\tau_{n}}|^{2}\right]+2\mathbb{E}\left[\int_{0}^{\tau_{n}}e^{\beta\int_{0}^{s}a_r\mathrm{d}r}|h_s||Y^\tau_{s}|\mathrm{d}s\right].}}\end{array}
\end{equation*}
For $\varepsilon>0$, $2a b\leq{\textstyle{\frac{1}{\varepsilon}}}a^{2}+\varepsilon b^{2}$. Hence, for each $s\in[0,\tau]$, we have 
\begin{equation} \label{2ab}
\left\{
    \begin{array}{l}
        2u_{s}|\bar{Y}^\tau_{s}||Y^\tau_{s}|\leq{\frac{\beta}{2}}u_{s}|Y^\tau_{s}|^{2}+{\frac{2}{\beta}}u_{s}|\bar{Y}^\tau_{s}|^{2}\leq{\frac{\beta}{2}}u_{s}|Y^\tau_s|^{2}+{\frac{2}{\beta}}a_{s}|\bar{Y}^\tau_{s}|^{2}, \\

        ~ \\ 

        2v_{s}|\bar{Z}^\tau_{s}||Y^\tau_{s}|\leq\frac{\beta}{2}v_{s}^{2}|Y^\tau_{s}|^{2}+\frac{2}{\beta}|\bar{Z}^\tau_{s}|^{2}, \\

        ~ \\ 

        2cu_{s}|Y^\tau_{s}|\int^T_s\left|\bar{Y}^\tau_s\right|d\tau\leq{\frac{\beta}{2}}u_{s}|Y^\tau_{s}|^{2}+{\frac{2c^2}{\beta}}u_{s}\left(\int^T_s\left|\bar{Y}^\tau_s\right|d\tau\right)^{2}\leq{\frac{\beta}{2}}u_{s}|Y^\tau_{s}|^{2}+{\frac{2c^2}{\beta}}a_{s}\left(\int^T_s\left|\bar{Y}^\tau_s\right|d\tau\right)^{2},\\

        ~ \\ 

        2cv_{s}|Y^\tau_{s}|\int^T_s|\bar{Z}^\tau_{s}|d\tau\leq\frac{\beta}{2}v_{s}^{2}|Y^\tau_{s}|^{2}+\frac{2c^2}{\beta}\left(\int^T_s|\bar{Z}^\tau_{s}|d\tau\right)^2.
        
    \end{array}
\right. 
\end{equation} 
Moreover, the estimations of $2u_{s}|\bar{Y}^T_{s}||Y^\tau_{s}|$ and $2v_{s}|\bar{Z}^T_{s}||Y^\tau_{s}|$ are similar to ones of $2u_{s}|\bar{Y}^\tau_{s}||Y^\tau_{s}|$ and $2v_{s}|\bar{Z}^\tau_{s}||Y^\tau_{s}|$. Consequently, one has 
\begin{equation} \label{estimYZ}
    \begin{aligned} \mathbb{E} & {\left[\int_0^{\tau_n} e^{\beta \int_0^s a_r \mathrm{~d} r}\left|Z^\tau_s\right|^2 \mathrm{~d} s\right]+\frac{\beta}{2} \mathbb{E}\left[\int_0^{\tau_n} e^{\beta \int_0^s a_r \mathrm{~d} r} a_s\left|Y^\tau_s\right|^2 \mathrm{~d} s\right] } \\ \leq & \mathbb{E}\left[e^{\beta \int_0^{\tau_n} a_r \mathrm{~d} r}\left|Y^\tau_{\tau_n}\right|^2\right]+2 \mathbb{E}\left[\int_0^{\tau_n} e^{\beta \int_0^s a_r \mathrm{~d} r}|h^\tau(s,\bm{0})|\left|Y^\tau_s\right| \mathrm{d} s\right] \\ 
& +\frac{2}{\beta} \mathbb{E}\left[\int_0^{\tau_n} e^{\beta \int_0^s a_r \mathrm{~d} r}\left(a_s\left|\bar{Y}^\tau_s\right|^2+ca_s\left(\int^T_s\left|\bar{Y}^\tau_s\right|d\tau\right)^2+a_s\left|\bar{Y}^T_s\right|^2+\left|\bar{Z}^\tau_s\right|^2+c\left(\int^T_s\left|\bar{Z}^\tau_s\right|d\tau\right)^2+\left|\bar{Z}^T_s\right|^2\right) \mathrm{d} s\right] 
\\ \leq & \mathbb{E}\left[e^{\beta \int_0^{\tau_n} a_r \mathrm{~d} r}\left|Y^\tau_{\tau_n}\right|^2\right]+2 \mathbb{E}\left[\sup _{0 \leq t \leq \tau}\left(\left|Y^\tau_t\right| e^{\frac{\beta}{2} \int_0^t a_r \mathrm{~d} r}\right) \int_0^{\tau_n} e^{\frac{\beta}{2} \int_0^s a_r \mathrm{~d} r}|h^\tau(s,\bm{0})| \mathrm{d} s\right] \\ 
& +\frac{2}{\beta} \mathbb{E}\left[\int_0^{\tau_n} e^{\beta \int_0^s a_r \mathrm{~d} r}\left(a_s\left|\bar{Y}^\tau_s\right|^2+ca_s\left(\int^T_s\left|\bar{Y}^\tau_s\right|d\tau\right)^2+a_s\left|\bar{Y}^T_s\right|^2+\left|\bar{Z}^\tau_s\right|^2+c\left(\int^T_s\left|\bar{Z}^\tau_s\right|d\tau\right)^2+\left|\bar{Z}^T_s\right|^2\right) \mathrm{d} s\right] \\
\leq & 2\|Y^\tau_\cdot\|_{\beta,c}^{2}+K(c,T,\beta)\left(\sup\limits_{\tau\in[0,T]}\lVert \sqrt{a_\cdot}\bar{Y}^\tau_\cdot\rVert^2_\beta+\sup\limits_{\tau\in[0,T]}\lVert \bar{Z}^\tau_\cdot\rVert^2_\beta\right)+\mathbb{E}\left[\left(\int_0^T e^{\frac{\beta}{2} \int_0^s a_r \mathrm{~d} r}|h^\tau(s,\bm{0})| \mathrm{d} s\right)^2\right] <+\infty.
\end{aligned}
\end{equation}
By taking the limit with respect to $n$ (i.e. $ n\to\infty$) in the above inequality, the use of Levi’s lemma together
with \eqref{ContinuityYtau}, we have 
\begin{equation} \label{BoundYZtau}
    \lVert \sqrt{a_\cdot}Y^\tau_\cdot\rVert^2_\beta+\lVert Z^\tau_\cdot\rVert^2_\beta<\infty. 
\end{equation}

\vspace{0.3cm}

\noindent\textbf{(Regularities along the parameter direction of $\bm{\tau}$)} Until now, it has been proven that the solution $(Y^\tau_s,Z^\tau_s)_{s\in[0,\tau]}$ belongs to the subspace $\big(\overline{L}^{2}(\beta,a_\cdot,[0,\tau];\mathbb{R}^{d})\cap L^{2,c}(\beta,a_\cdot,[0,\tau];\mathbb{R}^{d})\big)\times L^{2}(\beta,a_\cdot,[0,\tau];\mathbb{R}^{d\times k})$ for any fixed $\tau\in[0,T]$. After studying the regularities of $\{(Y^\tau_\cdot,Z^\tau_\cdot)\}_{\tau\in[0,T]}$ along the temporal direction $s\in[0,\tau]$, we begin to examine its properties along the parameter direction $\tau\in[0,T]$ and prove that it belongs to $\mathbb{M}^c_T(\beta,a_\cdot)$. In addition to upper bounds of $\lVert Y^\tau_\cdot \rVert_{\beta,c}$, $\lVert \sqrt{a_\cdot}Y^\tau_\cdot\rVert_\beta$, and $\lVert Z^\tau_\cdot\rVert_\beta$ with respect to $\tau\in[0,T]$, we also need to show the continuity of the mapping $\tau\mapsto \lVert Y^\tau_\cdot \rVert_{\beta,c}$, $\tau\mapsto\lVert \sqrt{a_\cdot}Y^\tau_\cdot\rVert_\beta$, and $\tau\mapsto\lVert Z^\tau_\cdot\rVert_\beta$. By taking the supremum on left side of \eqref{estimY} and \eqref{estimYZ} over $\tau\in[0,T]$, it is obvious that 
$$\|\{(Y^\tau_\cdot,Z^\tau_\cdot)\}\|_{\beta,a_\cdot,c}^{2}:=\sup\limits_{\tau\in[0,T]}\left\{\|Y^\tau_\cdot\|_{\beta,a_\cdot,c,\tau}^{2}+\|{\sqrt{a}}.Y^\tau_\cdot\|_{\beta,a_\cdot,\tau}^{2}+\|Z^\tau_\cdot\|_{\beta,a_\cdot,\tau}^{2}\right\}<\infty$$

Next, let us first show the continuity of $\tau\mapsto\lVert Z^\tau_\cdot\rVert_\beta$. Suppose that $0\leq \tau\leq \tau^\prime\leq T$, 
\begin{equation} \label{ContinuityZ}
    \begin{split}
        & \left|\|Z^\tau_\cdot\|_{\beta,a_\cdot,\tau}-\|Z^{\tau^\prime}_\cdot\|_{\beta,a_\cdot,\tau^\prime}\right|=\left|\mathbb{E}\left[\int_{0}^{\tau}e^{\beta\int_{0}^{s}a_{r}dr}|Z^\tau_{s}|^{2}\mathrm{d}s\right]-\mathbb{E}\left[\int_{0}^{\tau^\prime}e^{\beta\int_{0}^{s}a_{r}dr}|Z^{\tau^\prime}_{s}|^{2}\mathrm{d}s\right]\right| \\
        & \leq \left|\mathbb{E}\left[\int_{0}^{\tau}e^{\beta\int_{0}^{s}a_{r}dr}|Z^\tau_{s}|^{2}\mathrm{d}s\right]-\mathbb{E}\left[\int_{0}^{\tau}e^{\beta\int_{0}^{s}a_{r}dr}|Z^{\tau^\prime}_{s}|^{2}\mathrm{d}s\right]\right|+\\
        & \qquad \left|\mathbb{E}\left[\int_{0}^{\tau}e^{\beta\int_{0}^{s}a_{r}dr}|Z^{\tau^\prime}_{s}|^{2}\mathrm{d}s\right]-\mathbb{E}\left[\int_{0}^{\tau^\prime}e^{\beta\int_{0}^{s}a_{r}dr}|Z^{\tau^\prime}_{s}|^{2}\mathrm{d}s\right]\right|
    \end{split}
\end{equation}

Since $\sup_{\tau\in[0,T]}\lVert Z^\tau_\cdot\rVert^2_\beta<\infty$ and $\int_{0}^{\tau}e^{\beta\int_{0}^{s}a_{r}dr}|Z^{\tau^\prime}_{s}|^{2}\mathrm{d}s$ is continuous in $\tau$, the Lebesgue's dominated convergence theorem promises that the second difference of \eqref{ContinuityZ} tends to zero as $\tau\to\tau^\prime$. Next, we investigate the first term of \eqref{ContinuityZ}. It is beneficial to consider two BSDEs with different terminal conditions $Y^{\tau^\prime,\xi^\prime}_\tau$ and $\xi^\tau$, 

\begin{equation} 
\left\{
    \begin{array}{l}
     \displaystyle{Y^{\tau^\prime}_t=Y^{\tau^\prime,\xi^\prime}_\tau+\int^\tau_t h^{\tau^\prime}\bigg(s,\bar{Y}^\tau_s,\int^T_s \phi(s,\sigma,\bar{Y}^\sigma_s)d\sigma,\bar{Y}^T_s,\bar{Z}^\tau_s,\int^T_s\varphi(s,\sigma,\bar{Z}^\sigma_s)d\sigma,\bar{Z}^T_s\bigg)ds-\int^\tau_s Z^{\tau^\prime}_sdB_s}, \\ 

     
     \displaystyle{Y^\tau_t=\xi^\tau+\int^\tau_t h^\tau\bigg(s,\bar{Y}^\tau_s,\int^T_s \phi(s,\sigma,\bar{Y}^\sigma_s)d\sigma,\bar{Y}^T_s,\bar{Z}^\tau_s,\int^T_s\varphi(s,\sigma,\bar{Z}^\sigma_s)d\sigma,\bar{Z}^T_s\bigg)ds-\int^\tau_s Z^\tau_sdB_s}, 
    \end{array}
\right. 
\end{equation} 
where $Y^{\tau^\prime,\xi^\prime}_\tau=\xi^{\tau^\prime}+\int^{\tau^\prime}_\tau h^{\tau^\prime}\bigg(s,\bar{Y}^\tau_s,\int^T_s \phi(s,\sigma,\bar{Y}^\sigma_s)d\sigma,\bar{Y}^T_s,\bar{Z}^\tau_s,\int^T_s\varphi(s,\sigma,\bar{Z}^\sigma_s)d\sigma,\bar{Z}^T_s\bigg)ds-\int^{\tau^\prime}_\tau Z^{\tau^\prime}_sdB_s$. Thanks to the continuous dependence of the solutions on the terminal value (Refer to Theorem 3.1 of \cite{li2023bsdes}), one has 
\begin{equation*}
    \begin{split}
        & \left|\mathbb{E}\left[\int_{0}^{\tau}e^{\beta\int_{0}^{s}a_{r}dr}|Z^\tau_{s}|^{2}\mathrm{d}s\right]-\mathbb{E}\left[\int_{0}^{\tau}e^{\beta\int_{0}^{s}a_{r}dr}|Z^{\tau^\prime}_{s}|^{2}\mathrm{d}s\right]\right|^2\leq \|Z^\tau_\cdot-Z^{\tau^\prime}_\cdot\|_{\beta,a_\cdot,\tau}^{2}\\
        & \leq K \|\xi^\tau-Y^{\tau^\prime,\xi^\prime}_\tau\|_{\beta,a_\cdot,\tau}^{2} + K\rho(|\tau-\tau^\prime|)\|\{(\bar{Y}^\tau_\cdot,\bar{Z}^\tau_\cdot)\}\|_{\beta,a_\cdot}^{2} \\
        & \leq K\left(\|\xi^\tau-\xi^{\tau^\prime}\|_{\beta,a_\cdot,\tau}^{2}+\|\xi^{\tau^\prime}-Y^{\tau^\prime,\xi^\prime}_\tau\|_{\beta,a_\cdot,\tau}^{2}\right)+ K\rho(|\tau-\tau^\prime|)\|\{(\bar{Y}^\tau_\cdot,\bar{Z}^\tau_\cdot)\}\|_{\beta,a_\cdot}^{2}.
    \end{split}
\end{equation*}

Considering that both $\{\xi^\tau\}_T$ and $(Y^{\tau^\prime}_s)_{s\in[0,\tau^\prime]}$ are continuous processes (in $\tau$ and in $s$, respectively) with $\|\xi^\tau\|_{c,\beta,a_\cdot,T}^{2}:=\mathbb{E}\left[\operatorname*{sup}_{\tau\in[0,T]}e^{\beta\int_{0}^{\tau}a_{r}dr}|\xi^\tau|^{2}\right]<\infty$ and $\|Y^{\tau^\prime}_\cdot\|_{\beta,a_\cdot,c,\tau^\prime}^{2}:=\mathbb{E}\left[\operatorname*{sup}_{s\in[0,\tau^\prime]}\left(e^{\beta\int_{0}^{s}a_{r}dr}|Y^{\tau^\prime}_{s}|^{2}\right)\right]<\infty$, the dominated convergence theorem ensures that the first difference of \eqref{ContinuityZ} also tends to zero as $\tau\to\tau^\prime$. Similarly, one can also show that the mapping $\tau\mapsto \lVert Y^\tau_\cdot \rVert_{\beta,c}$ and $\tau\mapsto\lVert \sqrt{a_\cdot}Y^\tau_\cdot\rVert_\beta$ are both continuous in $\tau\in[0,T]$. 

In conclusion, we prove that the output $\{(Y^\tau_\cdot,Z^\tau_\cdot)\}_{\tau\in[0,T]}$ of operator \eqref{SimpleSDEInt} belongs to $\mathbb{M}^c_T(\beta,a_\cdot)$ if the input $\{(\bar{Y}^\tau_\cdot,\bar{Z}^\tau_\cdot)\}_{\tau\in[0,T]}\in\mathbb{M}_T(\beta,a_\cdot)$. It completes the proof.      
\end{proof}

\begin{proof}[Proof of Theorem \ref{WellposednessBSDE}]
By taking advantage of Lemma \ref{WellposednessSimpleBSDE}, we can construct a mapping $\Lambda:\mathbb{M}_T(\beta,a_\cdot)\to\mathbb{M}_T(\beta,a_\cdot)$. Specially, $\Lambda\big(\{(\bar{Y}^\tau_\cdot,\bar{Z}^\tau_\cdot)\}_{\tau\in[0,T]}\big)=\{(Y^\tau_\cdot,Z^\tau_\cdot)\}_{\tau\in[0,T]}$ is defined by \eqref{SimpleSDEInt}, and seek the solution of \eqref{NonlocalBSDEDiff} as a fixed point of the operator $\Lambda$. Next, we suppose that $\{(\bar{Y}^{1,\tau}_\cdot,\bar{Z}^{1,\tau}_\cdot)\}_{\tau\in[0,T]}$ and $\{(\bar{Y}^{2,\tau}_\cdot,\bar{Z}^{2,\tau}_\cdot)\}_{\tau\in[0,T]}\in\mathbb{M}_T(\beta,a_\cdot)$. Furthermore, 
\begin{equation*}
\begin{aligned}
    \Lambda\big(\{(\bar{Y}^{1,\tau}_\cdot,\bar{Z}^{1,\tau}_\cdot)\}_{\tau\in[0,T]}\big)=\{(Y^{1,\tau}_\cdot,Z^{1,\tau}_\cdot)\}_{\tau\in[0,T]}, \\ \Lambda\big(\{(\bar{Y}^{2,\tau}_\cdot,\bar{Z}^{2,\tau}_\cdot)\}_{\tau\in[0,T]}\big)=\{(Y^{2,\tau}_\cdot,Z^{2,\tau}_\cdot)\}_{\tau\in[0,T]}.
\end{aligned} 
\end{equation*}

For convenience, we introduce some following notations:
\begin{equation*}
    \begin{array}{c}{{\Delta Y^\tau_{s}:=Y_{s}^{1,\tau}-Y_{s}^{2,\tau},\quad \Delta Z^\tau_{s}:=Z_{s}^{1,\tau}-Z_{s}^{2,\tau},
    \Delta \bar{Y}^\tau_{s}:=\bar{Y}_{s}^{1,\tau}-\bar{Y}_{s}^{2,\tau},\quad \Delta \bar{Z}^\tau_{s}:=\bar{Z}_{s}^{1,\tau}-\bar{Z}_{s}^{2,\tau}.}}\\ \Delta{h}^\tau_{s}:=h^\tau\left(s,\bar{Y}^{1,\tau}_s,\int^T_s \phi(s,\tau)\bar{Y}^{1,\tau}_sd\tau,\bar{Y}^{1,T}_s,\bar{Z}^{1,\tau}_s,\int^T_s\varphi(s,\tau)\bar{Z}^{1,\tau}_sd\tau,\bar{Z}^{1,T}_s\right) \\
    \qquad\qquad\qquad\qquad\qquad   
    -h^\tau\left(s,\bar{Y}^{2,\tau}_s,\int^T_s \phi(s,\tau)\bar{Y}^{2,\tau}_sd\tau,\bar{Y}^{2,T}_s,\bar{Z}^{2,\tau}_s,\int^T_s\varphi(s,\tau)\bar{Z}^{2,\tau}_sd\tau,\bar{Z}^{2,T}_s\right), \qquad 0\leq s\leq \tau\leq T. \end{array}
\end{equation*}

It is clear that $\{(\Delta Y^\tau_\cdot,\Delta Z^\tau_\cdot)\}_{\tau\in[0,T]}$ is the solution of 
\begin{equation*}
    \Delta Y^\tau_t=\int^\tau_t \Delta h^\tau_s ds-\int^\tau_t \Delta Z^\tau_sdB_s, \qquad 0\leq t\leq \tau\leq T. 
\end{equation*}

Similar to the argument of \eqref{2ab}, one has   
    \begin{equation*}
        \begin{aligned} 
        & \mathbb{E}\left[\int_0^\tau e^{\beta \int_0^s a_r \mathrm{~d} r}\left|\Delta Z^\tau_s\right|^2 \mathrm{~d} s\right]+\beta \mathbb{E}\left[\int_0^\tau e^{\beta \int_0^s a_r \mathrm{~d} r} a_s\left|\Delta Y^\tau_s\right|^2 \mathrm{~d} s\right]\\
        \leq & 2 \mathbb{E}\left[\int_0^\tau e^{\beta \int_0^s a_r \mathrm{~d} r}\left|\Delta h^\tau_s\right|\left|\Delta Y_s\right| \mathrm{d} s\right] \\ 
\leq & \frac{\beta}{2} \mathbb{E}\left[\int_0^\tau e^{\beta \int_0^s a_r \mathrm{~d} r} a_s\left|\Delta Y_s\right|^2 \mathrm{~d} s\right] \\
& \qquad 
+\frac{2}{\beta} \mathbb{E}\Biggl[\int_0^\tau e^{\beta \int_0^s a_r \mathrm{~d} r}\Biggl(a_s\left|\Delta\bar{Y}^\tau_s\right|^2+k^2a_s\left(\int^T_s\left|\Delta\bar{Y}^\tau_s\right|d\tau\right)^2+a_s\left|\Delta\bar{Y}^T_s\right|^2+\left|\Delta\bar{Z}^\tau_s\right|^2 \\
& \qquad +k^2\left(\int^T_s\left|\Delta\bar{Z}^\tau_s\right|d\tau\right)^2+\left|\Delta\bar{Z}^T_s\right|^2\Biggl) \mathrm{d} s\Biggl] \\
\leq & \frac{\beta}{2} \mathbb{E}\left[\int_0^\tau e^{\beta \int_0^s a_r \mathrm{~d} r} a_s\left|\Delta Y_s\right|^2 \mathrm{~d} s\right] +\frac{12\max\{c^2T^2,1\}}{\beta}\|\{(\Delta \bar{Y}^\tau_\cdot,\Delta \bar{Z}^\tau_\cdot)\}\|_{\beta,a_\cdot}^{2} 
\end{aligned}
\end{equation*}

Consequently, one has 
\begin{equation} \label{LamdaContraction}
        \begin{aligned} 
        &\mathbb{E}\left[\int_0^\tau e^{\beta \int_0^s a_r \mathrm{~d} r}\left|\Delta Z^\tau_s\right|^2 \mathrm{~d} s\right]+\mathbb{E}\left[\int_0^\tau e^{\beta \int_0^s a_r \mathrm{~d} r} a_s\left|\Delta Y^\tau_s\right|^2 \mathrm{~d} s\right]\\
        \leq &\left(\frac{12\max\{c^2T^2,1\}}{\beta}+\frac{24\max\{c^2T^2,1\}}{\beta^2}\right)\|\{(\Delta \bar{Y}^\tau_\cdot,\Delta \bar{Z}^\tau_\cdot)\}\|_{\beta,a_\cdot}^{2} 
\end{aligned}
\end{equation}
Furthermore, by taking the supremum over $\tau\in[0,T]$ on both sides of \eqref{LamdaContraction}, it is clear that the mapping $\Lambda$ is a contraction if $\frac{12\max\{c^2T^2,1\}}{\beta}+\frac{24\max\{c^2T^2,1\}}{\beta^2}<1$. It completes the proof. 
\end{proof}

\end{document}